\documentclass[10pt]{article}
\usepackage{hyperref}
\usepackage{amsmath,amsthm}
\usepackage{caption}
\usepackage{fullpage}
\usepackage{amsfonts}
\usepackage{changepage}
\usepackage{bookmark}
\usepackage{graphicx}
\usepackage[absolute]{textpos}
\usepackage{tikz}

\usetikzlibrary{shapes,snakes}
\def\boxit#1{\vbox{\hrule\hbox{\vrule\kern4pt
  \vbox{\kern1pt#1\kern1pt}
\kern2pt\vrule}\hrule}}
\usetikzlibrary{arrows}
\usetikzlibrary{decorations.markings}

\newtheorem{lemma}{Lemma}
\newtheorem{theorem}{Theorem}

\newtheorem{claim}{Claim}

\newcommand{\finaltreewidth}{24}

\def\dist(#1){\text{dist}(#1)}

\newcommand{\defproblem}[3]{\par
 \vspace{3mm}
\noindent\fbox{
 \begin{minipage}{0.96\textwidth}
 \begin{tabular*}{\textwidth}{@{\extracolsep{\fill}}lr} #1  \vspace{1mm} \\ \end{tabular*}
 {\textbf{Input:}} #2%
	\vspace{1mm}\\%
 {\textbf{Question:}} #3%
 \end{minipage}
 }
 \vspace{3mm}
\par
}

\begin{document}

\title{Hardness of Metric Dimension in Graphs of Constant Treewidth%
\thanks{This research is a part of a project that have received funding from the European Research Council (ERC) under the European Union's Horizon 2020 research and innovation programme
Grant Agreement 714704.}}

\author{Shaohua Li \thanks{Institute of Informatics, University of Warsaw, Poland, \texttt{Shaohua.Li@mimuw.edu.pl}.}
\and Marcin Pilipczuk \thanks{Institute of Informatics, University of Warsaw, Poland, \texttt{malcin@mimuw.edu.pl}.}}

 \date{}

\maketitle

\begin{textblock}{20}(0, 12.5)
\includegraphics[width=40px]{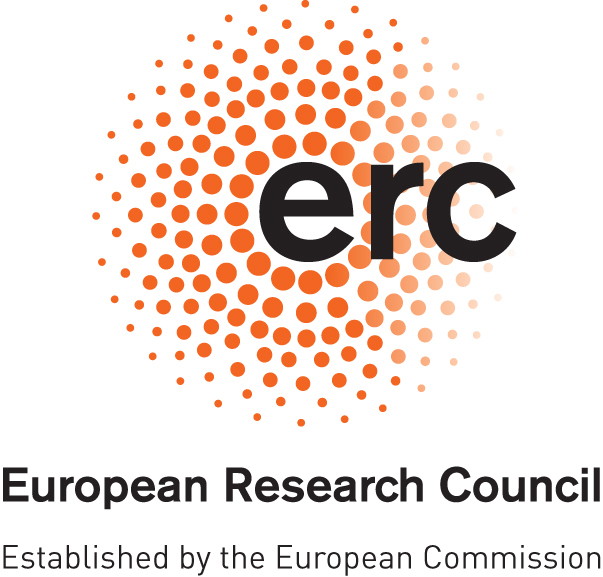}%
\end{textblock}
\begin{textblock}{20}(0, 13.4)
\includegraphics[width=40px]{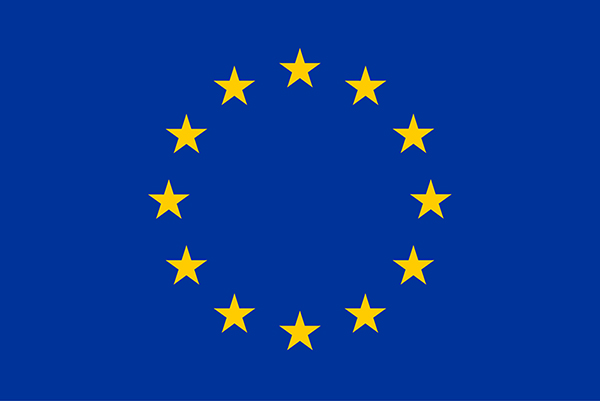}%
\end{textblock}

\begin{abstract}
The \textsc{Metric Dimension} problem asks for a minimum-sized \emph{resolving set} in a given (unweighted, undirected) graph $G$.
Here, a set $S \subseteq V(G)$ is \emph{resolving} if no two distinct vertices of $G$ have the same distance vector to $S$.
The complexity of \textsc{Metric Dimension} in graphs of bounded treewidth remained elusive in the past years.
Recently, Bonnet and Purohit~[IPEC 2019] showed that the problem is W[1]-hard under treewidth parameterization.
In this work, we strengthen their lower bound to show that \textsc{Metric Dimension} is NP-hard in graphs of treewidth $\finaltreewidth$.
\end{abstract}

\section{Introduction}
Let $G$ be an unweighted and undirected graph and let $S \subseteq V(G)$. For a vertex $v \in V(G)$, the \emph{distance vector}
from $v$ to $S$ is the assignment $S \ni w \mapsto \mathrm{dist}_G(v,w)$, where $\mathrm{dist}_G$ denotes the distance in the graph $G$.
The set $S$ is \emph{resolving} if a distance vector to $S$ uniquely determines the source vertex; that is, no two vertices of $G$ have the same distance vector to $S$.
The \textsc{Metric Dimension} problem asks for a resolving set of minimum possible size; such a set is sometimes called the \emph{metric basis} of $G$.
The decision version of \textsc{Metric Dimension} asks for a resolving set of size not exceeding a given threshold $k$.

\textsc{Metric Dimension} has already been introduced in 1970s~\cite{MDdef1,MDdef2}. Determining its computational complexity turned out to be quite challenging.
It is polynomial-time solvable on trees~\cite{MDdef1,MDdef2,KhullerRR96}, outerplanar graphs~\cite{DiazPSL17}, and chain graphs~\cite{FernauHHMS15}, but NP-hard for example on planar graphs~\cite{DiazPSL17} or split graphs~\cite{EpsteinLW15}.
From the parameterized complexity point of view, the FPT status of the \textsc{Metric Dimension} parameterized by the solution size has been open for a while and finally resolved
in negative by Hartung and Nichterlein~\cite{HartungN13}. In the search of a tractable structural parameterization, FPT algorithms has been shown for parameters: treelength plus maximum degree~\cite{BelmonteFGR17},
vertex cover number~\cite{HartungN13}, max leaf number~\cite{Eppstein15}, and modular-width~\cite{BelmonteFGR17}.

The above list misses probably the most important graph width measure, namely treewidth.
Determining the complexity of \textsc{Metric Dimension}, parameterized by treewidth, remained elusive in the last years, and has been asked a few times~\cite{BelmonteFGR17,DiazPSL17,Eppstein15}.
Bonnet and Purohit in a paper presented at IPEC 2019~\cite{BonnetP19} showed that the problem is W[1]-hard, even with pathwidth parameterization.
In this work we strengthed their result by proving para-NP-hardness of this parameterization.

\begin{theorem}\label{thm:main}
\textsc{Metric Dimension}, restricted to graphs of treewidth at most $\finaltreewidth$, is NP-hard.
\end{theorem}

Theorem~\ref{thm:main} brings us much closer to closing (unfortunately mostly in negative) the question of the complexity of \textsc{Metric Dimension}
in graphs of bounded treewidth. The remaining gap is to determine the exact treewidth value where the problem becomes hard: note that
it is open if \textsc{Metric Dimension} is polynomial-time solvable on graphs of treewidth $2$.

The proof of Theorem~\ref{thm:main} starts with a construction of a graph with  a separation of order $9$ over which \emph{a lot} of information on a partial solution to \textsc{Metric Dimension}
is transfered.
More formally, similarly as Bonnet and Purohit~\cite{BonnetP19}, we use the \textsc{Multicolored Resolving Set} problem as an auxiliary intermediate problem.
In this problem, the input graph is additionally equipped with an integer $k$, a tuple of $k$ disjoint vertex sets $X_1,X_2,\ldots,X_k$, and a set $\mathcal{P}$ of vertex pairs.
The goal is to choose a set $S$ consisting of exactly one vertex from each set $X_i$ so that for every $\{u,v\} \in \mathcal{P}$, the pair $\{u,v\}$ is resolved by $S$, that is,
$u$ and $v$ have different distance vectors to $S$.
In our construction, the sets $X_i$ are on one side of the said separation of order $9$, while the pairs $\mathcal{P}$ are on the second side. The crux of the construction
is to make every distance from a vertex of the separator to a chosen vertex of $S$ count: despite the fact that the separation has constant size, $S$ is of unbounded size, giving $\Omega(|S|)$ distances to work with.
Overall, the above gives a relatively clean reduction giving NP-hardness of \textsc{Multicolored Resolving Set} in graphs of constant treewidth, presented in Section~\ref{sec:MRS}.
This reduction is the main new insight and technical contribution of this paper.

Then, again similarly as in the work of Bonnet and Purohit~\cite{BonnetP19}, it takes a lot of effort (presented in Section~\ref{sec:MD})
  to turn the above reduction to \textsc{Multicolored Resolving Set} into a reduction to \textsc{Metric Dimension}. Here, there are two problems. First, one needs to introduce some gadgets to force the solution to take exactly one vertex from each set $X_i$. Second, one needs to ensure that the intended
solution resolves \emph{all} vertex pairs, not only the ones from $\mathcal{P}$.
For both problems, we borrow the tools from Bonnet and Purohit~\cite{BonnetP19}. In particular, the first problem is resolved by \emph{forced set gadgets} in a way very similar to~\cite{BonnetP19}.
The second problem is resolved by adding a number of new connections to the graph and \emph{forced vertex gadgets} of~\cite{BonnetP19}.
Thus, while the toolbox remains the same as in~\cite{BonnetP19}, the application is different as the graph we are working with is significantly different.
The construction is presented in Sections~\ref{ss:forced-set}-\ref{ss:forced-vertex}.

After applying all the modifications to obtain a \textsc{Metric Dimension} instance, one needs to check three aspects. First, one needs to ensure that the forced set gadgets work as intended,
forcing the solution to take one vertex from each $X_i$; this check is rather simple and very similar to the analogous check of~\cite{BonnetP19}.
Second, one needs to check that the introduced forced vertex gadgets, that contain extra vertices from the intended resolving set (apart from the ones in $X_i$s), do not accidentally
resolve any pair from $\mathcal{P}$. This check is not trivial, but still relatively simple. Note that the mentioned two properties already ensure one of the implications in the proof of the
correctness of the reduction: if the final \textsc{Metric Dimension} instance is a yes-instance, then the input instance of the source problem is a yes-instance.
These two checks are presented in Section~\ref{ss:easy-direction}.

Then one needs to check that every pair of vertices is resolved by an intended solution. Due to the complexity of the construction, this turned out to be very tedious
and spans essentially the second half of the volume of this paper (Section~\ref{ss:difficult-direction}).

Besides, we show that the treewidth of the constructed graph is bounded by a constant in Section~\ref{ss:treewidth}.

\section{Preliminaries}
In this paper, all graphs are undirected. In a graph $G$, let $V(G)$ be the set of vertices of $G$.
For a vertex $v\in V(G)$, we denote the open neighborhood and closed neighborhood of $v$ by $N_G(v)$ and $N_G[v]$ respectively (or just $N(v)$ and $N[v]$ if the graph is clear in the context).
For two vertices $u,v\in V(G)$, let $P(u,v)$ be a path from $u$ to $v$.
Since the graph is undirected, $P(u,v)$ and $P(v,u)$ denote the same path.
We denote the neighbor of $u$ on $P(u,v)$ by $N_u(u,v)$ (also the neighbor of $v$ on $P(u,v)$ by $N_v(u,v)$).
Similarly, if there is a path which is named as, for example, $P^{h}(i,j,x)$ such that $u$ is one endpoint of $P^{h}(i,j,x)$, we denote the neighbor of $u$ on $P^{h}(i,j,x)$ by $N^{h}_{u}(i,j,x)$.
For simplicity, we abuse the notation in the sense that for a path $P$, we use $P$ to denote the path or the vertex set of the path.
The meaning should be clear in the context.
We define the length of a path $P$ to be the number of edges on the path and denote it by $|P|$.
For two vertices $u,v\in V(G)$, we define the distance between $u$ and $v$ to be the length of any shortest path from $u$ to $v$, denoted by $\text{dist}_G(u,v)$.
Note that in this paper we use $\dist(u,v)$ to denote the distance between $u$ and $v$ mostly if the graph is clear in the context.
For a path $P$ of even length with two endpoints $u$ and $v$, let $w$ be the vertex on $P$ such that the length of the subpath of $P$ from $u$ to $w$ equals the length of the subpath of $P$ from $w$ to $v$.
Then we call $w$ the \emph{middle vertex} of $P$ and denote it by $\text{mid}(P)$.
We say that two distinct vertices $u,u'$ are \emph{false twins} if $N[u]=N[u']$.
Since a path decomposition is also a tree decomposition, the treewidth of a graph $G$ is at most the pathwidth of $G$.
In this paper, for convenience of the proof, we use the alternative characterization of pathwidth,
i.e. the pathwidth of a graph $G$ equals the node search number of $G$ minus one~\cite{kirousis1985interval}.

\section{Reduction from 3-Dimensional Matching to Multicolored Resolving Set}\label{sec:MRS}

Bonnet and Purohit introduced \textsc{$k$-Multicolored Resolving Set} as an intermediate problem in order to show the W[1]-hardness of \textsc{Metric Dimension} parameterized by treewidth~\cite{BonnetP19}.

\defproblem{\textsc{$k$-Multicolored Resolving Set}}
{An undirected graph $G=(V,E)$, an integer $k$, a set $\chi=\{X_1,...,X_k\}$ where $X_1,...,X_k$ are disjoint subsets of $V(G)$ and a set ${\cal P}=\{\{x_1,y_1\},...,\{x_h,y_h\}\}$ where $\{x_1,y_1\},...,\{x_h,y_h\}$ are vertex pairs of $G$. }
{Is there a set of $k$ vertices $S$ such that \\
(i) $|S\cap X_i|=1$ for every $i=1,...,k$, and \\
(ii) for every $\ell\in\{1,...,h\}$, there exists a vertex $v\in S$ such that dist($v$,$x_{\ell}$)$\neq$ dist($v$,$y_{\ell}$).}


We show that this problem is NP-hard on graphs of constant treewidth. We make a reduction from \textsc{$3$-dimensional matching}, which is well-known to be NP-hard~\cite{karp1972reducibility}.


\defproblem{\textsc{$3$-dimensional matching}}
{the universe $U=\{1,2,3\}\times [n]$ and a set ${\cal F}=\{A_1,...,A_m\}$ of tuples
such that for every $j\in [m]$, the tuple $A_j=\{(1,x),(2,y),(3,z)\}$ where $(1,x),(2,y),(3,z)\in U$. }
{are there $n$ tuples $A_{j_1},...,A_{j_n}$ such that $\bigcup\limits_{h=1}^{n} A_{j_h}=U$.}

Given an instance $(U,\cal F)$ of \textsc{$3$-dimensional matching} with the universe $U=\{1,2,3\}\times [n]$ and a set ${\cal F}$ of $m$ tuples $A_1,...,A_m\subseteq U$, we construct an instance $(G,n,\chi,\cal P)$ of \textsc{$n$-Multicolored Resolving Set} as follows. First, we create $m$ vertices $s_i^1,...,s_i^m$ as $X_i$ for each $i\in [n]$. Let $\chi=\{X_1,...,X_n\}$ and $X=\bigcup\limits_{i=1}^{n}X_i$.  Then we create $n$ vertex pairs $\{u_r^1,v_r^1\},...,\{u_r^n,v_r^n\}$ for each $r\in \{1,2,3\}$ and let ${\cal P}_r=\{\{u_r^i,v_r^i\}|i=1,...,n\}$. We create $3$ vertices $a_r,b_r,c_r$ and let $W_r=\{a_r,b_r,c_r\}$ for each $r\in\{1,2,3\}$. Let ${\cal P}={\cal P}_1\cup {\cal P}_2\cup {\cal P}_3$ and $W=W_1\cup W_2\cup W_3$. Finally, let $M=40(n+1)$. For each tuple $A_j=\{(1,x),(2,y),(3,z)\}$ ($j\in[m],x,y,z\in [n]$) of the given instance and each integer $i\in [n]$, we link $s_i^j$ to $a_1,b_1,c_1$ with three paths $P(s_i^j,a_1),P(s_i^j,b_1),P(s_i^j,c_1)$ of lengths $\frac{M}{2}+10x, \frac{M}{2}+5x+1$ and $\frac{M}{2}-10x$ respectively, link $s_i^j$ to $a_2,b_2,c_2$ with three paths $P(s_i^j,a_2),P(s_i^j,b_2),P(s_i^j,c_2)$ of lengths $\frac{M}{2}+10y, \frac{M}{2}+5y+1$ and $\frac{M}{2}-10y$ respectively, and link $s_i^j$ to $a_3,b_3,c_3$ with three paths $P(s_i^j,a_3),P(s_i^j,b_3),P(s_i^j,c_3)$ of lengths $\frac{M}{2}+10z, \frac{M}{2}+5z+1$ and $\frac{M}{2}-10z$ respectively. For every vertex pair $\{u_r^p,v_r^p\}$ ($p\in [n],r\in\{1,2,3\}$), we link $u_r^p$ to $a_r,b_r,c_r$ with three paths $P(u_r^p, a_r),P(u_r^p, b_r),P(u_r^p, c_r)$ of lengths $\frac{M}{2}-10p, \frac{M}{2}-5p-1$ and $\frac{M}{2}+10p$ respectively, and link $v_r^p$ to $a_r,b_r,c_r$ with three paths $P(v_r^p, a_r),P(v_r^p, b_r),P(v_r^p, c_r)$ of lengths $\frac{M}{2}-10p, \frac{M}{2}-5p-2$ and $\frac{M}{2}+10p$ respectively. This finishes the construction.

\begin{lemma} \label{resolve}
For an arbitrary vertex pair $\{u_r^x,v_r^x\}\in {\cal P}$ ($r\in \{1,2,3\}, x\in [n]$),$\{u_r^x,v_r^x\}$ is resolved by $s_i^j$ ($i\in [n],j\in [m]$) if and only if $(r,x)\in A_j$.
\end{lemma}
\begin{proof}
On one hand, suppose that $(r,x)\in A_j$. For an arbitrary $i\in [n]$, the three paths from $s_i^j$ to $u_r^x$ via $a_r,b_r$ and $c_r$ have lengths $M,M+1$ and $M$ respectively. The three paths from $s_i^j$ to $v_r^x$ via $a_r,b_r$ and $c_r$ have lengths $M,M-1$ and $M$ respectively. Note that there could be other paths from $s_i^j$ to $v_r^x$ or $u_r^x$ that go repeatedly between vertices in $X$ and vertices in $W$. However, the lengths of such paths are at least $M-20n+M-10n>M$. As a result, the shortest paths from $s_i^j$ to $u_r^x$ and $v_r^x$ are of lengths $M$ and $M-1$ respectively. Thus $\{u_r^x,v_r^x\}$ is resolved by $s_i^j$.

On the other hand, for an arbitrary tuple $A_i=\{(1,p_1),(2,p_2),(3,p_3)\}$, the paths from the vertex $s_i^j$ ($i\in [n]$) to $u_r^x$ ($r\in \{1,2,3\}$) via $a_r,b_r$ and $c_r$ have lengths $M+10(p_r-x),M+5(p_r-x)+1$ and $M-10(p_r-x)$ respectively. The paths from the vertex $s_i^j$ ($i\in [n]$) to $v_r^x$ ($r\in \{1,2,3\}$) via $a_r,b_r$ and $c_r$ have lengths $M+10(p_r-x),M+5(p_r-x)-1$ and $M-10(p_r-x)$ respectively. Note that the paths from $s_i^j$ to $u_r^x$ (or $v_r^x$) that go repeatedly between the vertices in $X$ and the vertices in $W$ have lengths at least $M-20n+M-10n>M+10n$.  They are not the shortest paths from $s_i^j$ to $u_r^x$ (or $v_r^x$). If $p_r<x$, the shortest paths from $s_i^j$ to $u_r^x$ and $v_r^x$ both have lengths $M+10(p_r-x)$. If $p_r>x$, the shortest paths from $s_i^j$ to $u_r^x$ and $v_r^x$ both have lengths $M-10(p_r-x)$. If $p_r=x$, the shortest paths from $s_i^j$ to $u_r^x$ and $v_r^x$ have lengths $M$ and $M-1$ respectively. As a result, if $\{u_r^x,v_r^x\}$ is resolved by $s_i^j$, then $p_r=x$. According to the construction, $(r,x)\in A_j$.
\end{proof}

\begin{lemma} \label{multiRSnp}
The constructed instance $(G,n,\chi,\cal P)$ of \textsc{$n$-Multicolored Resolving Set} is a yes-instance if and only if the given instance $(U,\cal F)$ of \textsc{$3$-dimensional matching} is a yes-instance.
\end{lemma}
\begin{proof}

($\Leftarrow$)
For an arbitrary tuple $A_i=\{(1,x),(2,y),(3,z)\}$, according to Lemma~\ref{resolve}, pairs $\{u_1^x,v_1^x\}$,$\{u_2^y,v_2^y\}$ and $\{u_3^z,v_3^z\}$ are all resolved by $s_i^j$ for every $i\in [n]$. Suppose that the given instance of \textsc{$3$-dimensional matching} is a yes-instance, that is, there exists $A_{j_1},...,A_{j_n}$ satisfying that $\bigcup\limits_{h=1}^{n} A_{j_h}=U$. It follows that $S=\{s_h^{j_h}:h\in [n]\}$ is a solution for the constructed instance of \textsc{$n$-Multicolored Resolving Set}.

($\Rightarrow$)
Let $S=\{s_h^{j_h}:h\in [n]\}$ be a solution for the constructed instance of \textsc{$n$-Multicolored Resolving Set}. For an arbitrary pair $\{u_r^x,v_r^x\}$, since it is resolved by some $s_{h'}^{j_{h'}}\in S$, according to Lemma~\ref{resolve}, $(r,x)\in A_{j_{h'}}$. As a result, $\{A_{j_h}:h\in [n]\}$ is a solution for the instance of \textsc{$3$-dimensional matching}.
\end{proof}

It is well-known that the treewidth of a graph is bounded by the size of a minimum feedback vertex set of the graph. We can easily observe that $W$ is a feedback vertex set of size $9$ for $G$. It follows that the treewidth of $G$ is at most $10$. Then we have the following lemma.

\begin{lemma}
$k$-Multicolored Resolving Set is NP-hard even on graphs of treewidth at most $10$.
\end{lemma}

\section{Reduction from Multicolored Resolving Set to Metric Dimension}\label{sec:MD}
In this section, we create in polynomial time an instance $(G',k)$ of \textsc{Metric Dimension} which is equivalent to the instance $(G,n,\chi,\cal P)$ of \textsc{$n$-Multicolored Resolving Set} we created in last section.
Roughly speaking, the reduction consists in adding gadgets on base of the constructed instance $(G,n,\chi,\cal P)$ to solve the following two issues:
(1) the solution for \textsc{Metric Dimension} could contain vertices not in any set of $\chi$ or more than one vertex from some set of $\chi$, which could spoil the desired reduction;
(2) we did not make sure that every pair of distinct vertices are resolved by the solution in an instance of \textsc{$n$-Multicolored Resolving Set}.
We find that similar strategies to those in~\cite{BonnetP19} can be used to solve these two issues.
More specifically, we solve the first issue by adding \emph{forced set gadgets}.
One such gadget contains two pairs of vertices such that they are only resolved simultaneously by a vertex of $X_i$ (where it is attached).
We solve the second issue by adding \emph{forced vertex gadgets}.
One such gadget contains a pair of pendant neighboring vertices (false twins).
Obviously one vertex of the false twins has to be chosen in the solution.
Thus the chosen vertices (\emph{forced vertices}) are designed to resolve the remaining unresolved vertex pairs.
Besides, we need to add a number of extra paths and set appropriate budget $k$ to make sure that the reduction works as described above.

\subsection{Construction of the forced set gadgets}\label{ss:forced-set}

Let $(G,n,\chi,\cal P)$ be an instance of \textsc{$n$-Multicolored Resolving Set} that we created in last section. For every $X_i\in \chi$ ($i\in [n]$), we add two pairs of isolated vertices $\{p_i^1,q_i^1\}$ and $\{p_i^2,q_i^2\}$.  Then we add two vertices $\pi_i^1$ and $\pi_i^2$ such that $p_i^1,q_i^1$ are adjacent to $\pi_i^1$, $p_i^2,q_i^2$ are adjacent to $\pi_i^2$. The vertex triples $p_i^1,q_i^1,\pi_i^1$ and $p_i^2,q_i^2,\pi_i^2$ ($i\in [n]$) form a forced set gadget. Then we create a path $P(s_i^j, p_i^1)$ of length $20(n+1)$ from $s_i^j$ to $p_i^1$ and create a path $P(s_i^j, p_i^2)$ of length $20(n+1)$ from $s_i^j$ to $p_i^2$ for each $i\in [n], j\in [m]$. In order to make sure that a vertex can resolve $p_i^1,q_i^1$ and $p_i^2,q_i^2$ simultaneously if and only if it belongs to $X_i$, we need to create $4$ paths of length $20(n+1)$ from $\pi_i^1$ to $N_{s_i^j}(s_i^j,a_r)$, from $\pi_i^1$ to $N_{s_i^j}(s_i^j,b_r)$, from $\pi_i^1$ to $N_{s_i^j}(s_i^j,c_r)$ and from $\pi_i^1$ to $N_{s_i^j}(s_i^j,p_i^2)$ respectively for each $i\in [n]$, $j\in [m]$ and $r\in \{1,2,3\}$. For simplicity, we name the four paths as $P^{1}(i,j,a_r)$, $P^{1}(i,j,b_r)$, $P^{1}(i,j,c_r)$ and $P^{1}(i,j,p_i^2)$ respectively. Symmetrically, we need to create $4$ paths of length $20(n+1)$ from $\pi_i^2$ to $N_{s_i^j}(s_i^j,a_r)$, from $\pi_i^2$ to $N_{s_i^j}(s_i^j,b_r)$, from $\pi_i^2$ to $N_{s_i^j}(s_i^j,c_r)$ and from $\pi_i^2$ to $N_{s_i^j}(s_i^j,p_i^1)$ respectively for each $i\in [n]$ and $r\in \{1,2,3\}$. For simplicity, we name the four paths as $P^{2}(i,j,a_r)$, $P^{2}(i,j,b_r)$, $P^{2}(i,j,c_r)$ and $P^{2}(i,j,p_i^1)$ respectively. Let $\Pi^{h}(i,j,r)=\{P^{h}(i,j,a_r), P^{h}(i,j,b_r), P^{h}(i,j,c_r), P^{h}(i,j,p_i^{3-h})\}$ for $i\in [n],j\in [m],r\in\{1,2,3\},h\in\{1,2\}$.

This completes the construction of the first phase.

\subsection{Construction of the forced vertex gadgets}\label{ss:forced-vertex}
A forced vertex gadget consists of a triangle, namely three vertices such that each vertex is adjacent to the other two vertices.
Two vertices of the triangle are false twins whose degrees are exactly $2$ and we call the other vertex in the triangle the \emph{connecting vertex} of the gadget.
When we say that we add a forced vertex gadget $F$ to a vertex $v$, we mean that we create a forced vertex gadget $F$ such that $v$ is identified with the connecting vertex of $F$.
For each $i\in [n],j\in [m],r\in\{1,2,3\},h\in\{1,2\}$,
we add a forced vertex gadget $F^{h}(i,j,a_r)$ to $N_{\pi_h}^{h}(i,j,a_r)$, $F^{h}(i,j,b_r)$ to $N_{\pi_h}^{h}(i,j,b_r)$, $F^{h}(i,j,c_r)$ to $N_{\pi_h}^{h}(i,j,c_r)$ and $F^{h}(i,j,p_i^{3-h})$ to $N_{\pi_h}^{h}(i,j,,p_i^{3-h})$.

In order to make sure that $f^{h}(i,j,b_r)$ for $i\in [n],j\in [m],r\in\{1,2,3\},h\in\{1,2\}$ does not resolve any vertex pair of $\cal P$, we create a path $P(\pi_i^h,a_r)$ and a path $P(\pi_i^h,c_r)$ both of length $10(n+1)$ for $i\in [n]$, $h\in \{1,2\}$ and $r\in \{1,2,3\}$.

For each $i\in [n],j\in [m],r\in\{1,2,3\},h\in\{1,2\}$, we add a forced vertex gadget $F(\pi_i^h,a_r)$ to $N_{a_r}(\pi_i^h,a_r)$ and a forced vertex gadget $F(\pi_i^h,c_r)$ to $N_{c_r}(\pi_i^h,c_r)$. For each $i\in [n],j\in [m],r\in\{1,2,3\}$, we add a forced vertex gadget $F(s_i^j,a_r)$ to $N_{a_r}(s_i^j,a_r)$ and a forced vertex gadget $F(s_i^j,c_r)$ to $N_{c_r}(s_i^j,c_r)$.

Let $\text{mid}(P^{h}(i,j,p_i^{3-h}))$ be the middle vertex of $P^{h}(i,j,p_i^{3-h})$ for $i\in [n],j\in [m],h\in\{1,2\}$.  In order to make sure that $f^{h}(i,j,p_i^{3-h})$ does not resolve the vertex pair $\{p_i^{3-h},q_i^{3-h}\}$, create a path $P(q_i^h,\text{mid}(P^{3-h}(i,j,p_i^{h})))$ from $q_i^h$ to $\text{mid}(P^{3-h}(i,j,p_i^{h}))$ of length $|P^{3-h}(i,j,p_i^h)|/2+|P(s_i^j, p_i^h)|-1$. Then add a forced vertex gadget $F^{mid}(i,j,h)$ to $\text{mid}(P^{h}(i,j,p_i^{3-h}))$.

For $i\in [n],j\in [m],r\in\{1,2,3\},h\in\{1,2\}$, add a forced vertex gadget $F^{ecc}(i,j,h,r)$ to the vertex $x\in P^{h}(i,j,a_r)$ such that $\dist(\pi_i^h,x)=10(n+1)+1$.

For each $i\in [n],r\in \{1,2,3\}$, create two forced vertex gadgets $F^1(u_r^i,v_r^i)$ and $F^2(u_r^i,v_r^i)$ for the vertex pair $\{u_r^i,v_r^i\}\in {\cal P}_r$. Then create an edge from the connecting vertex of $F^1(u_r^i,v_r^i)$ to $u_r^i$, to $v_r^i$, to $N_{u_r^i}(a_r,u_r^i)$ and to $N_{u_r^i}(c_r,u_r^i)$ respectively for $i\in [n],r\in \{1,2,3\}$. Create an edge from the connecting vertex of $F^2(u_r^i,v_r^i)$ to $u_r^i$, to $v_r^i$, to the vertex $x$ such that $x\in P(a_r,u_r^i)$ and $\dist(x,u_r^i)=2$, and to the vertex $y$ such that $y\in P(c_r,u_r^i)$ and $\dist(y,u_r^i)=2$. This completes the construction of the second phase.

Finally, let $G'$ be the graph constructed by above two phases and set $k=34nm+19n$.
This finishes constructing the instance $(G',k)$ of \textsc{Metric Dimension}.
Figure~\ref{fig:main} shows a part of $G'$.

\begin{figure}[hbtp]
\hspace{6mm}
 \scalebox{0.8}{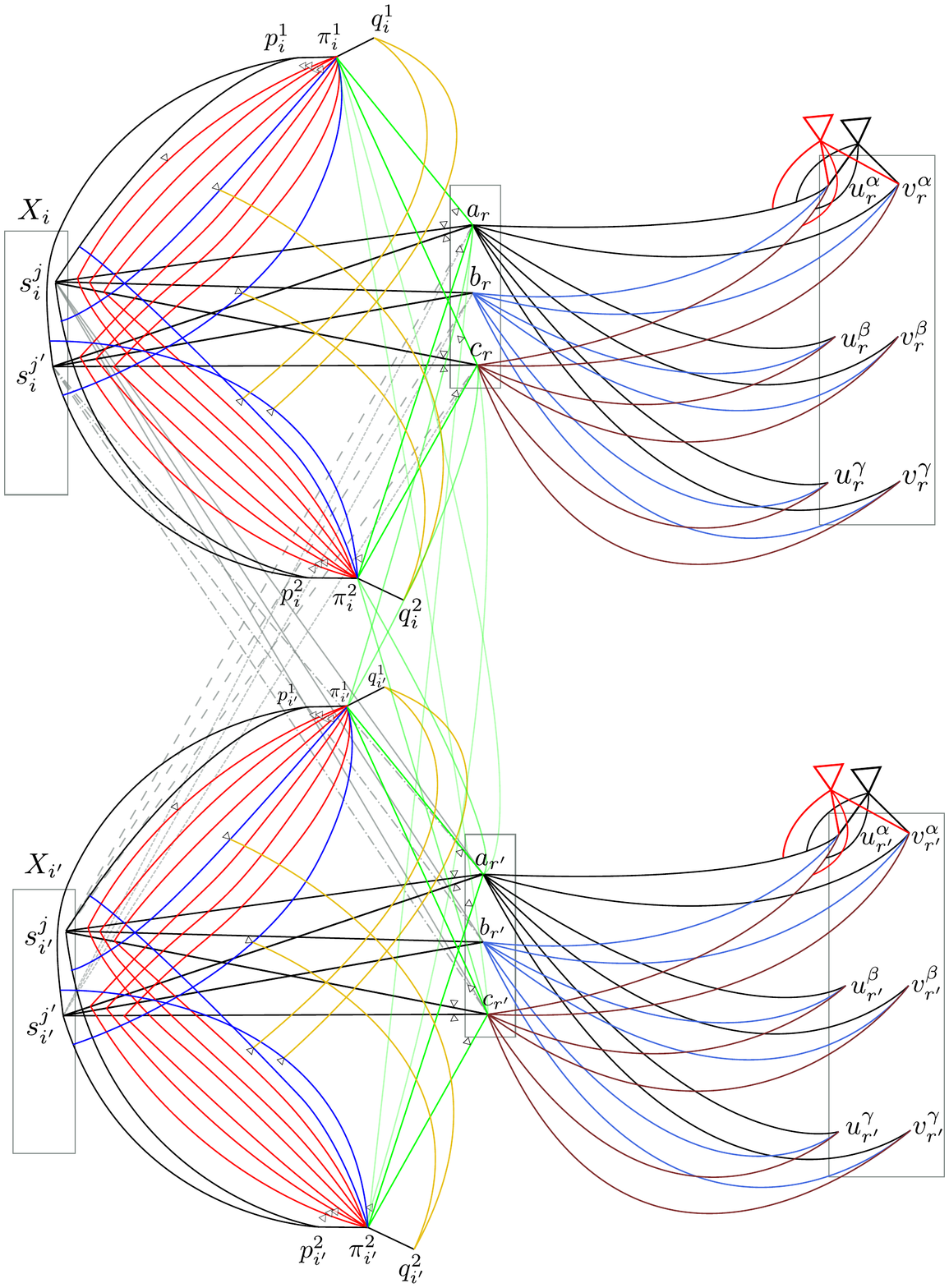}
\vspace*{-33mm}
\caption{An example showing a part of $G'$. Triangles represent corresponding forced vertex gadgets. For clarity, some forced vertex gadgets do not appear on the figure. Dotted or dashed lines are used in order for cleanness of the figure. }
\label{fig:main}
\end{figure}

\subsection{Soundness of the reduction}\label{ss:easy-direction}
First, we define the vertex sets to be used in the following parts.  For every $i\in [n],r\in\{1,2,3\},h\in\{1,2\}$, let
$$U_i^h=\bigcup_{j\in [m]}P(s_i^j,p_i^h),$$
$$H_{i,r}=\bigcup_{j\in [m]}P(s_i^j,a_r)\cup P(s_i^j,b_r)\cup P(s_i^j,c_r),$$
$$S_i^h=\bigcup_{r\in \{1,2,3\}}P(\pi_i^h,a_r)\cup P(\pi_i^h,c_r),$$
$$L_i^h=\bigcup_{j\in [m]}P(q_i^h,\text{mid}(P^{3-h}(i,j,p_i^{h}))),$$
$$R_r=\bigcup_{i\in [n]}P(a_r,u_r^i)\cup P(a_r,v_r^i)\cup P(b_r,u_r^i)\cup P(b_r,v_r^i)\cup P(c_r,u_r^i)\cup P(c_r,v_r^i),\text{ and}$$
$$\Pi^h(i,j,r)=P^{h}(i,j,a_r)\cup P^{h}(i,j,b_r)\cup P^{h}(i,j,c_r)\cup P^{h}(i,j,p_i^{3-h}).$$

For every $i\in [n]$, let
$$U_i=\bigcup_{h\in \{1,2\}}U_i^h \quad\quad\quad H_i=\bigcup_{r\in \{1,2,3\}}H_{i,r} \quad\quad\quad S_i=\bigcup_{h\in \{1,2\}}S_i^h$$
$$L_i=\bigcup_{h\in \{1,2\}}L_i^h \quad\quad\quad\quad\quad\quad \Pi_i=\bigcup_{j\in [m],r\in \{1,2,3\},h\in\{1,2\}}\Pi^h(i,j,r).$$

Let ${\cal F}$ be the union of all forced vertex gadgets, i.e. ${\cal F}=\bigcup_{i\in[n],j\in [m],r\in \{1,2,3\},h\in \{1,2\}}(F(s_i^j,a_r)\cup F(s_i^j,c_r)\cup F(\pi_i^h,a_r)\cup F(\pi_i^h,c_r)\cup F^h(u_r^i,v_r^i)\cup F^{h}(i,j,a_r)\cup F^{h}(i,j,b_r)\cup F^{h}(i,j,c_r)\cup F^{h}(i,j,p_i^{3-h})\cup F^{mid}(i,j,h)\cup F^{ecc}(i,j,h,r))$.

Next we introduce a lemma about forced set gadgets and this lemma is important for the correctness of the reduction.

\begin{lemma} \label{forcedSet}
The following three statements are true for the instance $(G',k)$.
\begin{enumerate}
     \item[(a)] The vertex $s_i^j$ for $i\in [n],j\in [m]$ resolves both pairs $\{p_i^1,q_i^1\}$ and $\{p_i^2,q_i^2\}$. Moreover, $s_i^j$ does not resolve any vertex pair $\{p_{i'}^h,q_{i'}^h\}$ such that $i'\in [n],h\in \{1,2\}$ and $i'\neq i$.
     \item[(b)] The vertices of any forced vertex gadget do not resolve any vertex pair of $\{\{p_i^h,q_i^h\}~|~i\in [n],h\in \{1,2\}\}$.
     \item[(c)] Any vertex $v\in V(G')\setminus (X_i\cup {\cal F})$ resolves at most one vertex pair of $\{\{p_i^h,q_i^h\}~|~i\in [n],h\in \{1,2\}\}$.
\end{enumerate}
\end{lemma}
\begin{proof}
By the construction of $G'$, $\dist(s_i^j,q_i^h)=|P(s_i^j,p_i^h)|+2=20(n+1)+2>\dist(s_i^j,p_i^h)$ for $i\in [n],j\in [m]$ and $h\in\{1,2\}$.
Thus any vertex of $X_i$ resolves both pairs $\{p_i^1,q_i^1\}$ and $\{p_i^2,q_i^2\}$ for $i\in [n]$.
For a vertex pair $\{p_{i'}^{h'},q_{i'}^{h'}\}$ such that $i'\neq i$,
there is a shortest path from $s_i^j$ to $p_{i'}^{h'}$ or $q_{i'}^{h'}$ going through $c_{r'}$ and $\pi_{i'}^{h'}$ with some integer $r'\in \{1,2,3\}$.
Thus a vertex $s\in X_i$ resolves exactly two vertex pairs of $\{\{p_i^h,q_i^h\}:i\in [n], h\in \{1,2\}\}$.

First we claim that vertices of ${\cal F}$ do not resolve any vertex pair $\{p_{i'}^{h'},q_{i'}^{h'}\}$ for $i'\in [n],h'\in\{1,2\}$.
For any vertex $v\in F^h(u_r^i,v_r^i)$ for $i\in[n],r\in \{1,2,3\},h\in \{1,2\}$,
there is a shortest path from $v$ to $p_{i'}^{h'}$ or $q_{i'}^{h'}$ going through $a_r$ and $\pi_{i'}^{h'}$.
Thus $v$ does not resolve any vertex pair $\{p_{i'}^{h'},q_{i'}^{h'}\}$ for $i'\in [n],h'\in\{1,2\}$.
For any vertex $v\in F^{mid}(i,j,h)\cup F^{ecc}(i,j,h,r)$ for $i\in[n],j\in [m],h\in \{1,2\},r\in \{1,2,3\}$,
we can see that $\dist(v,p_{i'}^{h'})=\dist(v,q_{i'}^{h'})$ with $i'=i$.
There is a shortest path from $v$ to $p_{i'}^{h'}$ or $q_{i'}^{h'}$ going through $\pi_i^h$, $a_r$ and $\pi_{i'}^{h'}$ with $i'\neq i$.
Thus $v$ does not resolve any vertex pair $\{p_{i'}^{h'},q_{i'}^{h'}\}$ for $i'\in [n],h'\in\{1,2\}$.
For any vertex $v\in {\cal F}\setminus \bigcup_{i\in[n],j\in [m],r\in \{1,2,3\},h\in \{1,2\}}(F^h(u_r^i,v_r^i)\cup F^{ecc}(i,j,h,r)\cup F^{mid}(i,j,h))$,
there is a shortest path from $v$ to $p_{i'}^{h'}$ or $q_{i'}^{h'}$ going through $\pi_{i'}^{h'}$ with $i'=i$.
There is a shortest path from $v$ to $p_{i'}^{h'}$ or $q_{i'}^{h'}$ going through $c_r$ (or $a_r$) and $\pi_{i'}^{h'}$ with $i'\neq i$.
Thus $v$ does not resolve any pair $\{p_{i'}^{h'},q_{i'}^{h'}\}$.
As a result, vertices of ${\cal F}$ do not resolve any vertex pair $\{p_{i'}^{h'},q_{i'}^{h'}\}$ for $i'\in [n],h'\in\{1,2\}$.

Then we show that any vertex $v\in V(G')\setminus (X_i\cup {\cal F})$ resolves at most one pair of $\{p_i^1,q_i^1\}$ and $\{p_i^2,q_i^2\}$.

For a vertex $v\in U_i^h\setminus X_i$ for $i\in [n],h\in\{1,2\}$, $\dist(v,p_i^h)=\dist(v,q_i^h)-2<\dist(v,q_i^h)$.   $\dist(v,q_i^{3-h})=\dist(v,N_{s_i^j}(s_i^j,p_i^h))+|P^{3-h}(i,j,p_i^{h})|+1=\dist(v,p_i^{3-h})$.
For a vertex pair $\{p_{i'}^{h'},q_{i'}^{h'}\}$ such that $i'\neq i$,
there is a shortest path from $v$ to $p_{i'}^{h'}$ or $q_{i'}^{h'}$ going through $\pi_{i'}^{h'}$.
Thus $v\in U_i^h\setminus X_i$ for $i\in [n],h\in\{1,2\}$ resolves exactly one vertex pair of $\{\{p_i^h,q_i^h\}:i\in [n], h\in \{1,2\}\}$.

Let $P(\text{mid}(P^{3-h}(i,j,p_i^{h})),N_{s_i^j}(s_i^j,p_i^h))$ be the subpath of $P^{3-h}(i,j,p_i^{h})$ from $\text{mid}(P^{3-h}(i,j,p_i^{h}))$ to $N_{s_i^j}(s_i^j,p_i^h)$.
Let $\Lambda_i^h=(\bigcup_{j\in [m]}P(\text{mid}(P^{3-h}(i,j,p_i^{h})),N_{s_i^j}(s_i^j,p_i^h)))\setminus \{\text{mid}(P^{3-h}(i,j,p_i^{h}))~|~j\in [m]\}$.
For a vertex $v\in \Lambda_i^h$ for $i\in [n],h\in\{1,2\}$,
$\dist(v,p_i^h)=\dist(v,q_i^h)-2<\dist(v,q_i^h)$.
$\dist(v,q_i^{3-h})=\dist(v,\pi_i^{3-h})+1=\dist(v,p_i^{3-h})$.
For a vertex pair $\{p_{i'}^{h'},q_{i'}^{h'}\}$ such that $i'\neq i$,
there is a shortest path from $v$ to $p_{i'}^{h'}$ or $q_{i'}^{h'}$ going through $\pi_{i'}^{h'}$.
Thus $v\in \Lambda_i^h$ for $i\in [n],h\in\{1,2\}$ resolves exactly one vertex pair of $\{\{p_i^h,q_i^h\}:i\in [n], h\in \{1,2\}\}$.

For a vertex $v\in L_i^h\setminus \{\text{mid}(P^{h}(i,j,p_i^{3-h}))~|~j\in [m]\}$ for $i\in [n],h\in\{1,2\}$, $\dist(v,q_i^h)=\dist(v,p_i^h)-2<\dist(v,p_i^h)$.
There is a shortest path from $v$ to $p_i^{3-h}$ or $q_i^{3-h}$ going through $\pi_i^{3-h}$.
For a vertex pair $\{p_{i'}^{h'},q_{i'}^{h'}\}$ such that $i'\neq i$,
there is a shortest path from $v$ to $p_{i'}^{h'}$ or $q_{i'}^{h'}$ going through $\pi_{i'}^{h'}$.
Thus $v\in L_i^h\setminus \{\text{mid}(P^{h}(i,j,p_i^{3-h}))~|~j\in [m]\}$ for $i\in [n],h\in\{1,2\}$ resolves exactly one vertex pair of $\{\{p_i^h,q_i^h\}:i\in [n], h\in \{1,2\}\}$.

For a vertex $v\in \Pi_i\cup S_i\cup H_i\setminus (X_i\cup \Lambda_i^1\cup \Lambda_i^2)$ for $i\in [n]$,
there is a shortest path from $v$ to $p_{i'}^{h'}$ or $q_{i'}^{h'}$ going through $\pi_{i'}^{h'}$ with $i=i',h'\in \{1,2\}$.
For a vertex pair $\{p_{i'}^{h'},q_{i'}^{h'}\}$ such that $i'\neq i$,
there is a shortest path from $v$ to $p_{i'}^{h'}$ or $q_{i'}^{h'}$ going through $\pi_{i'}^{h'}$.
Thus $v$ does not resolve any vertex pair of $\{\{p_i^h,q_i^h\}:i\in [n], h\in \{1,2\}\}$.

For a vertex $v\in R_r$ for $r\in \{1,2,3\}$, there is a shortest path from $v$ to $p_i^h$ or $q_i^h$ for $i\in [n],h\in \{1,2\}$ going through $a_r$ and $\pi_i^h$.
Thus $v$ does not resolve any vertex pair of $\{\{p_i^h,q_i^h\}:i\in [n], h\in \{1,2\}\}$.
This completes the proof for the lemma.
\end{proof}

\begin{lemma} \label{forcedVertex}
The forced vertices do not resolve any vertex pair $\{u_r^i,v_r^i\}\in \cal P$ for $r\in \{1,2,3\}$ and $i\in [n]$.
\end{lemma}
\begin{proof}
We fix arbitrary integers $i\in [n],j\in [m],r\in\{1,2,3\},h\in\{1,2\}$.
For the forced vertex $f^{h}(i,j,a_r)$, $\dist(f^{h}(i,j,a_r),u_{r'}^{i'})=2+|P(\pi_i^h,a_{r'})|+|P(a_{r'},u_{r'}^{i'})|=2+|P(\pi_i^h,a_{r'})|+|P(a_{r'},v_{r'}^{i'})|=\dist(f^{h}(i,j,a_r),v_{r'}^{i'})$ for $i'\in [n], r'\in \{1,2,3\}$.
Thus $f^{h}(i,j,a_r)$ does not resolve any vertex pair of $\cal P$.
Similarly, the forced vertices $f^{h}(i,j,b_r)$, $f^{h}(i,j,c_r)$ and $f^{h}(i,j,p_i^{3-h})$ do not resolve any vertex pair of $\cal P$.
For the forced vertex $f^{mid}(i,j,h)$, $\dist(f^{mid}(i,j,h),u_{r'}^{i'})=\dist(f^{mid}(i,j,h),v_{r'}^{i'})=|P^h(i,j,p_i^{3-h})|/2+|P(\pi_i^h,a_{r'})|+|P(a_{r'},u_{r'}^{i'})|$.
Thus $f^{mid}(i,j,h)$ does not resolve any vertex pair of $\cal P$.
For the forced vertex $f^{ecc}(i,j,h,r)$, $\dist(f^{ecc}(i,j,h,r),u_{r'}^{i'})=\dist(f^{ecc}(i,j,h,r),v_{r'}^{i'})=10(n+1)+1+|P(\pi_i^h,a_{r'})|+|P(a_{r'},u_{r'}^{i'})|$.
Thus $f^{ecc}(i,j,h,r)$ does not resolve any vertex pair of $\cal P$.

We fix arbitrary integers $i\in [n],j\in [m],r\in\{1,2,3\}$.
For the forced vertex $f(s_i^j,a_r)$, $\dist(f(s_i^j,a_r),u_r^{i'})=2+|P(a_r,u_r^{i'})|=2+|P(a_r,v_r^{i'})|=\dist(f(s_i^j,a_r),v_r^{i'})$ for $i'\in [n]$.
For the forced vertex $f(s_i^j,c_r)$, $\dist(f(s_i^j,c_r),u_r^{i'})=2+|P(c_r,u_r^{i'})|=2+|P(c_r,v_r^{i'})|=\dist(f(s_i^j,c_r),v_r^{i'})$ for $i'\in [n]$.
Thus $f(s_i^j,a_r)$ and $f(s_i^j,c_r)$ do not resolve any vertex pair of ${\cal P}_r$.
Similarly, $f(\pi_i^h,a_r)$ and $f(\pi_i^h,c_r)$ for $i\in [n],h\in \{1,2\},r\in\{1,2,3\}$ do not resolve any vertex pair of ${\cal P}_r$.
For vertex pairs of ${\cal P}_{r'}$ with $r'\in \{1,2,3\}$ and $r'\neq r$, $\dist(f(s_i^j,a_r),u_{r'}^{i'})=2+|P(a_r,\pi_i^1)|+|P(a_{r'},\pi_i^1)|+|P(a_{r'},u_r^{i'})|=2+|P(a_r,\pi_i^1)|+|P(a_{r'},\pi_i^1)|+|P(a_{r'},v_r^{i'})|
=\dist(f(s_i^j,a_r),v_{r'}^{i'})$ for $i'\in [n]$.
$\dist(f(s_i^j,c_r),u_{r'}^{i'})=2+|P(c_r,\pi_i^1)|+|P(a_{r'},\pi_i^1)|+|P(a_{r'},u_r^{i'})|=2+|P(c_r,\pi_i^1)|+|P(a_{r'},\pi_i^1)|+|P(a_{r'},v_r^{i'})|
=\dist(f(s_i^j,a_r),v_{r'}^{i'})$ for $i'\in [n]$.
Thus $f(s_i^j,a_r)$ and $f(s_i^j,c_r)$ do not resolve any vertex pair of ${\cal P}_{r'}$.


We fix arbitrary integers $i\in [n],r\in\{1,2,3\}$.
For the forced vertex $f^1(u_r^i,v_r^i)$ or $f^2(u_r^i,v_r^i)$, obviously it does not resolve the vertex pair $\{u_r^i,v_r^i\}$.
For a vertex pair $\{u_r^{i'},v_r^{i'}\}$ with $i'\in [n]\text{ and }i'\neq i$, $\dist(f^1(u_r^i,v_r^i),u_r^{i'})=2+|P(a_r,u_r^i)|-1+|P(a_r,u_r^{i'})|=2+|P(a_r,u_r^i)|-1+|P(a_r,v_r^{i'})|=\dist(f^1(u_r^i,v_r^i),v_r^{i'})$.
For a vertex pair $\{u_{r'}^{i'},v_{r'}^{i'}\}$ with $i'\in [n]$ and $r'\in \{1,2,3\}\text{ and }r'\neq r$, $\dist(f^1(u_r^i,v_r^i),u_{r'}^{i'})=2+|P(a_r,u_r^i)|-1+|P(\pi_i^1,a_r)|+|P(\pi_i^1,a_{r'})|+|P(a_{r'},u_{r'}^{i'})|=\dist(f^1(u_r^i,v_r^i),v_{r'}^{i'})$.
As a result, $f^1(u_r^i,v_r^i)$ does not resolve any vertex pair of $\cal P$.
For a vertex pair $\{u_r^{i'},v_r^{i'}\}$ with $i'\in [n]\text{ and }i'\neq i$, $\dist(f^2(u_r^i,v_r^i),u_r^{i'})=2+|P(a_r,u_r^i)|-2+|P(a_r,u_r^{i'})|=2+|P(a_r,u_r^i)|-2+|P(a_r,v_r^{i'})|=\dist(f^2(u_r^i,v_r^i),v_r^{i'})$.
For a vertex pair $\{u_{r'}^{i'},v_{r'}^{i'}\}$ with $i'\in [n]$, $r'\in \{1,2,3\}\text{ and }r'\neq r$, $\dist(f^2(u_r^i,v_r^i),u_{r'}^{i'})=2+|P(a_r,u_r^i)|-2+|P(\pi_i^1,a_r)|+|P(\pi_i^1,a_{r'})|+|P(a_{r'},u_{r'}^{i'})|=\dist(f^2(u_r^i,v_r^i),v_{r'}^{i'})$.
As a result, $f^2(u_r^i,v_r^i)$ does not resolve any vertex pair of $\cal P$.
This completes the proof for the lemma.
\end{proof}

\begin{lemma} \label{soundness}
If $G'$ has a resolving set of size at most $34nm+19n$, then $(G,n,\chi,\cal P)$ is a yes-instance.
\end{lemma}
\begin{proof}
Suppose that $S$ is a resolving set for $G'$ of size at most $34nm+19n$. Let $\hat{S}=S\cap X$.
We claim that $\hat{S}$ is solution for $(G,n,\chi,\cal P)$.
Note that for the false twins $\{u,u'\}$ of a forced vertex gadget, no vertex resolves the vertex pair $\{u,u'\}$ except $u$ (or $u'$).
It follows that $S$ contains $34nm+18n$ forced vertices since there are $34nm+18n$ forced vertex gadgets in $G'$.
Since $X$ has no intersection with the vertex set of all forced vertex gadgets, $|\hat{S}|\leq n$.
By Lemma~\ref{forcedSet}, we get that $|\hat{S}\cap X_i|=1$ for each $i\in [n]$. Thus $|\hat{S}|=n$.
By Lemma ~\ref{forcedVertex} and the assumption that $S$ is a resolving set for $G'$,
$\hat{S}$ resolves every pair $\{u_r^i,v_r^i\}$ in $G'$ for $r\in \{1,2,3\}$ and $i\in [n]$.
We can check that the distance between $s_i^j$ and $u_r^{i'}$ in $G'$ (and the distance between $s_i^j$ and $v_r^{i'}$ in $G'$) for $i\in [n],j\in [m],i'\in [n],r\in \{1,2,3\}$ is the same as that in $G$.
Thus $\hat{S}$ is a solution for $(G,n,\chi,\cal P)$.
\end{proof}

\subsection{Treewidth bound of the graph} \label{ss:treewidth}
Since the completeness proof takes a large amount of space, before proceeding to that, we first show that $G'$ is of constant treewidth.
In fact, we will prove a slightly stronger statement that $G'$ is of constant pathwidth by giving a search strategy with a constant number of searchers.

\begin{lemma}  \label{pathwidth}
The pathwidth of $G'$ is at most $\finaltreewidth$.
\end{lemma}
\begin{proof}
We give a search strategy with 25 searchers.
First, we put 9 searchers on $\bigcup_{r\in \{1,2,3\}}\{a_r,b_r,c_r\}$.
The 9 searchers remain there until the end of the whole search process.
The search process consists of two phases.
We search the ``left'' part of $G'$ in the first phase and the ``right'' part of $G'$ in the second phase.

The first phase of the search process consists of $n$ rounds.
At the beginning of the $i$-th round ($i\in [n]$), we put 6 searchers on $\bigcup_{h\in \{1,2\}}\{p_i^h,q_i^h,\pi_i^h\}$.
Here when we say that we clean a path, there are already two searchers guarding at the endpoints (or the neighbor of the endpoints) of this path and
we use 3 extra searchers $x,y,z$ such that $x,y$ move alternately from one end of the path to the other end to clean the edges of the path.
When a searcher, say $x$ arrives at the connecting point of a forced vertex gadget,
we put $y,z$ on the false twins of this forced vertex gadget to clean the edges of this gadget and then after removing $y,z$, put $y$ ahead of $x$ to continue the alternating process if $x$ does not reach the endpoint of this path.
Then for each $j\in [m]$, we
\begin{itemize}
     \item put 5 vertices on $N_{G'}(s_i^j)$.
     \item put 2 vertices on $\text{mid}(P^h(i,j,p_i^{3-h}))$ for $h\in \{1,2\}$.
     \item use 3 extra searchers to clean the paths $P(s_i^j,p_i^h)$ for $h\in \{1,2\}$, the paths $P(s_i^j,a_r)$, $P(s_i^j,b_r)$, $P(s_i^j,c_r)$ for $r\in \{1,2,3\}$, the paths $P^{h}(i,j,a_r)$, $P^{h}(i,j,b_r)$, $P^{h}(i,j,c_r)$, $P^{h}(i,j,p_i^{3-h})$ for $h\in \{1,2\},r\in \{1,2,3\}$ successively (including all forced vertex gadgets attached to the vertices on these paths).
     \item remove the above 10 searchers that are still on the graph.
\end{itemize}
At the end of the $i$-th round, we remove the 6 searchers on $\bigcup_{h\in \{1,2\}}\{p_i^h,q_i^h,\pi_i^h\}$.

The second phase of the search process consists of $3$ rounds.
During the $r$-th round ($r\in \{1,2,3\}$), we operate as follows. For each $i\in [n]$, we
\begin{itemize}
     \item put 4 searchers on $u_r^i$, $v_r^i$ and the connecting point of $F^h(u_r^i,v_r^i)$ for $h\in \{1,2\}$.
     \item use 2 extra searchers to clean the paths $P(a_r,u_r^i),P(b_r,u_r^i),P(c_r,u_r^i),P(a_r,v_r^i),P(b_r,v_r^i)$ and $P(c_r,v_r^i)$ (including the forced vertex gadgets $F^h(u_r^i,v_r^i)$ for $h\in \{1,2\}$ and the edges between $F^h(u_r^i,v_r^i)$ and the paths listed above).
     \item remove the above 6 searchers that are still on the graph.
\end{itemize}
This completes the description of the the search strategy.

As a result, the node search number of $G'$ is at most 25. It follows that the pathwidth of $G'$ is bounded by $\finaltreewidth$.
\end{proof}

\subsection{Completeness of the reduction}\label{ss:difficult-direction}
For every forced vertex gadget of $G'$, we choose a vertex from the false twins arbitrarily as a forced vertex and let the set of all chosen forced vertices be $F$.
In this section, we show that if $(G,n,\chi,\cal P)$ has a solution $S$, then $S'=S\cup F$ is a resolving set of size at most $34nm+19n$ for $G'$.
Formally, we will prove the following lemma.
\begin{lemma} \label{completeness}
If $(G,n,\chi,\cal P)$ is a yes-instance, then $G'$ has a resolving set of size at most $34nm+19n$.
\end{lemma}

To prove the above lemma, we need to show that every pair of distinct vertices of $G'$ is resolved by some vertex of $S'$.
First of all, We have the following claim.
\begin{claim} \label{basic}
Every vertex pair $\{u_r^{i'},v_r^{i'}\}$ in $G'$ for $r\in \{1,2,3\}, i'\in [n]$ is resolved by $S'$.
\end{claim}
\begin{proof}
Since $(G,n,\chi,\cal P)$ is a yes-instance,
$\text{dist}_{G}(s_i^j,u_r^{i'})=\text{dist}_{G'}(s_i^j,u_r^{i'})$, $\text{dist}_{G}(s_i^j,v_r^{i'})=\text{dist}_{G'}(s_i^j,v_r^{i'})$ for $i,i'\in [n],j\in [m],r\in \{1,2,3\}$,
every vertex pair $\{u_r^{i'},v_r^{i'}\}$ in $G'$ for $r\in \{1,2,3\}$ and $i'\in [n]$ is resolved by some vertex of $S\subset S'$.
\end{proof}

Suppose that $V(G')=V_1\cup V_2\cup...\cup V_t$.
Our general method is to show that for each $i\in [t]$, every internal vertex pair of $V_i$ is resolved by $S'$ and every vertex pair of $V_{i'}\times V_i$ for each $i'<i$ is resolved by $S'$.
Note that when we mention the vertex pairs of $V_{i'}\times V_i$, we ignore the vertex pairs with two identical vertices by default as it's meaningless in our problem.
In the following parts, we give a series of lemmas to show that Lemma~\ref{completeness} is true. Let $U=\bigcup_{i\in [n]}U_i,\Pi=\bigcup_{i\in [n]}\Pi_i,H=\bigcup_{i\in [n]}H_i,S=\bigcup_{i\in [n]}S_i,L=\bigcup_{i\in [n]}L_i$ and $R=\bigcup_{r\in \{1,2,3\}}R_r$. Table~\ref{tab} shows the indexes of the corresponding lemmas.

\begin{table}
\centering
\begin{tabular}{ |c|c|c|c|c|c|c| }
 \hline
      & $U$ & $\Pi$ & $S$ & $L$ & $H$ & $R$ \\
 \hline
  $U$ & \ref{UxU} & \ref{UxPi} & \ref{UxS} & \ref{UxL} & \ref{UxH} & \ref{UxR} \\
 \hline
 $\Pi$&          & \ref{PixPi} & \ref{PixS} & \ref{PixL} & \ref{PixH} & \ref{PixR} \\
 \hline
  $S$ &           &            & \ref{SxS} & \ref{LxS} & \ref{SxH} &  \ref{SxR} \\
 \hline
  $L$ &           &            &           & \ref{LxL} & \ref{LxH} & \ref{LxR} \\
 \hline
  $H$ &           &            &           &           & \ref{HxH} & \ref{HxR} \\
 \hline
  $R$ &           &            &           &           &           &  \ref{RxR} \\
 \hline
\end{tabular}
\caption{Indexes of the lemmas for the completeness of the reduction. }
\centering
\label{tab}
\end{table}

\begin{lemma} \label{UxU}
Every pair of distinct vertices $x,y\in\bigcup_{i\in [n],h\in\{1,2\}}U_i^h$ is resolved by $S'$.
\end{lemma}
\begin{proof}
First, we show that every pair of distinct vertices $x,y\in\bigcup_{h\in\{1,2\}}U_i^h$ for $i\in [n]$ is resolved by $S'$.
We fix arbitrary integers $i\in [n],j,j'\in [m]$ and $h\in \{1,2\}$ such that $j'\neq j$.\

Suppose that $x,x'\in P(s_i^j,p_i^h)$.
For a vertex $x\in P(s_i^j,p_i^h)$, let $P(s_i^j,x)$ be the subpath of $P(s_i^j,p_i^h)$ from $s_i^j$ to $x$ and $|P(s_i^j,x)|=\ell_x$.
Since $\dist(f^h(i,j',a_1),x)=3+20(n+1)-\ell_x$,
$f^h(i,j',a_1)$ resolves every pair $\{x,x'\}$ such that $\ell_x\neq \ell_{x'}$.

Suppose that $x\in P(s_i^j,p_i^h)$ and $y\in P(s_i^{j'},p_i^h)$ such that $j'\in [m]$ and $j'\neq j$.
We define $\ell_x$ and $\ell_y$ in a similar way to that of $\ell_x$ in the second paragraph.
$\dist(f^{mid}(i,j,3-h),x)=10(n+1)+\ell_x$ if $\ell_x\geq 1$ and $\dist(f^{mid}(i,j,3-h),x)=10(n+1)+2$ if $\ell_x=1$.
$\dist(f^{mid}(i,j,3-h),y)=\text{min }(2+|P^{3-h}(i,j,p_i^h)|/2+|P(s_i^{j'},p_i^{3-h})|+\ell_y,|P^{3-h}(i,j,p_i^h)|/2+|P(s_i^j,p_i^h)|+|P(s_i^{j'},p_i^h)|-\ell_y)=\text{min }(2+30(n+1)+\ell_y,50(n+1)-\ell_y)\geq 30(n+1)\geq \dist(f^{mid}(i,j,3-h),x)$ and the equalities hold if and only if $x=y=p_i^h$.
Thus every pair $\{x,y\}$ is resolved by $f^{mid}(i,j,3-h)$.



Suppose that $x\in P(s_i^j,p_i^h)\setminus \{s_i^j\}$ and $y\in P(s_i^j,p_i^{3-h})\setminus \{s_i^j\}$.
We define $\ell_x$ and $\ell_y$ in a similar way to that of $\ell_x$ in the second paragraph.
Then $\dist(f^h(i,j',a_1),x)=3+20(n+1)-\ell_x\leq 2+20(n+1)$ and $\dist(f^h(i,j',a_1),y)=\text{min }(1+|P^h(i,j,p_i^{3-h})|+\ell_y,3+|P(\pi_i^h,c_1)|+|P(\pi_i^{3-h},c_1)|+|P(s_i^j,p_i^{3-h})|-\ell_y)=\text{min }(1+20(n+1)+\ell_y,3+40(n+1)-\ell_y)\geq 2+20(n+1)$.
We see from the two equations that $\dist(f^h(i,j,a_1),x)\neq\dist(f^h(i,j,a_1),y)$ unless $\ell_x=\ell_y=1$.
We can check that $f^h(i,j,p_i^{3-h})$ resolves the pair $\{x,y\}$ with $\ell_x=\ell_y=1$.
Thus every pair $\{x,y\}$ is resolved by $f^h(i,j',a_1)$ or $f^h(i,j,p_i^{3-h})$.

Suppose that $x\in P(s_i^j,p_i^h)$ and $y\in P(s_i^{j'},p_i^{3-h})$.
We define $\ell_x$ and $\ell_y$ in a similar way to that of $\ell_x$ in the second paragraph.
Then $\dist(f^{mid}(i,j,h),x)=\text{min }(2+|P^h(i,j,p_i^{3-h})|/2+\ell_x,2+|P^h(i,j,p_i^{3-h})|/2+|P(s_i^j,p_i^h)|-\ell_x)=\text{min }(2+10(n+1)+\ell_x,2+30(n+1)-\ell_x)\leq 2+20(n+1)$.
$\dist(f^{mid}(i,j,h),y)=\text{min }(|P^h(i,j,p_i^{3-h})|/2+|P(s_i^j,p_i^{3-h})|+|P(s_i^{j'},p_i^{3-h})|-\ell_y,2+|P^h(i,j,p_i^{3-h})|/2+|P(s_i^{j'},p_i^h)|+\ell_y)=\text{min }(50(n+1)-\ell_y,2+30(n+1)+\ell_y)\geq 2+30(n+1)$.
Thus every pair $\{x,y\}$ is resolved by $f^{mid}(i,j,h)$.

Finally we show that every pair of distinct vertices $\{x,y\}\in \bigcup_{h\in\{1,2\}}U_i^h\times \bigcup_{h'\in\{1,2\}}U_{i'}^{h'}$ with $i,i'\in [n]$ and $i\neq i'$ is resolved by $S'$.
We fix arbitrary integers $i,i'\in [n], j,j'\in [m]$ and $h,h'\in \{1,2\}$ such that $i\neq i'$.
Let $x\in P(s_i^j,p_i^h)$ and $y\in P(s_{i'}^{j'},p_{i'}^{h'})$.
We define $\ell_x$ and $\ell_y$ in a similar way to that of $\ell_x$ in the second paragraph.
Then as we show in last paragraph, $\dist(f^{mid}(i,j,h),x)=\text{min }(2+10(n+1)+\ell_x,2+30(n+1)-\ell_x)\leq 2+20(n+1)$.
$\dist(f^{mid}(i,j,h),s_{i'}^{j'})=\text{min}_{r\in \{1,2,3\}}(1+|P^h(i,j,p_i^{3-h})|/2+|P(\pi_i^h,c_r)|+|P(c_r,s_{i'}^{j'})|)$.
$\dist(f^{mid}(i,j,h),y)=
\text{min }(\dist(f^{mid}(i,j,h),s_{i'}^{j'})+\ell_y,2+|P^h(i,j,p_i^{3-h})|/2+|P(\pi_i^h,c_1)|+|P(\pi_{i'}^{h'},c_1)|+|P(s_{i'}^{j'},p_{i'}^{h'})|-\ell_y)
> 1+30(n+1)>\dist(f^{mid}(i,j,h),x)$.
Thus every pair $\{x,y\}$ is resolved by $f^{mid}(i,j,h)$.
This completes the proof for the lemma.
\end{proof}

\begin{lemma} \label{HxH}
Every pair of distinct vertices $x,y\in\bigcup_{i\in [n],r\in\{1,2,3\}}H_{i,r}$ is resolved by $S'$.
\end{lemma}
\begin{proof}
First we show that every vertex pair of $H_{i,r}\times H_{i',r}$ such that $i',i\in [n]$ and $r\in \{1,2,3\}$ is resolved by $S'$.
We fix arbitrary integers $i\in [n],j\in [m]$ and $r\in \{1,2,3\}$.

Given two distinct vertices $x_1,x_2\in P(s_i^j,a_r)$, we can verify that $f(\pi_i^1,a_r)$ resolves the pair $\{x_1,x_2\}$.
Similarly, two distinct vertices $x'_1,x'_2\in P(s_i^j,c_r)$ are distinguished by $f(\pi_i^1,c_r)$.
Suppose that $x\in P(s_i^j,a_r)\setminus \{s_i^j\}$, $y\in P(s_i^j,b_r)\setminus \{s_i^j\}$ and $z\in P(s_i^j,c_r)\setminus \{s_i^j\}$.
Let $P(s_i^j,x)$, $P(s_i^j,y)$ and $P(s_i^j,z)$ be the subpath of $P(s_i^j,a_r)$, $P(s_i^j,b_r)$ and $P(s_i^j,c_r)$ respectively.
Let $|P(s_i^j,x)|=\ell_x$, $|P(s_i^j,y)|=\ell_y$ and $|P(s_i^j,z)|=\ell_z$.
Similarly, for a vertex $y'\in P(s_i^j,b_r)$, we define $|P(s_i^j,y')|=\ell_{y'}$.
Since $\dist(f^{mid}(i,j,1),y)=2+|P^{1}(i,j,p_i^{2})|/2+\ell_y$ and $\dist(f^{mid}(i,j,1),y')=2+|P^{1}(i,j,p_i^{2})|/2+\ell_{y'}$,
two distinct vertices $y$ and $y'$ are distinguished by $f^{mid}(i,j,1)$.
$\dist(f(\pi_i^1,a_r),x)=2+|P(s_i^j,a_r)|-\ell_x$.
$\dist(f(\pi_i^1,a_r),b_r)=\text{min}_{j'\in [m]}(2+|P(s_i^{j'},a_r)|+|P(s_i^{j'},b_r)|)=2+20(n+1)+10+20(n+1)+5$.
$\dist(f(\pi_i^1,a_r),y)=\text{min }(2+|P(s_i^j,a_r)|+\ell_y,\dist(f(\pi_i^1,a_r),b_r)+|P(s_i^j,b_r)|-\ell_y)$.
For a vertex pair $\{x,y\}$ such that $\dist(f(\pi_i^1,a_r),y)=2+|P(s_i^j,a_r)|+\ell_y$, obviously the pair is resolved by $f(\pi_i^1,a_r)$.
For a vertex pair $\{x,y\}$ such that $\dist(f(\pi_i^1,a_r),y)=\dist(f(\pi_i^1,a_r),b_r)+|P(s_i^j,b_r)|-\ell_y$,
$\dist(f(\pi_i^1,a_r),y)\geq 40(n+1)+17>\dist(f(\pi_i^1,a_r),x)$.
Thus every pair $\{x,y\}$ is resolved by $f(\pi_i^1,a_r)$.
Similarly, every pair $\{y,z\}$ is resolved by $f(\pi_i^1,c_r)$.
For a vertex pair $\{x,z\}$, $\dist(f(\pi_i^1,c_r),z)=2+|P(s_i^j,c_r)|-\ell_z<20(n+1)$ if $z\neq c_r$ and
$\dist(f(\pi_i^1,c_r),c_r)=2$.
$\dist(f(\pi_i^1,c_r),x)=\text{min }(2+|P(s_i^j,c_r)|+\ell_x,|P(\pi_i^1,c_r)|+|P(\pi_i^1,a_r)|+|P(s_i^j,a_r)|-\ell_x)\geq 2+|P(s_i^j,c_r)|>\dist(f(\pi_i^1,c_r),z)$.
Thus every pair $\{x,z\}$ is resolved by $f(\pi_i^1,c_r)$.

Let $i'\in [n],j'\in [m]$ be integers such that $i'\neq i$ or $j'\neq j$.
Suppose that $x\in P(s_i^j,a_r)\setminus \{a_r\}$, $y\in P(s_i^j,b_r)\setminus \{b_r\}$ and $z\in P(s_i^j,c_r)\setminus \{c_r\}$.
Suppose that $x'\in P(s_{i'}^{j'},a_r)\setminus \{a_r\}$, $y'\in P(s_{i'}^{j'},b_r)\setminus \{b_r\}$
and $z'\in P(s_{i'}^{j'},c_r)\setminus \{c_r\}$.
We define $\ell_x$, $\ell_y$, $\ell_z$, $\ell_{x'}$, $\ell_{y'}$ and $\ell_{z'}$ in a similar way to that of $\ell_{x}$ in the second paragraph.
For a pair $\{x,x'\}$, since $\dist(f(\pi_i^1,a_r),x)=2+|P(s_i^j,a_r)|-\ell_x$ and $\dist(f(\pi_i^1,a_r),x')=2+|P(s_{i'}^{j'},a_r)|-\ell_{x'}$,
$f(\pi_i^1,a_r)$ resolves every pair $\{x,x'\}$ such that $|P(s_i^j,a_r)|-\ell_x \neq |P(s_{i'}^{j'},a_r)|-\ell_{x'}$.
Since $\dist(f(s_i^j,a_r),x)=|P(s_i^j,a_r)|-\ell_x$ and $\dist(f(s_i^j,a_r),x')=2+|P(s_{i'}^{j'},a_r)|-\ell_{x'}$,
$f(s_i^j,a_r)$ resolves every pair $\{x,x'\}$ such that $|P(s_i^j,a_r)|-\ell_x=|P(s_{i'}^{j'},a_r)|-\ell_{x'}$.
As a result, every pair $\{x,x'\}$ is resolved by $f(\pi_i^1,a_r)$ or $f(s_i^j,a_r)$.
Similarly, every pair $\{z,z'\}$ is resolved by $f(\pi_i^1,c_r)$ or $f(s_i^j,c_r)$.
For a pair $\{y,y'\}$, there are two cases.
Case 1: $i=i'$ and $j\neq j'$.
$\dist(f^{mid}(i,j,1),y)=2+|P^{1}(i,j,p_i^{2})|/2+\ell_y$.
$\dist(f^{mid}(i,j,1),y')=\text{min }(2+|P^{1}(i,j,p_i^{2})|/2+|P(s_i^j,b_r)|+|P(s_{i'}^{j'},b_r)|-\ell_{y'},1+|P^{1}(i,j,p_i^{2})|/2+|P^1(i',j',b_r)|+\ell_{y'}-1)$ if $y'\neq s_{i'}^{j'}$ and $\dist(f^{mid}(i,j,1),s_{i'}^{j'})=2+30(n+1)$.
If a pair $\{y,y'\}$ satisfies that $\dist(f^{mid}(i,j,1),y')=2+|P^{1}(i,j,p_i^{2})|/2+|P(s_i^j,b_r)|+|P(s_{i'}^{j'},b_r)|-\ell_{y'}$,
obviously $f^{mid}(i,j,1)$ resolves this pair.
Thus if a pair $\{y,y'\}$ is not resolved by $f^{mid}(i,j,1)$,
then we have $\dist(f^{mid}(i,j,1),y)=2+10(n+1)+\ell_y=\dist(f^{mid}(i,j,1),y')=30(n+1)+\ell_{y'}$, i.e. $\ell_y-\ell_{y'}=20(n+1)-2$.
For this pair,
$\dist(f^{mid}(i',j',1),y')=2+10(n+1)+\ell_{y'}<\dist(f^{mid}(i',j',1),y)=\text{min }(2+10(n+1)+|P(s_{i'}^{j'},b_r)|+|P(s_i^j,b_r)|-\ell_y,30(n+1)+\ell_y)$.
Thus in this case, every pair $\{y,y'\}$ is resolved by $f^{mid}(i,j,1)$ or $f^{mid}(i',j',1)$.
Case 2: $i\neq i'$.
$\dist(f^{mid}(i,j,1),y)=2+|P^{1}(i,j,p_i^{2})|/2+\ell_y$.
$\dist(f^{mid}(i,j,1),s_{i'}^{j'})=1+|P^{1}(i,j,p_i^{2})|/2+\text{min}_{r'\in \{1,2,3\}}(|P(\pi_i^1,c_{r'})|+|P(s_{i'}^{j'},c_{r'})|)$.
$\dist(f^{mid}(i,j,1),y')=\text{min }(2+|P^{1}(i,j,p_i^{2})|/2+|P(s_i^j,b_r)|+|P(s_{i'}^{j'},b_r)|-\ell_{y'},\dist(f^{mid}(i,j,1),s_{i'}^{j'})+\ell_{y'})$.
If a pair $\{y,y'\}$ satisfies that $\dist(f^{mid}(i,j,1),y')=2+|P^{1}(i,j,p_i^{2})|/2+|P(s_i^j,b_r)|+|P(s_{i'}^{j'},b_r)|-\ell_{y'}$,
obviously $f^{mid}(i,j,1)$ resolves this pair.
Thus if a pair $\{y,y'\}$ is not resolved by $f^{mid}(i,j,1)$,
then we have $\dist(f^{mid}(i,j,1),y)=2+10(n+1)+\ell_y=\dist(f^{mid}(i,j,1),y')=\dist(f^{mid}(i,j,1),s_{i'}^{j'})+\ell_{y'}$,
i.e. $\ell_y-\ell_{y'}=\dist(\pi_i^1,s_{i'}^{j'})-1$.
For this pair,
$\dist(f^{mid}(i',j',1),y')=2+10(n+1)+\ell_{y'}<\dist(f^{mid}(i',j',1),y)=\text{min }(2+10(n+1)+|P(s_{i'}^{j'},b_r)|+|P(s_i^j,b_r)|-\ell_y, \dist(f^{mid}(i',j',1),s_i^j)+\ell_y)$.
Thus in this case, every pair $\{y,y'\}$ is resolved by $f^{mid}(i,j,1)$ or $f^{mid}(i',j',1)$.
It follows that every pair $\{y,y'\}$ is resolved by $f^{mid}(i,j,1)$ or $f^{mid}(i',j',1)$.
For a pair $\{x,y'\}$, there are two cases.
Case 1: $|P(s_i^j,a_r)|>20(n+1)+10\cdot 1=\min_{\alpha\in [m]}|P(s_i^{\alpha},a_r)|$.
Then $\dist(f(\pi_i^1,a_r),y')=\text{min }(2+|P(s_{i'}^{j'},a_r)|+\ell_{y'},2+20(n+1)+10\cdot 1+20(n+1)+5\cdot 1+1+|P(s_{i'}^{j'},b_r)|-\ell_{y'})=\dist(f(s_i^j,a_r),y')$
and $\dist(f(\pi_i^1,a_r),x)=2+|P(s_i^j,a_r)|-\ell_x=\dist(f(s_i^j,a_r),x)+2$.
In this case, $\{x,y'\}$ is resolved by $f(\pi_i^1,a_r)$ or $f(s_i^j,a_r)$.
Case 2: $|P(s_i^j,a_r)|=20(n+1)+10\cdot 1=\min_{p\in [m]}|P(s_i^p,a_r)|$.
Then $\dist(f(s_i^j,a_r),x)\leq 20(n+1)+10<\dist(f(s_i^j,a_r),y')$.
Thus in this case, $\{x,y'\}$ is resolved by $f(s_i^j,a_r)$.
It follows that every pair $\{x,y'\}$ is resolved by $f(\pi_i^1,a_r)$ or $f(s_i^j,a_r)$.
For a pair $\{x,z'\}$, $\dist(f(\pi_i^1,a_r),x)=2+|P(s_i^j,a_r)|-\ell_x=\dist(f(s_i^j,a_r),x)+2$, and $\dist(f(\pi_i^1,a_r),z')=\min{(2+|P(s_{i'}^{j'},a_r)|+\ell_{z'}, |P(\pi_i^1,a_r)|+|P(\pi_i^1,c_r)|+|P(s_{i'}^{j'},c_r)|-\ell_{z'})}\leq \dist(f(s_i^j,a_r),z')=\text{min }(2+|P(s_{i'}^{j'},a_r)|+\ell_{z'},2+|P(\pi_i^1,a_r)|+|P(\pi_i^1,c_r)|+|P(s_{i'}^{j'},c_r)|-\ell_{z'})$.
It follows that every pair $\{x,z'\}$ is resolved by $f(\pi_i^1,a_r)$ or $f(s_i^j,a_r)$.
For a pair $\{y,z'\}$, there are two cases.
Case 1: $|P(s_{i'}^{j'},c_r)|+|P(s_{i'}^{j'},b_r)|>20(n+1)-10n+20(n+1)+5n+1=\min_{\alpha\in [m]}{(|P(s_{i'}^{\alpha},c_r)|+|P(s_{i'}^{\alpha},b_r)|)}$.
Then $\dist(f(\pi_i^1,c_r),y)=\text{min }(2+|P(s_i^j,c_r)|+\ell_y,3+40(n+1)-5n+|P(s_i^j,b_r)|-\ell_y)=\dist(f(s_{i'}^{j'},c_r),y)$ and $\dist(f(\pi_i^1,c_r),z')=2+|P(s_{i'}^{j'},c_r)|-\ell_{z'}=\dist(f(s_{i'}^{j'},c_r),z')+2$.
In this case, $\{y,z'\}$ is resolved by $f(\pi_i^1,c_r)$ or $f(s_{i'}^{j'},c_r)$.
Case 2: $|P(s_{i'}^{j'},c_r)|+|P(s_{i'}^{j'},b_r)|=20(n+1)-10n+20(n+1)+5n+1$.
Then $\dist(f(s_{i'}^{j'},c_r),y)\geq 2+20(n+1)-10n>\dist(f(s_{i'}^{j'},c_r),z')$.
It follows that every pair $\{y,z'\}$ is resolved by $f(\pi_i^1,c_r)$ or $f(s_{i'}^{j'},c_r)$.
This completes the proof to show that every vertex pair of $H_{i,r}\times H_{i',r}$ such that $i',i\in [n]$ and $r\in \{1,2,3\}$ is resolved by $S'$.

Next we show that every vertex pair of $H_{i,r}\times H_{i',r'}$ such that $i,i'\in [n]$, $r,r'\in \{1,2,3\}$ and $r\neq r'$ is resolved by $S'$.
We fix arbitrary integers $i,i'\in [n],j,j'\in [m]$ and $r,r'\in \{1,2,3\}$ such that $r\neq r'$.
Suppose that $x\in P(s_i^j,a_r)$, $y\in P(s_i^j,b_r)$, $z\in P(s_i^j,c_r)$, $x'\in P(s_{i'}^{j'},a_{r'})$, $y'\in P(s_{i'}^{j'},b_{r'})$ and $z'\in P(s_{i'}^{j'},c_{r'})$.
We define $\ell_x$, $\ell_y$, $\ell_z$, $\ell_{x'}$, $\ell_{y'}$ and $\ell_{z'}$ in a similar way to that of $\ell_{x}$ in the second paragraph.
For a vertex pair $\{x,x'\}$, there are two cases.
Case 1: $i=i'$ and $j=j'$.
Then $\dist(f(\pi_i^1,a_r),x)=2+|P(s_i^j,a_r)|-\ell_x$,
$\dist(f(\pi_i^1,a_r),x')=\text{min }(2+|P(s_{i'}^{j'},a_r)|+\ell_{x'},|P(\pi_i^1,a_r)|+|P(\pi_i^1,a_{r'})|+|P(s_{i'}^{j'},a_r)|-\ell_{x'})$.
$\dist(f(s_i^j,a_r),x)=|P(s_i^j,a_r)|-\ell_x$ when $x\neq a_r$ and $\dist(f(s_i^j,a_r),a_r)=1$.
$\dist(f(s_i^j,a_r),x')=\text{min }(|P(s_i^j,a_r)|+\ell_{x'},2+|P(\pi_i^1,a_r)|+|P(\pi_i^1,a_{r'})|+|P(s_i^j,a_{r'})|-\ell_{x'})$.
Thus for the vertex pair $\{x,x'\}$ which is not resolved by $f(s_i^j,a_r)$, i.e.,$\dist(f(s_i^j,a_r),x)=\dist(f(s_i^j,a_r),x')=2+|P(\pi_i^1,a_r)|+|P(\pi_i^1,a_{r'})|+|P(s_i^j,a_{r'})|-\ell_{x'}$, it satisfies that $\dist(f(\pi_i^1,a_r),x)>\dist(f(s_i^j,a_r),x)$ and $\dist(f(\pi_i^1,a_r),x')<\dist(f(s_i^j,a_r),x')$.
Thus in this case, every pair every pair $\{x,x'\}$ is resolved by $f(s_i^j,a_r)$ or $f(\pi_i^1,a_r)$.
Case 2: $i\neq i'$ or $j\neq j'$.
$\dist(f(s_i^j,a_r),x)=|P(s_i^j,a_r)|-\ell_x$ when $x\neq a_r$ and $\dist(f(s_i^j,a_r),a_r)=1$.
$\dist(f(s_i^j,a_r),x')=\text{min }(2+|P(s_{i'}^{j'},a_r)|+\ell_{x'},2+|P(\pi_i^1,a_r)|+|P(\pi_i^1,a_{r'})|+|P(s_{i'}^{j'},a_{r'})|-\ell_{x'})$.
For the vertex pair $\{x,x'\}$ which is not resolved by $f(s_i^j,a_r)$, i.e., $\dist(f(s_i^j,a_r),x)=\dist(f(s_i^j,a_r),x')$, it satisfies that $\dist(f(\pi_i^1,a_r),x)>\dist(f(s_i^j,a_r),x)$ and $\dist(f(\pi_i^1,a_r),x')\leq \dist(f(s_i^j,a_r),x')$.
Thus in this case, every pair every pair $\{x,x'\}$ is resolved by $f(s_i^j,a_r)$ or $f(\pi_i^1,a_r)$.
It follows that every pair $\{x,x'\}$ is resolved by $f(s_i^j,a_r)$ or $f(\pi_i^1,a_r)$.
For a vertex pair $\{z,z'\}$, similarly it is resolved by $f(s_i^j,c_r)$ or $f(\pi_i^1,c_r)$.
For a vertex pair $\{y,y'\}$, let $\{u_r^{i_r},v_r^{i_r}\}$ be the vertex pair resolved by $s_i^j$, i.e. $|P(s_i^j,b_r)|=20(n+1)+5i_r+1$.
Then $\dist(f^1(u_r^{i_r},v_r^{i_r}),y)=40(n+1)+1-\ell_y=\dist(f^2(u_r^{i_r},v_r^{i_r}),y)$ when $y\neq s_i^j$.
We observe that there is a shortest path from $f^h(u_r^{i_r},v_r^{i_r})$ ($h\in\{1,2\}$) to $y'$ which
either goes through one vertex of $\{a_r,c_r\}$, then goes through $s_{i'}^{j'}$, finally reaches $y'$
or goes through one vertex of $\{a_r,c_r\}$, then goes through some vertex $s_{i''}^{j''}$ ($i''\in [n],j''\in [m]$), then goes through $b_{r'}$, finally reaches $y'$. Thus we get that $\dist(f^1(u_r^{i_r},v_r^{i_r}),y')=\dist(f^2(u_r^{i_r},v_r^{i_r}),y')+1$.
Thus every vertex pair $\{y,y'\}$ such that $y\neq s_i^j$ is resolved by $f^1(u_r^{i_r},v_r^{i_r})$ or $f^2(u_r^{i_r},v_r^{i_r})$.
For a pair $\{s_i^j,y'\}$, obviously $\dist(f^{mid}(i,j,1),s_i^j)<dist(f^{mid}(i,j,1),y')$.
It follows that every vertex pair $\{y,y'\}$ is resolved by $f^1(u_r^{i_r},v_r^{i_r})$, $f^2(u_r^{i_r},v_r^{i_r})$ or $f^{mid}(i,j,1)$.
For a vertex pair $\{x,y'\}$, there are two cases.
Case 1: $i\neq i'$ or $j\neq j'$.
$\dist(f(\pi_i^1,a_r),x)=2+|P(s_i^j,a_r)|-\ell_x$.
$\dist(f(\pi_i^1,a_r),b_{r'})=2+\text{min}_{\alpha\in [m]} (|P(s_i^{\alpha},a_r)|+|P(s_i^{\alpha},b_{r'})|)$.
Then $\dist(f(\pi_i^1,a_r),y')=\text{min }(2+|P(s_{i'}^{j'},a_r)|+\ell_{y'},\dist(f(\pi_i^1,a_r),b_{r'})+|P(s_{i'}^{j'},b_{r'})|-\ell_{y'})$.
$\dist(f(s_i^j,a_r),x)=|P(s_i^j,a_r)|-\ell_x<\dist(f(\pi_i^1,a_r),x)$ if $x\neq a_r$ and $\dist(f(s_i^j,a_r),a_r)=2$.
$\dist(f(s_i^j,a_r),y')=\dist(f(\pi_i^1,a_r),y')$.
Thus in this case, every pair $\{x,y'\}$ is resolved by $f(\pi_i^1,a_r)$ or $f(s_i^j,a_r)$.
Case 2: $i=i'$ and $j=j'$.
$\dist(f(s_i^j,a_r),y')=\text{min }(|P(s_i^j,a_r)|+\ell_{y'},\dist(f(\pi_i^1,a_r),b_{r'})+|P(s_{i'}^{j'},b_{r'})|-\ell_{y'})$.
If $\dist(f(s_i^j,a_r),y')=|P(s_i^j,a_r)|+\ell_{y'}$, then $\{x,y'\}$ is obviously resolved by $f(s_i^j,a_r)$.
Otherwise, for a vertex pair $\{x,y'\}$ which is not resolved by $f(s_i^j,a_r)$,
$\dist(f(s_i^j,a_r),x)=\dist(f(s_i^j,a_r),y')=\dist(f(\pi_i^1,a_r),y')<\dist(f(\pi_i^1,a_r),x)$.
Thus in this case, every pair $\{x,y'\}$ is resolved by $f(\pi_i^1,a_r)$ or $f(s_i^j,a_r)$.
It follows that every pair $\{x,y'\}$ is resolved by $f(\pi_i^1,a_r)$ or $f(s_i^j,a_r)$.
For a vertex pair $\{y,z'\}$, similarly we can show that every pair $\{y,z'\}$ is resolved by $f(\pi_i^1,c_r)$ or $f(s_i^j,c_r)$.
For a vertex pair $\{x,z'\}$, there are two cases.
Case 1: $i\neq i'$ or $j\neq j'$.
$\dist(f(\pi_i^1,a_r),x)=2+|P(s_i^j,a_r)|-\ell_x>\dist(f(s_i^j,a_r),x)$ if $x\neq a_r$.
$\dist(f(\pi_i^1,a_r),z')=\text{min }(2+|P(s_{i'}^{j'},a_r)|+\ell_{z'},|P(\pi_i^1,a_r)|+|P(\pi_i^1,c_{r'})|+|P(s_{i'}^{j'},c_{r'})|-\ell_{z'})$.
$\dist(f(s_i^j,a_r),z')=\text{min }(2+|P(s_{i'}^{j'},a_r)|+\ell_{z'},2+|P(\pi_i^1,a_r)|+|P(\pi_i^1,c_{r'})|+|P(s_{i'}^{j'},c_{r'})|-\ell_{z'})\geq \dist(f(\pi_i^1,a_r),z')$.
Thus in this case, every pair $\{x,z'\}$ is resolved by $f(\pi_i^1,a_r)$ or $f(s_i^j,a_r)$.
Case 2: $i=i'$ and $j=j'$.
$\dist(f(s_i^j,a_r),z')=\text{min }(|P(s_i^j,a_r)|+\ell_{z'},2+|P(\pi_i^1,a_r)|+|P(\pi_i^1,c_{r'})|+|P(s_{i'}^{j'},c_{r'})|-\ell_{z'})$.
If $\dist(f(s_i^j,a_r),z')=|P(s_i^j,a_r)|+\ell_{z'}$
Then $\{x,z'\}$ ($x\neq a_r$) is obviously resolved by $f(s_i^j,a_r)$.
Otherwise, for a vertex pair $\{x,z'\}$ which is not resolved by $f(s_i^j,a_r)$,
$\dist(f(\pi_i^1,a_r),x)>\dist(f(s_i^j,a_r),x)=\dist(f(s_i^j,a_r),z')>\dist(f(\pi_i^1,a_r),z')$.
Thus every pair $\{x,z'\}$ is resolved by $f(\pi_i^1,a_r)$ or $f(s_i^j,a_r)$.
It follows that every pair $\{x,z'\}$ is resolved by $f(\pi_i^1,a_r)$ or $f(s_i^j,a_r)$.
This completes the proof for the lemma.
\end{proof}

\begin{lemma} \label{SxS}
Every pair of distinct vertices $x,y\in\bigcup_{i\in [n],h\in\{1,2\}}S_i^h$ is resolved by $S'$.
\end{lemma}
\begin{proof}
Let $x\in P(\pi_i^h,a_r)$ for arbitrary integers $i\in [n],h\in \{1,2\},r\in \{1,2,3\}$.
Let $x'\in P(\pi_{i'}^{h'},a_r)$ for arbitrary integers $i'\in [n],h'\in \{1,2\}$.
We fix an arbitrary integer $j\in [m]$.
Let $P(x,a_r)$ be the subpath of $P(\pi_i^h,a_r)$ and $\ell_x=|P(x,a_r)|$.
Let $P(x',a_{r})$ be the subpath of $P(\pi_{i'}^{h'},a_{r})$ and $\ell_{x'}=|P(x',a_{r})|$.
For a vertex pair $\{x,x'\}$, $\dist(f(s_i^j,a_r),x)=2+\ell_x$ and $\dist(f(s_i^j,a_r),x')=2+\ell_{x'}$.
Then every vertex pair $\{x,x'\}$ such that $\ell_x\neq \ell_{x'}$ is resolved by $f(s_i^j,a_r)$.
Since $\dist(f(\pi_i^h,a_r),x)=\ell_x$ if $x\neq a_r$ and $\dist(f(\pi_i^h,a_r),x')=2+\ell_{x'}$ if $i\neq i'$ or $h\neq h'$,
every vertex pair $\{x,x'\}$ such that $\ell_x=\ell_{x'}$ is resolved by $f(\pi_i^h,a_r)$.
It follows that every vertex pair $\{x,x'\}$ is resolved by $f(\pi_i^h,a_r)$ or $f(s_i^j,a_r)$.
Let $y\in P(\pi_i^h,c_r)$ for arbitrary integers $i\in [n],h\in \{1,2\},r\in \{1,2,3\}$.
Let $y'\in P(\pi_{i'}^{h'},c_r)$ for arbitrary integers $i'\in [n],h'\in \{1,2\}$.
Similarly, we can show that every vertex pair $\{y,y'\}$ is resolved by $f(\pi_i^h,c_r)$ or $f(s_i^j,c_r)$.

Let $x_1\in P(\pi_i^h,a_r)$ for arbitrary integers $i\in [n],h\in \{1,2\},r\in \{1,2,3\}$.
Let $x_2\in P(\pi_{i'}^{h'},a_{r'})$ for arbitrary integers $i'\in [n],h'\in \{1,2\},r'\in \{1,2,3\}$.
We fix an arbitrary integer $j\in [m]$.
We define $\ell_{x_1}$ and $\ell_{x_2}$ in a similar way to that of $\ell_x$ in the first paragraph.
For a vertex pair $\{x_1,x_2\}$ such that $r\neq r'$, $\dist(f(s_i^j,a_r),x_1)=2+\ell_{x_1}$ and $\dist(f(s_i^j,a_r),x_2)=2+|P(\pi_{i'}^{h'},a_r)|+|P(\pi_{i'}^{h'},a_{r'})|-\ell_{x_2}=2+20(n+1)-\ell_{x_2}>\dist(f(s_i^j,a_r),x_1)$ unless $x_1=\pi_i^h, x_2=\pi_{i'}^{h'}$ and $\pi_i^h\neq \pi_{i'}^{h'}$.
The vertex pair $\{\pi_i^h,\pi_{i'}^{h'}\}$ is obviously resolved by $f^h(i,j,a_r)$.
Thus every vertex pair $\{x_1,x_2\}$ such that $r\neq r'$ is resolved by $f(s_i^j,a_r)$ or $f^h(i,j,a_r)$.
Let $y_1\in P(\pi_i^h,c_r)$ for arbitrary integers $i\in [n],h\in \{1,2\},r\in \{1,2,3\}$.
Let $y_2\in P(\pi_{i'}^{h'},c_{r'})$ for arbitrary integers $i'\in [n],h'\in \{1,2\},r'\in \{1,2,3\}$ such that $r\neq r'$.
We define $\ell_{y_1}$ and $\ell_{y_2}$ in a similar way to that of $\ell_x$ in last paragraph.
Similarly, we can show that every vertex pair $\{y_1,y_2\}$ such that $r\neq r'$ is resolved by $f(s_i^j,c_r)$ or $f^h(i,j,a_r)$.
For a pair $\{x_1,y_2\}$, $\dist(f(s_i^j,a_r),x_1)=2+\ell_{x_1}$ and $\dist(f(s_i^j,a_r),y_2)=2+|P(\pi_{i'}^{h'},a_r)|+|P(\pi_{i'}^{h'},c_{r'})|-\ell_{y_2}=2+20(n+1)-\ell_{y_2}>\dist(f(s_i^j,a_r),x_1)$ unless $x_1=\pi_i^h, y_2=\pi_{i'}^{h'}$ and $\pi_i^h\neq \pi_{i'}^{h'}$.
The vertex pair $\{\pi_i^h,\pi_{i'}^{h'}\}$ is obviously resolved by $f^h(i,j,a_r)$.
Thus every vertex pair $\{x_1,y_2\}$ is resolved by $f(s_i^j,a_r)$ or $f^h(i,j,a_r)$.
This completes the proof for the lemma.
\end{proof}

\begin{lemma} \label{PixPi}
Every pair of distinct vertices $x,y\in\bigcup_{i\in [n],h\in\{1,2\},j\in [m],r\in \{1,2,3\}}\Pi^h(i,j,r)$ is resolved by $S'$.
\end{lemma}
\begin{proof}
Let $x_1,x_2\in P^{h}(i,j,a_r)$ be two distinct vertices for arbitrary integers $i\in [n],j\in [m],h\in \{1,2\},r\in \{1,2,3\}$.
Let $j'\in [m]$ be an integer such that $j\neq j'$.
Obviously the pair $\{x_1,x_2\}$ is resolved by $f^{h}(i,j',a_r)$.
Similarly, the vertex pair $\{y_1,y_2\}$ for two distinct vertices $y_1,y_2\in P^{h}(i,j,b_r)$,
the vertex pair $\{z_1,z_2\}$ for two distinct vertices $z_1,z_2\in P^{h}(i,j,c_r)$, and
the vertex pair $\{w_1,w_2\}$ for two distinct vertices $w_1,w_2\in P^{h}(i,j,p_i^{3-h})$ are resolved by $f^{h}(i,j',a_r)$.

Let $x\in P^{h}(i,j,a_r),y\in P^{h}(i,j,b_r),z\in P^{h}(i,j,c_r)$ and $w\in P^{h}(i,j,p_i^{3-h})$ for arbitrary integers $i\in [n],j\in [m],h\in \{1,2\},r\in \{1,2,3\}$.
Let $x'\in P^{h'}(i',j',a_{r'}),y'\in P^{h'}(i',j',b_{r'}),z'\in P^{h'}(i',j',c_{r'})$ and $w'\in P^{h'}(i',j',p_{i'}^{3-h'})$ for arbitrary integers $i'\in [n],j'\in [m],h'\in \{1,2\},r'\in \{1,2,3\}$.
We define $\ell_x=\dist(x,\pi_i^h)$. In a similar way, we define $\ell_y,\ell_z,\ell_w,\ell_{x'},\ell_{y'},\ell_{z'},\ell_{w'}$.
For a pair $\{x,y\}$, $\dist(f^{h}(i,j,c_r),x)=2+\ell_x$ and $\dist(f^{h}(i,j,c_r),y)=2+\ell_y$.
Thus $f^{h}(i,j,c_r)$ resolves every pair $\{x,y\}$ such that $\ell_x\neq \ell_y$.
Since $\dist(f^{h}(i,j,a_r),x)=\ell_x$ if $x\neq \pi_i^h$, $\dist(f^{h}(i,j,a_r),\pi_i^h)=2$ and $\dist(f^{h}(i,j,c_r),y)=2+\ell_y$, $f^{h}(i,j,a_r)$ resolves every pair $\{x,y\}$ such that $\ell_x=\ell_y$.
Thus every pair $\{x,y\}$ is resolved by $f^{h}(i,j,a_r)$ or $f^{h}(i,j,c_r)$.
In a similar way, we can show that two distinct vertices from $\bigcup_{j\in [m],r\in \{1,2,3\}}\Pi^h(i,j,r)$ are distinguished by $S'$.
For a pair $\{x,y'\}$ with $i=i'$ and $h\neq h'$, $\dist(f^{h}(i,j,c_r),x)=2+\ell_x<\dist(f^{h}(i,j,c_r),y')=\text{min }(2+|P(\pi_i^h,a_r)|+|P(\pi_i^{h'},a_r)|+\ell_{y'},2+|P(\pi_i^h(i,j',b_{r'}))|+|P(\pi_i^{h'}(i,j',b_{r'})|-\ell_{y'})$.
Thus every pair $\{x,y'\}$ with $i=i'$ and $h\neq h'$ is resolved by $f^{h}(i,j,c_r)$.
Similarly we can show that every vertex pair of $\bigcup_{j\in [m],r\in \{1,2,3\}}\Pi^h(i,j,r)\times \bigcup_{j\in [m],r\in \{1,2,3\}}\Pi^{3-h}(i,j,r)$ is resolved by $S'$.
For a pair $\{x,y'\}$ with $i\neq i'$, $\dist(f^{h}(i,j,c_r),x)=2+\ell_x$.
$\dist(f^{h}(i,j,c_r),s_{i'}^{j'})=\text{min}_{d\in \{1,2,3\}}(2+|P(\pi_i^h,c_d)|+|P(s_{i'}^{j'},c_d)|)$.
Thus $\dist(f^{h}(i,j,c_r),y')=\text{min }(2+|P(\pi_i^h,a_r)|+|P(\pi_i^{h'},a_r)|+\ell_{y'},\dist(f^{h}(i,j,c_r),s_{i'}^{j'})+1+|P^{h'}(i',j',b_{r'})|-\ell_{y'})$.
We can see that $\dist(f^{h}(i,j,c_r),x)<\dist(f^{h}(i,j,c_r),y')$ unless $\ell_x=20(n+1)$ and $\ell_{y'}=0$.
If $\ell_x=20(n+1)$ and $\ell_{y'}=0$, $\dist(f^{h}(i,j,a_r),y')=2+20(n+1)>\dist(f^{h}(i,j,a_r),x)=20(n+1)$.
Thus every pair $\{x,y'\}$ with $i\neq i'$ is resolved by $f^{h}(i,j,c_r)$ or $f^{h}(i,j,a_r)$.
Similarly we can show that every vertex pair of $\bigcup_{j\in [m],r\in \{1,2,3\}}\Pi^h(i,j,r)\times \bigcup_{j'\in [m],r'\in \{1,2,3\}}\Pi^{h'}(i',j',r')$ with $i\neq i'$ is resolved by $S'$.
This completes the proof for the lemma.
\end{proof}

\begin{lemma} \label{LxL}
Every pair of distinct vertices $x,y\in\bigcup_{i\in [n],h\in\{1,2\}}L_i^h$ is resolved by $S'$.
\end{lemma}
\begin{proof}
First we show that every vertex pair of $L_i^h\times L_i^h$ is resolved by $S'$ for $i\in [n],h\in \{1,2\}$.
We fix arbitrary integers $i\in [n],j\in [m]$ and $h\in \{1,2\}$.
For a vertex $x\in P(q_i^h,\text{mid}(P^{3-h}(i,j,p_i^{h})))$,
let $P(q_i^h,x)$ be the subpath of $P(q_i^h,\text{mid}(P^{3-h}(i,j,p_i^{h})))$ from $q_i^h$ to $x$ and
let $|P(q_i^h,x)|=\ell_x$.
For two distinct vertices $x_1,x_2\in P(q_i^h,\text{mid}(P^{3-h}(i,j,p_i^{h})))$,
$\dist(f^{mid}(i,j,3-h),x_1)=1+|P(q_i^h,\text{mid}(P^{3-h}(i,j,p_i^{h})))|-\ell_{x_1}=30(n+1)-\ell_{x_1}$ and
$\dist(f^{mid}(i,j,3-h),x_2)=30(n+1)-\ell_{x_2}$.
Thus $f^{mid}(i,j,3-h)$ resolves every pair $\{x_1,x_2\}$.
Let $x\in P(q_i^h,\text{mid}(P^{3-h}(i,j,p_i^{h})))$ and $x'\in P(q_i^h,\text{mid}(P^{3-h}(i,j',p_i^{h})))$ with some integer $j'\neq j$.
$\dist(f^h(i,j,a_1),x)=3+\ell_x$ and $\dist(f^h(i,j,a_1),x')=3+\ell_{x'}$.
Thus $f^h(i,j,a_1)$ resolves every pair $\{x,x'\}$ such that $\ell_x\neq \ell_{x'}$.
For a pair $\{x,x'\}$ such that $\ell_x=\ell_{x'}$,
$\dist(f^{mid}(i,j,3-h),x)=30(n+1)-\ell_{x}$ and
$\dist(f^{mid}(i,j,3-h),x')=\text{min }(1+|P(q_i^h,\text{mid}(P^{3-h}(i,j,p_i^{h})))|+\ell_{x'},1+|P^{3-h}(i,j,p_i^h)|/2+|P^{3-h}(i,j',p_i^h)|/2+|P(q_i^h,\text{mid}(P^{3-h}(i,j',p_i^{h})))|-\ell_{x'})=\text{min }(30(n+1)+\ell_{x'},50(n+1)-\ell_{x'})$.
Thus $\dist(f^{mid}(i,j,3-h),x)\neq \dist(f^{mid}(i,j,3-h),x')$ and $f^{mid}(i,j,3-h)$ resolves this pair.
It follows that every pair $\{x,x'\}$ is resolved by $f^{mid}(i,j,3-h)$ or $f^h(i,j,a_1)$.

Next we show that every vertex pair of $L_i^h\times L_i^{3-h}$ is resolved by $S'$ for $i\in [n],h\in \{1,2\}$.
We fix arbitrary integers $i\in [n],j\in [m]$ and $h\in \{1,2\}$.
Let $x\in P(q_i^h,\text{mid}(P^{3-h}(i,j,p_i^{h})))$ and $y\in P(q_i^{3-h},\text{mid}(P^{h}(i,j,p_i^{3-h})))$.
We define $\ell_x$ and $\ell_y$ in a similar way to that of $\ell_x$ in last paragraph.
For a pair $\{x,y\}$, $\dist(f^{mid}(i,j,3-h),x)=30(n+1)-\ell_{x}$ and
$\dist(f^{mid}(i,j,3-h),y)=\text{min }(2+|P^{3-h}(i,j,p_i^h)|/2+\ell_y,3+|P^{3-h}(i,j,p_i^h)|/2+|P^h(i,j,p_i^{3-h})|/2+|P(q_i^{3-h},\text{mid}(P^h(i,j,p_i^{3-h})))|-\ell_y)$.
For a pair $\{x,y\}$ which is not resolved by $f^{mid}(i,j,3-h)$, there are two cases.
Case 1: $\dist(f^{mid}(i,j,3-h),y)=2+|P^{3-h}(i,j,p_i^h)|/2+\ell_y=2+10(n+1)+\ell_y\leq 2+50(n+1)-\ell_y$ when $\ell_y\leq 20(n+1)$.
We have $\dist(f^{mid}(i,j,3-h),x)=30(n+1)-\ell_{x}=\dist(f^{mid}(i,j,3-h),y)=2+10(n+1)+\ell_y$, i.e. $\ell_x+\ell_y=20(n+1)-2$.
Case 2: $\dist(f^{mid}(i,j,3-h),y)=2+50(n+1)-\ell_y$ when $\ell_y>20(n+1)$.
We have $\dist(f^{mid}(i,j,3-h),x)=30(n+1)-\ell_{x}=\dist(f^{mid}(i,j,3-h),y)=2+50(n+1)-\ell_y$, i.e. $\ell_y-\ell_x=20(n+1)+2$.
$\dist(f^{3-h}(i,j,a_1),y)=3+\ell_y$.
$\dist(f^{3-h}(i,j,a_1),x)=3+|P(\pi_i^{3-h},a_1)|+|P(\pi_i^{h},a_1)|+\ell_x$ when $\ell_x\leq 10(n+1)$ and
$\dist(f^{3-h}(i,j,a_1),x)=2+|P^{3-h}(i,j,p_i^h)|/2+|P(q_i^{h},\text{mid}(P^{3-h}(i,j,p_i^{h})))|-\ell_x$ when $\ell_x>10(n+1)$.
For a pair $\{x,y\}$ which is not resolved by $f^{3-h}(i,j,a_1)$, there are two cases.
Case 1: $\dist(f^{3-h}(i,j,a_1),y)=\dist(f^{3-h}(i,j,a_1),x)$ when $\ell_x\leq 10(n+1)$, i.e. $\ell_y-\ell_x=3+20(n+1)$.
Case 2: $\dist(f^{3-h}(i,j,a_1),y)=\dist(f^{3-h}(i,j,a_1),x)$ when $\ell_x>10(n+1)$, i.e. $\ell_y+\ell_x=40(n+1)-2$.
Thus we see that if a pair $\{x,y\}$ is not resolved by $f^{mid}(i,j,3-h)$, then it is resolved by $f^{3-h}(i,j,a_1)$.
It follows that every pair $\{x,y\}$ is resolved by $f^{mid}(i,j,3-h)$ and $f^{3-h}(i,j,a_1)$.
Let $x'\in P(q_i^{h},\text{mid}(P^{3-h}(i,j',p_i^{h})))$ with some integer $j'\in [m]$ and $j'\neq j$.
Then $\dist(f^{mid}(i,j',3-h),x')=30(n+1)-\ell_{x'}$ and $\dist(f^{mid}(i,j',3-h),y)=2+|P^{3-h}(i,j',p_i^h)|/2+\ell_{y}=2+10(n+1)+\ell_y$.
For a pair $\{x',y\}$ which is not resolved by $f^{mid}(i,j',3-h)$, it satisfies that $30(n+1)-\ell_{x'}=2+10(n+1)+\ell_y$, i.e. $\ell_{x'}+\ell_y=20(n+1)-2$.
Similar to vertex $x$, for vertex $x'$, $\dist(f^{3-h}(i,j,a_1),x')=3+|P(\pi_i^{3-h},a_1)|+|P(\pi_i^{h},a_1)|+\ell_{x'}$ when $\ell_{x'}\leq 10(n+1)$ and
$\dist(f^{3-h}(i,j,a_1),x')=2+|P^{3-h}(i,j',p_i^h)|/2+|P(q_i^{h},\text{mid}(P^{3-h}(i,j',p_i^{h})))|-\ell_{x'}$ when $\ell_{x'}>10(n+1)$.
Thus for a pair $\{x',y\}$ which is not resolved by $f^{3-h}(i,j,a_1)$, there are two cases.
Case 1: $\dist(f^{3-h}(i,j,a_1),y)=\dist(f^{3-h}(i,j,a_1),x')$ when $\ell_{x'}\leq 10(n+1)$, i.e. $\ell_y-\ell_{x'}=3+20(n+1)$.
Case 2: $\dist(f^{3-h}(i,j,a_1),y)=\dist(f^{3-h}(i,j,a_1),x')$ when $\ell_{x'}>10(n+1)$, i.e. $\ell_y+\ell_{x'}=40(n+1)-2$.
Thus we see that if a pair $\{x',y\}$ is not resolved by $f^{mid}(i,j',3-h)$, then it is resolved by $f^{3-h}(i,j,a_1)$.
It follows that every pair $\{x',y\}$ is resolved by $f^{mid}(i,j',3-h)$ and $f^{3-h}(i,j,a_1)$.

Finally we show that every vertex pair of $L_i^h\times L_{i'}^{h'}$ is resolved by $S'$ for $i,i'\in [n],h,h'\in \{1,2\}$ and $i\neq i'$.
We fix arbitrary integers $i,i'\in [n],j,j'\in [m],h,h'\in \{1,2\}$ such that $i\neq i'$.
Let $x\in P(q_i^h,\text{mid}(P^{3-h}(i,j,p_i^h)))$ and $y\in P(q_{i'}^{h'},\text{mid}(P^{3-h'}(i',j',p_{i'}^{h'})))$.
We define $\ell_x$ and $\ell_y$ in a similar way to that of $\ell_x$ in the first paragraph.
For a pair $\{x,y\}$, $\dist(f^{mid}(i,j,3-h),x)=30(n+1)-\ell_{x}$ and
$\dist(f^{mid}(i,j,3-h),y)=\text{min }(2+|P^{3-h}(i,j,p_i^h)|/2+|P(\pi_i^{3-h},a_1)|+|P(\pi_{i'}^{3-h'},a_1)|+\ell_y,
1+|P^{3-h}(i,j,p_i^h)|/2+|P(\pi_i^{3-h},a_1)|+|P(\pi_{i'}^{3-h'},a_1)|+|P^{3-h'}(i',j',p_{i'}^{h'})|+|P(q_{i'}^{h'},\text{mid}(P^{3-h'}(i',j',p_{i'}^{h'})))|-\ell_y)
\geq 1+30(n+1)>\dist(f^{mid}(i,j,3-h),x)$.
It follows that every pair $\{x,y\}$ is resolved by $f^{mid}(i,j,3-h)$.
This completes the proof for the lemma.
\end{proof}

\begin{lemma} \label{RxR}
Every pair of distinct vertices $x,y\in\bigcup_{r\in\{1,2,3\}}R_r$ is resolved by $S'$.
\end{lemma}
\begin{proof}
Let's fix an arbitrary integer $r\in \{1,2,3\}$.

First, we show that every pair of distinct vertices of $\bigcup_{i\in [n]}(P(a_r,u_r^i)\cup P(a_r,v_r^i))$ is resolved by $S'$.
Let's fix an arbitrary integer $i\in [n]$.
Let $x_u\in P(a_r,u_r^i)$.
Let $P(a_r,x_u)$ be the subpath of $P(a_r,u_r^i)$ from $a_r$ to $x_u$ and let $\ell_{x_u}=|P(a_r,x_u)|$.
Since $\dist(f(\pi_1^1,a_r),x_u)=2+\ell_{x_u}$, obviously two distinct vertices of $P(a_r,u_r^i)$ are distinguished by $f(\pi_1^1,a_r)$.
Let $x_v\in P(a_r,v_r^i)$.
Let $P(a_r,x_v)$ be the subpath of $P(a_r,v_r^i)$ from $a_r$ to $x_v$ and let $\ell_{x_v}=|P(a_r,x_v)|$.
Since $\dist(f(\pi_1^1,a_r),x_v)=2+\ell_{x_v}$, obviously two distinct vertices of $P(a_r,v_r^i)$ are distinguished by $f(\pi_1^1,a_r)$.
For the pair $\{x_u,x_v\}$, if $\ell_{x_u}\neq \ell_{x_v}$, then it is resolved by $f(\pi_1^1,a_r)$.
Otherwise, if $\ell_{x_u}=\ell_{x_v}<|P(a_r,u_r^i)|$, then $\dist(f^1(u_r^i,v_r^i),x_u)=1+|P(a_r,u_r^i)|-\ell_{x_u}<\dist(f^1(u_r^i,v_r^i),x_v)=2+|P(a_r,v_r^i)|-\ell_{x_v}$.
By Claim~\ref{basic}, the vertex pair $\{u_r^i,v_r^i\}$ is resolved by $S'$.
Thus every pair $\{x_u,x_v\}$ is resolved by $S'$.
Let $x_u'\in P(a_r,u_r^{i'})$ and $x_v'\in P(a_r,v_r^{i'})$ for some integer $i'\in [n]$ such that $i'\neq i$.
We define $\ell_{x_u'}$ and$\ell_{x_v'}$ in a similar way to that of $\ell_{x_u}$ and $\ell_{x_v}$.   
For a pair $\{x_u,x_u'\}$, $\dist(f^1(u_r^i,v_r^i),x_u)=1+|P(a_r,u_r^i)|-\ell_{x_u}<\dist(f^1(u_r^i,v_r^i),x_u')=1+|P(a_r,u_r^i)|+\ell_{x_u'}$.
Thus every pair $\{x_u,x_u'\}$ is resolved by $f^1(u_r^i,v_r^i)$.
Similarly, we can show that every pair $\{x_u,x_v'\}$, $\{x_v,x_v'\}$ is resolved by $f^1(u_r^i,v_r^i)$.

Then we show that every pair of distinct vertices of $\bigcup_{i\in [n]}(P(c_r,u_r^i)\cup P(c_r,v_r^i))$ is resolved by $S'$.
Let's fix arbitrary integers $i,i'\in [n]$ such that $i\neq i'$.
Let $z_u\in P(c_r,u_r^i)$, $z_v\in P(c_r,v_r^i)$, $z_u'\in P(c_r,u_r^{i'})$ and $z_v'\in P(c_r,v_r^{i'})$.
We define $\ell_{z_u}$, $\ell_{z_v}$, $\ell_{z_u'}$ and $\ell_{z_v'}$ in a similar way to that of $\ell_{x_u}$ and $\ell_{z_v}$ in last paragraph.
Since $\dist(f(\pi_1^1,c_r),z_u)=2+\ell_{z_u}$, obviously two distinct vertices of $P(c_r,u_r^i)$ are distinguished by $f(\pi_1^1,c_r)$.
Since $\dist(f(\pi_1^1,c_r),z_v)=2+\ell_{z_v}$, two distinct vertices of $P(c_r,v_r^i)$ are distinguished by $f(\pi_1^1,c_r)$.
For a pair $\{z_u,z_v\}$, if $\ell_{z_u}\neq \ell_{z_v}$, then it is resolved by $f(\pi_1^1,c_r)$.
Otherwise, if $\ell_{z_u}=\ell_{z_v}<|P(c_r,u_r^i)|$, then $\dist(f^1(u_r^i,v_r^i),z_u)=1+|P(c_r,u_r^i)|-\ell_{z_u}<\dist(f^1(u_r^i,v_r^i),z_v)=2+|P(c_r,v_r^i)|-\ell_{z_v}$.
Thus every pair $\{z_u,z_v\}$ is resolved by $S'$.
For a pair $\{z_u,z_u'\}$, if $\ell_{z_u}\neq \ell_{z_u'}$, then it is resolved by  $f(\pi_1^1,c_r)$.
Otherwise if $\ell_{z_u}=\ell_{z_u'}$, then
$\dist(f^1(u_r^i,v_r^i),z_u')=\text{min }(1+|P(c_r,u_r^i)|+\ell_{z_u'},1+|P(a_r,u_r^i)|+|P(a_r,u_r^{i'})|+|P(c_r,u_r^{i'})|-\ell_{z_u'}-2)$ if $|P(c_r,u_r^{i'})|-\ell_{z_u'}\geq 2$.
$\dist(f^1(u_r^i,v_r^i),z_u')=1+|P(a_r,u_r^i)|+|P(a_r,u_r^{i'})|+|P(c_r,u_r^{i'})|-\ell_{z_u'}$ if $|P(c_r,u_r^{i'})|-\ell_{z_u'}<2$.
If $\dist(f^1(u_r^i,v_r^i),z_u')=1+|P(c_r,u_r^i)|+\ell_{z_u'}$, then $\{z_u,z_u'\}$ is resolved by $f^1(u_r^i,v_r^i)$.
Otherwise, if $\dist(f^1(u_r^i,v_r^i),z_u')=1+|P(a_r,u_r^i)|+|P(a_r,u_r^{i'})|+|P(c_r,u_r^{i'})|-\ell_{z_u'}$,
suppose that there is a pair $\{z_u,z_u'\}$ which is not resolved by $f^1(u_r^i,v_r^i)$.
We get that $i=2(n+1)$, a contradiction.
If $\dist(f^1(u_r^i,v_r^i),z_u')=1+|P(a_r,u_r^i)|+|P(a_r,u_r^{i'})|+|P(c_r,u_r^{i'})|-\ell_{z_u'}-2$,
suppose that there is a pair $\{z_u,z_u'\}$ which is not resolved by $f^1(u_r^i,v_r^i)$.
We get that $20i=40(n+1)-2$, a contradiction.
Thus every pair $\{z_u,z_u'\}$ is resolved by $f^1(u_r^i,v_r^i)$ or $f(\pi_1^1,c_r)$.
Similarly, we can show that every pair $\{z_u,z_v'\}$ and $\{z_v,z_v'\}$ is resolved by $f^1(u_r^i,v_r^i)$ or $f(\pi_1^1,c_r)$.


Next we show that every pair of distinct vertices of $\bigcup_{i\in [n]}(P(b_r,u_r^i)\cup P(b_r,v_r^i))$ is resolved by $S'$.
Let's fix an arbitrary integer $i,i'\in [n]$ such that $i'\neq i$
Let $y_u\in P(b_r,u_r^i)$, $y_v\in P(b_r,v_r^i)$, $y_u'\in P(b_r,u_r^{i'})$ and $y_v'\in P(b_r,v_r^{i'})$.
We define $\ell_{y_u}$, $\ell_{y_v}$, $\ell_{y_u'}$ and $\ell_{y_v'}$ in a similar way to that of $\ell_{x_u}$ and $\ell_{x_v}$ in the second paragraph.
Since $(G,n,\chi,\cal P)$ is a YES-instance, by Claim~\ref{basic}, the pair $\{u_r^i,v_r^i\}$ is resolved by some vertex of $S$, say $s_{\eta}^{\tau}$.
Since $\dist(s_{\eta}^{\tau},y_u)=|P(s_{\eta}^{\tau},b_r)|+\ell_{y_u}=20(n+1)+5i+1+\ell_{y_u}$, every vertex pair of $P(b_r,u_r^i)$ is resolved by $s_{\eta}^{\tau}$.
Since $\dist(s_{\eta}^{\tau},y_v)=|P(s_{\eta}^{\tau},b_r)|+\ell_{y_v}=20(n+1)+5i+1+\ell_{y_v}$, every vertex pair of $P(b_r,v_r^i)$ is resolved by $s_{\eta}^{\tau}$.
For a pair $\{y_u,y_v\}$, if $\ell_{y_u}\neq \ell_{y_v}$, then it is resolved by $s_{\eta}^{\tau}$.
For a pair $\{y_u,y_v\}$ such that $\ell_{y_u}=\ell_{y_v}$,
$\dist(f^1(u_r^i,v_r^i),y_u)=2+|P(b_r,u_r^i)|-\ell_{y_u}=1+20(n+1)-5i-\ell_{y_u}>\dist(f^1(u_r^i,v_r^i),y_v)=2+|P(b_r,v_r^i)|-\ell_{y_v}=20(n+1)-5i-\ell_{y_v}$.
Thus every pair $\{y_u,y_v\}$ is resolved by $s_{\eta}^{\tau}$ or $f^1(u_r^i,v_r^i)$.
For a pair $\{y_u,y_u'\}$ such that $y_u\neq b_r$ and $y_u'\neq b_r$, $\dist(f^1(u_r^i,v_r^i),y_u)=1+20(n+1)-5i-\ell_{y_u}\leq 20(n+1)-5i$.
$\dist(f^1(u_r^i,v_r^i),y_u')=\text{min }(2+|P(b_r,v_r^i)|+\ell_{y_u'},1+|P(a_r,u_r^i)|+|P(a_r,u_r^{i'})|+|P(b_r,u_r^{i'})|-\ell_{y_u'})>20(n+1)-5i$.
Thus every pair $\{y_u,y_u'\}$ such that $y_u\neq b_r$ and $y_u'\neq b_r$ is resolved by $f^1(u_r^i,v_r^i)$.
Similarly, every pair $\{y_u,y_v'\}$ and $\{y_v,y_v'\}$ is resolved by $f^1(u_r^i,v_r^i)$.

Then we show that every pair of distinct vertices of $R_r$ is resolved by $S'$.
Let's fix arbitrary integers $i,i'\in [n]$ such that $i'\neq i$.
Let $x_u\in P(a_r,u_r^i), x_v\in P(a_r,v_r^i), y_u\in P(b_r,u_r^i), y_v\in P(b_r,v_r^i), z_u\in P(c_r,u_r^i)$ and $z_v\in P(c_r,v_r^i)$.
Let $x_u'\in P(a_r,u_r^{i'}), x_v'\in P(a_r,v_r^{i'}), y_u'\in P(b_r,u_r^{i'}), y_v'\in P(b_r,v_r^{i'}), z_u'\in P(c_r,u_r^{i'})$ and $z_v'\in P(c_r,v_r^{i'})$.
We define $\ell_{x_u},\ell_{x_v},\ell_{y_u},\ell_{y_v},\ell_{z_u},\ell_{z_v},\ell_{x_u'},\ell_{x_v'},\ell_{y_u'},\ell_{y_v'},\ell_{z_u'}$ and $\ell_{z_v'}$ in a similar way to that of $\ell_{x_u}$ and $\ell_{x_v}$ in the second paragraph.
For a pair $\{x_u,y_u\}$, $\dist(f(\pi_1^1,a_r),x_u)=2+\ell_{x_u}<\dist(f(\pi_1^1,a_r),y_u)=2+|P(a_r,u_r^i)|+|P(b_r,u_r^i)|-\ell_{y_u}$.
Thus every pair $\{x_u,y_u\}$ is resolved by $f(\pi_1^1,a_r)$.
Similarly, every pair $\{x_u,y_v\}$, $\{x_v,y_v\}$ and $\{x_v,y_u\}$ are resolved by $f(\pi_1^1,a_r)$.
For a pair $\{x_u,z_u\}$, $\dist(f(\pi_1^1,a_r),x_u)=2+\ell_{x_u}$. 
$\dist(f(\pi_1^1,a_r),z_u)=\text{min }(2+|P(a_r,u_r^i)|+|P(c_r,u_r^i)|-\ell_{z_u}-2,|P(\pi_1^1,a_r)|+|P(\pi_1^1,c_r)|+\ell_{z_u})$ if $|P(c_r,u_r^i)|-\ell_{z_u}\geq 2$.
$\dist(f(\pi_1^1,a_r),z_u)=2+|P(a_r,u_r^i)|+|P(c_r,u_r^i)|-\ell_{z_u}$ if $|P(c_r,u_r^i)|-\ell_{z_u}<2$.
It follows that $\dist(f(\pi_1^1,a_r),x_u)=2+\ell_{x_u}<\dist(f(\pi_1^1,a_r),z_u)$. 
Thus every pair $\{x_u,z_u\}$ is resolved by $f(\pi_1^1,a_r)$.
Similarly, every pair $\{x_u,z_v\}$, $\{x_v,z_v\}$ and $\{x_v,z_u\}$ are resolved by $f(\pi_1^1,a_r)$.
For a pair $\{y_u,z_u\}$, $\dist(f(\pi_1^1,c_r),z_u)=2+\ell_{z_u}$.
$\dist(f(\pi_1^1,c_r),b_r)=\text{min}_{j\in [m]}(|P(c_r,s_1^j)|+|P(s_1^j,b_r)|)=35n+41$.
$\dist(f(\pi_1^1,c_r),y_u)=\text{min }(2+|P(c_r,u_r^i)|+|P(b_r,u_r^i)|-\ell_{y_u},\dist(f(\pi_1^1,c_r),b_r)+\ell_{y_u})>\dist(f(\pi_1^1,c_r),z_u)$.
Thus every pair $\{y_u,z_u\}$ is resolved by $f(\pi_1^1,a_r)$.
Similarly, every pair $\{y_u,z_v\}$, $\{y_v,z_v\}$ and $\{y_v,z_u\}$ are resolved by $f(\pi_1^1,c_r)$.
For a pair $\{x_u,y_u'\}$, $\dist(f^1(u_r^i,v_r^i),x_u)=1+|P(a_r,u_r^i)|-\ell_{x_u}$ if $x_u\neq u_r^i$ and
$\dist(f^1(u_r^i,v_r^i),u_r^i)=2$.
$\dist(f^1(u_r^i,v_r^i),y_u')=\text{min }(1+|P(a_r,u_r^i)|+|P(a_r,u_r^{i'})|+|P(b_r,u_r^{i'})|-\ell_{y_u'},2+|P(b_r,v_r^i)|+\ell_{y_u'})>\dist(f^1(u_r^i,v_r^i),x_u)$.
Thus every pair $\{x_u,y_u'\}$ is resolved by $f^1(u_r^i,v_r^i)$.
Similarly, every pair $\{x_u,y_v'\}$, $\{x_v,y_v'\}$, $\{x_v,y_u'\}$, $\{x_u,z_u'\}$, $\{x_u,z_v'\}$, $\{x_v,z_v'\}$ and $\{x_v,z_u'\}$ are resolved by $f^1(u_r^i,v_r^i)$.
For a pair $\{y_u,z_u'\}$, $\dist(f^1(u_r^i,v_r^i),y_u)=2+|P(b_r,u_r^i)|-\ell_{y_u}$ if $y_u\neq b_r$ and
$\dist(f^1(u_r^i,v_r^i),y_u)=1+|P(b_r,u_r^i)|$ if $y_u=b_r$.
$\dist(f^1(u_r^i,v_r^i),z_u')=\text{min }(1+|P(a_r,u_r^i)|+|P(a_r,u_r^{i'})|+|P(c_r,u_r^{i'})|-2-\ell_{z_u'},1+|P(c_r,u_r^i)|+\ell_{z_u'})$ if $|P(c_r,u_r^{i'})|-\ell_{z_u'}\geq 2$.
$\dist(f^1(u_r^i,v_r^i),z_u')=\text{min }(1+|P(a_r,u_r^i)|+|P(a_r,u_r^{i'})|+|P(c_r,u_r^{i'})|-\ell_{z_u'},1+|P(c_r,u_r^i)|+\ell_{z_u'})$ if $|P(c_r,u_r^{i'})|-\ell_{z_u'}<2$.
It follows that $\dist(f^1(u_r^i,v_r^i),z_u')>20(n+1)>\dist(f^1(u_r^i,v_r^i),y_u)$.
Thus every pair $\{y_u,z_u'\}$ is resolved by $f^1(u_r^i,v_r^i)$.
Similarly, every pair $\{y_u,z_v'\}$, $\{y_v,z_v'\}$ and $\{y_v,z_u'\}$ are resolved by $f^1(u_r^i,v_r^i)$.
As a result, every pair of distinct vertices of $R_r$ is resolved by $S'$.

Finally, we show that every vertex pair of $R_r\times R_{r'}$ with $r'\in \{1,2,3\}$ and $r'\neq r$ is resolved by $S'$.
Let's fix arbitrary integers $i,i'\in [n]$ and $r'\in \{1,2,3\}$ such that $r'\neq r$.
Let $x_u\in P(a_r,u_r^i), x_v\in P(a_r,v_r^i), y_u\in P(b_r,u_r^i), y_v\in P(b_r,v_r^i), z_u\in P(c_r,u_r^i)$ and $z_v\in P(c_r,v_r^i)$.
Let $x_u'\in P(a_{r'},u_{r'}^{i'}), x_v'\in P(a_{r'},v_{r'}^{i'}), y_u'\in P(b_{r'},u_{r'}^{i'}), y_v'\in P(b_{r'},v_{r'}^{i'}), z_u'\in P(c_{r'},u_{r'}^{i'})$ and $z_v'\in P(c_{r'},v_{r'}^{i'})$.
We define $\ell_{x_u},\ell_{x_v},\ell_{y_u},\ell_{y_v}$, $\ell_{z_u},\ell_{z_v},\ell_{x_u'},\ell_{x_v'},\ell_{y_u'},\ell_{y_v'},\ell_{z_u'}$ and $\ell_{z_v'}$ in a similar way to that of $\ell_{x_u}$ and $\ell_{x_v}$ in the second paragraph.
For a pair $\{x_u,x_u'\}$, $\dist(f(s_1^1,a_r),x_u)=2+\ell_{x_u}<20(n+1)<\dist(f(s_1^1,a_r),x_u')=2+|P(\pi_1^1,a_r)|+|P(\pi_1^1,a_{r'})|+\ell_{x_u'}$.
Thus every pair $\{x_u,x_u'\}$ is resolved by $f(s_1^1,a_r)$.
Similarly, every pair $\{x_u,x_v'\}$, $\{x_v,x_u'\}$ and $\{x_v,x_v'\}$ are resolved by $f(s_1^1,a_r)$.
For a pair $\{x_u,z_u'\}$, $\dist(f(s_1^1,a_r),x_u)=2+\ell_{x_u}<20(n+1)$.
$\dist(f(s_1^1,a_r),z_u')=\text{min }(2+|P(\pi_1^1,a_r)|+|P(\pi_1^1,c_{r'})|+\ell_{z_u'},2+|P(\pi_1^1,a_r)|+|P(\pi_1^1,a_{r'})|+|P(a_{r'},u_{r'}^{i'})|-2+|P(c_{r'},u_{r'}^{i'})|-\ell_{z_u'})$ if $|P(c_{r'},u_{r'}^{i'})|-\ell_{z_u'}\geq 2$.
$\dist(f(s_1^1,a_r),z_u')=\text{min }(2+|P(\pi_1^1,a_r)|+|P(\pi_1^1,c_{r'})|+\ell_{z_u'},2+|P(\pi_1^1,a_r)|+|P(\pi_1^1,a_{r'})|+|P(a_{r'},u_{r'}^{i'})|+|P(c_{r'},u_{r'}^{i'})|-\ell_{z_u'})$ if $|P(c_{r'},u_{r'}^{i'})|-\ell_{z_u'}<2$.
It follows that $\dist(f(s_1^1,a_r),z_u')>20(n+1)$.
Thus every pair $\{x_u,z_u'\}$ is resolved by $f(s_1^1,a_r)$.
Similarly, every pair $\{x_u,z_v'\}$, $\{x_v,z_u'\}$ and $\{x_v,z_v'\}$ are resolved by $f(s_1^1,a_r)$.
For a pair $\{x_u,y_u'\}$ such that $y_u'\neq b_{r'}$, $\dist(f(\pi_1^1,a_r),x_u)=2+\ell_{x_u}<20(n+1)$.
$\dist(f(\pi_1^1,a_r),b_{r'})=2+\text{min}_{j\in [m]}(|P(s_1^j,a_r)|+|P(s_1^j,b_{r'})|)$.
$\dist(f(\pi_1^1,a_r),y_u')=\text{min }(\dist(f(\pi_1^1,a_r),b_{r'})+\ell_{z_u'},|P(\pi_1^1,a_r)|+|P(\pi_1^1,a_{r'})|+|P(a_{r'},u_{r'}^{i'})|+|P(b_{r'},u_{r'}^{i'})|-\ell_{y_u'})>20(n+1)$.
Thus every pair $\{x_u,y_u'\}$ such that $y_u'\neq b_{r'}$ is resolved by $f(\pi_1^1,a_r)$.
Similarly, every pair $\{x_v,y_u'\}$ such that $y_u'\neq b_{r'}$, every pair $\{x_v,y_v'\}$ and $\{x_u,y_v'\}$ are resolved by $f(\pi_1^1,a_r)$.
For a pair $\{y_u,y_u'\}$ such that $y_u\neq b_r$, $\dist(f^1(u_r^i,v_r^i),y_u)=2+|P(b_r,u_r^i)|-\ell_{y_u}<20(n+1)$.
$\dist(f^1(u_r^i,v_r^i),b_{r'})=\text{min }(2+|P(b_r,v_r^i)|+\text{min}_{j\in [m]}(|P(s_1^j,b_r)|+|P(s_1^j,b_{r'})|),1+|P(a_r,v_r^i)|+\text{min}_{j'\in [m]}(|P(s_1^{j'},a_r)|+|P(s_1^{j'},b_{r'})|))$.
$\dist(f^1(u_r^i,v_r^i),y_u')=\text{min }(\dist(f^1(u_r^i,v_r^i),b_{r'})+\ell_{y_u'},1+|P(a_r,v_r^i)|+|P(\pi_1^1,a_r)|+|P(\pi_1^1,a_{r'})|+|P(a_{r'},u_{r'}^{i'})|+|P(b_{r'},u_{r'}^{i'})|-\ell_{y_u'})>20(n+1)$.
Thus every pair $\{y_u,y_u'\}$ such that $y_u\neq b_r$ is resolved by $f^1(u_r^i,v_r^i)$.
Similarly, every pair $\{y_u,y_v'\}$ such that $y_u\neq b_r$, every pair $\{y_v,y_v'\}$ and $\{y_v,y_u'\}$ are resolved by $f^1(u_r^i,v_r^i)$.
For a pair $\{y_u,z_u'\}$, $\dist(f^1(u_r^i,v_r^i),y_u)<20(n+1)$.
$\dist(f^1(u_r^i,v_r^i),z_u')=\text{min }(1+|P(a_r,u_r^i)|+|P(\pi_1^1,a_r)|+|P(\pi_1^1,c_{r'})|+\ell_{z_u'},
1+|P(a_r,u_r^i)|+|P(\pi_1^1,a_r)|+|P(\pi_1^1,a_{r'})|+|P(a_{r'},u_{r'}^{i'})|+|P(c_{r'},u_{r'}^{i'})|-\ell_{z_u'}-2)$ if $|P(c_{r'},u_{r'}^{i'})|-\ell_{z_u'}\geq 2$.
$\dist(f^1(u_r^i,v_r^i),z_u')=\text{min }(1+|P(a_r,u_r^i)|+|P(\pi_1^1,a_r)|+|P(\pi_1^1,c_{r'})|+\ell_{z_u'},
1+|P(a_r,u_r^i)|+|P(\pi_1^1,a_r)|+|P(\pi_1^1,a_{r'})|+|P(a_{r'},u_{r'}^{i'})|+|P(c_{r'},u_{r'}^{i'})|-\ell_{z_u'})$ if $|P(c_{r'},u_{r'}^{i'})|-\ell_{z_u'}<2$.
It follows that $\dist(f^1(u_r^i,v_r^i),z_u')>30(n+1)$.
Thus every pair $\{y_u,z_u'\}$ is resolved by $f^1(u_r^i,v_r^i)$.
Similarly, every pair $\{y_u,z_v'\}$, $\{y_v,z_u'\}$ and $\{y_v,z_v'\}$ are resolved by $f^1(u_r^i,v_r^i)$.
For a pair $\{z_u,z_u'\}$, $\dist(f^1(u_r^i,v_r^i),z_u)=1+|P(c_r,u_r^i)|-\ell_{z_u}$ if $z_u\neq u_r^i$ and $\dist(f^1(u_r^i,v_r^i),z_u)=2$ if $z_u=u_r^i$.
Thus $\dist(f^1(u_r^i,v_r^i),z_u)<30(n+1)<\dist(f^1(u_r^i,v_r^i),z_u')$ and every pair $\{z_u,z_u'\}$ is resolved by $f^1(u_r^i,v_r^i)$.
Similarly, every pair $\{z_u,z_v'\}$, $\{z_v,z_u'\}$ and $\{z_v,z_v'\}$ are resolved by $f^1(u_r^i,v_r^i)$.
As a result, every vertex pair of $R_r\times R_{r'}$ with $r'\in \{1,2,3\}$ and $r'\neq r$ is resolved by $S'$.
This completes the proof for the lemma.
\end{proof}

\begin{lemma} \label{UxPi}
Every pair $\{x,y\}\in \bigcup_{i\in [n]}U_i\times \bigcup_{i\in [n]}\Pi_i$ is resolved by $S'$.
\end{lemma}
\begin{proof}
First, we show that every pair $\{x,y\}\in U_i^h\times \bigcup_{j\in [m],r\in \{1,2,3\}}\Pi^h(i,j,r)$ for $i\in [n],h\in \{1,2\}$ is resolved by $S'$.
We fix arbitrary integers $i\in [n],j,j'\in [m]$, $h\in \{1,2\}$ and $r\in \{1,2,3\}$.
Suppose that $x\in P(s_i^j,p_i^h)$ and $y\in P^h(i,j',a_r)$.
For a vertex $x\in P(s_i^j,p_i^h)$, let $P(x,p_i^h)$ be the subpath of $P(s_i^j,p_i^h)$ from $x$ to $p_i^h$ and $|P(x,p_i^h)|=\ell_x$.
For a vertex $y\in P^h(i,j',a_r)$, let $P(y,\pi_i^h)$ be the subpath of $P^h(i,j',a_r)$ from $y$ to $\pi_i^h$ and $|P(y,\pi_i^h)|=\ell_y$.
Let $j^*\in [m]$ be an integer such that $j^*\neq j$ and $j^*\neq j'$.
Then $\dist(f^h(i,j^*,a_r),x)=3+\ell_x$ and $\dist(f^h(i,j^*,a_r),y)=2+\ell_y$.
Thus every pair $\{x,y\}$ is resolved by $f^h(i,j^*,a_r)$ unless $\ell_y-\ell_x=1$.
For the pair $\{x,y\}$ such that $\ell_y-\ell_x=1$, there are two cases.
Case 1: $j=j'$.
Since $\dist(f^h(i,j',a_r),x)=3+\ell_x$ if $x\neq s_i^j$, $\dist(f^h(i,j',a_r),s_i^j)=20(n+1)+1$ and
$\dist(f^h(i,j',a_r),y)=\ell_y$ if $y\neq \pi_i^h$, $\dist(f^h(i,j',a_r),\pi_i^h)=2$,
every pair $\{x,y\}$ such that $\ell_y-\ell_x=1$ is resolved by $f^h(i,j',a_r)$.
Case 2: $j\neq j'$.
Since $\dist(f^h(i,j',a_r),x)=3+\ell_x$,
$\dist(f^h(i,j',a_r),y)=\ell_y$ if $y\neq \pi_i^h$ and $\dist(f^h(i,j',a_r),\pi_i^h)=2$,
every pair $\{x,y\}$ such that $\ell_y-\ell_x=1$ is resolved by $f^h(i,j',a_r)$.
It follows that every pair $\{x,y\}$ is resolved by $f^h(i,j^*,a_r)$ or $f^h(i,j',a_r)$.
Similarly, we can show that every vertex pair of $P(s_i^j,p_i^h)\times P^h(i,j',b_r)$, $P(s_i^j,p_i^h)\times P^h(i,j',c_r)$ and $P(s_i^j,p_i^h)\times P^h(i,j',p_i^{3-h})$ are resolved by $S'$.

Next we show that every pair $\{x,y\}\in U_i^h\times \bigcup_{j\in [m],r\in \{1,2,3\}}\Pi^{3-h}(i,j,r)$ for $i\in [n],h\in \{1,2\}$ is resolved by $S'$.
We fix arbitrary integers $i\in [n],j,j'\in [m]$, $h\in \{1,2\}$ and $r\in \{1,2,3\}$.
Suppose that $x\in P(s_i^j,p_i^h)\setminus \{s_i^j\}$ and $y\in P^{3-h}(i,j',a_r)$.
We define $\ell_x$ and $\ell_y$ in a similar way to that of $\ell_x$ in the first paragraph.
There are two cases.
Case 1: $j=j'$.
$\dist(f^{3-h}(i,j',a_r),y)=\ell_y$ if $y\neq \pi_i^h$, $\dist(f^h(i,j',a_r),\pi_i^h)=2$.
$\dist(f^{3-h}(i,j',a_r),x)=\text{min }(|P^{3-h}(i,j,a_r)|+1+|P(s_i^j,p_i^h)|-\ell_x,3+|P(\pi_i^{3-h},c_r)|+|P(c_r,\pi_i^h)|+\ell_x)=
\text{min }(40(n+1)+1-\ell_x,20(n+1)+3+\ell_x)\geq 20(n+1)+1>\dist(f^{3-h}(i,j',a_r),y)$.
Thus in this case, every pair $\{x,y\}$ is resolved by $f^{3-h}(i,j',a_r)$.
Case 2: $j\neq j'$.
$\dist(f^{3-h}(i,j',a_r),y)=\ell_y$ if $y\neq \pi_i^h$, $\dist(f^h(i,j',a_r),\pi_i^h)=2$.
$\dist(f^{3-h}(i,j',a_r),x)=\text{min }(3+|P(s_i^j,p_i^{3-h})|+|P(s_i^j,p_i^h)|-\ell_x,3+|P(\pi_i^{3-h},c_r)|+|P(c_r,\pi_i^h)|+\ell_x)=
\text{min }(40(n+1)+3-\ell_x,20(n+1)+3+\ell_x)\geq 20(n+1)+3>\dist(f^{3-h}(i,j',a_r),y)$.
Thus in this case, every pair $\{x,y\}$ is resolved by $f^{3-h}(i,j',a_r)$.
It follows that every pair $\{x,y\}$ is resolved by $f^{3-h}(i,j',a_r)$.
Similarly, we can show that every vertex pair of $P(s_i^j,p_i^h)\times P^{3-h}(i,j',b_r)$, $P(s_i^j,p_i^h)\times P^{3-h}(i,j',c_r)$ and $P(s_i^j,p_i^h)\times P^{3-h}(i,j',p_i^h)$ are resolved by $S'$.

Finally we show that every pair $\{x,y\}\in U_i^h\times \bigcup_{j\in [m],r\in \{1,2,3\}}\Pi^{h'}(i',j,r)$ for $i,i'\in [n],h,h'\in \{1,2\}$ such that $i\neq i'$ is resolved by $S'$.
We fix arbitrary integers $i,i'\in [n],j,j'\in [m]$, $h,h'\in \{1,2\}$ and $r\in \{1,2,3\}$ such that $i\neq i'$.
Suppose that $x\in P(s_i^j,p_i^h)$ and $y\in P^{h'}(i',j',a_r)$.
We define $\ell_x$ and $\ell_y$ in a similar way to that of $\ell_x$ in the first paragraph.
Let $j^*\in [m]$ be an integer such that $j^*\neq j'$.
Then $\dist(f^{h'}(i',j^*,a_r),y)=2+\ell_y$.
$\dist(f^{h'}(i',j^*,a_r),s_i^j)=\text{min}_{r'\in \{1,2,3\}}(2+|P(\pi_{i'}^{h'},c_{r'})|+|P(s_i^j,c_{r'})|)$.
$\dist(f^{h'}(i',j^*,a_r),x)=\text{min }(\dist(f^{h'}(i',j^*,a_r),s_i^j)+|P(s_i^j,\pi_i^h)|-\ell_x,3+|P(\pi_{i'}^{h'},c_r)|+|P(\pi_i^h,c_r)|+\ell_x)>2+20(n+1)\geq \dist(f^{h'}(i',j^*,a_r),y)$.
Thus every pair $\{x,y\}$ is resolved by $f^{h'}(i',j^*,a_r)$.
Similarly, we can show that every vertex pair of $P(s_i^j,p_i^h)\times P^{h'}(i',j',b_r)$, $P(s_i^j,p_i^h)\times P^{h'}(i',j',c_r)$ and $P(s_i^j,p_i^h)\times P^{h'}(i',j',p_{i'}^{3-h'})$ are resolved by $f^{h'}(i',j^*,a_r)$.
This completes the proof for the lemma.
\end{proof}

\begin{lemma} \label{UxL}
Every pair $\{x,y\}\in \bigcup_{i\in [n]}U_i\times \bigcup_{i\in [n]}L_i$ is resolved by $S'$.
\end{lemma}
\begin{proof}
First, we show that every pair $\{x,y\}\in P(s_i^j,p_i^h)\times P(q_i^h,\text{mid}(P^{3-h}(i,j,p_i^{h})))$ for $i\in [n],j\in [m]$ and $h\in \{1,2\}$ is resolved by $S'$.
We fix arbitrary integers $i\in [n],j\in [m]$ and $h\in \{1,2\}$.
Suppose that $x\in P(s_i^j,p_i^h)$ and $y\in P(q_i^h,\text{mid}(P^{3-h}(i,j,p_i^{h})))$.
For a vertex $x\in P(s_i^j,p_i^h)$, let $P(x,p_i^h)$ be the subpath of $P(s_i^j,p_i^h)$ from $x$ to $p_i^h$ and $|P(x,p_i^h)|=\ell_x$.
For a vertex $y\in P(q_i^h,\text{mid}(P^{3-h}(i,j,p_i^{h})))$, let $P(y,q_i^h)$ be the subpath of $P(q_i^h,\text{mid}(P^{3-h}(i,j,p_i^{h})))$ from $y$ to $q_i^h$ and $|P(y,q_i^h)|=\ell_y$.
Then $\dist(f^h(i,j,a_1),x)=3+\ell_x$.
$\dist(f^h(i,j,a_1),y)=3+\ell_y$.
Thus every pair $\{x,y\}$ is resolved by $f^h(i,j,a_1)$ unless $\ell_x=\ell_y$.
Suppose that $S'\cap X_i=\{s_i^{j^*}\}$.
For the pair $\{x,y\}$ such that $\ell_x=\ell_y$,
there are two cases.
Case 1: $j^*=j$.
Then $\dist(s_i^{j^*},x)=|P(s_i^j,p_i^h)|-\ell_x=20(n+1)-\ell_x$.
$\dist(s_i^{j^*},y)=\text{min }(|P(s_i^j,p_i^h)|+2+\ell_y,1+|P^{3-h}(i,j,p_i^h)|/2+|P(q_i^h,\text{mid}(P^{3-h}(i,j,p_i^{h})))|-\ell_y)=
\text{min }(20(n+1)+2+\ell_y,40(n+1)-\ell_y)\neq 20(n+1)-\ell_y$.
Thus in this case every pair $\{x,y\}$ such that $\ell_x=\ell_y$ is resolved by $s_i^{j^*}$.
Case 2: $j^*\neq j$.
$\dist(s_i^{j^*},s_i^j)=\text{min}_{r\in \{1,2,3\}}(|P(s_i^{j^*},c_r)|+|P(s_i^j,c_r)|)=20(n+1)-10\lambda^*+20(n+1)-10\lambda$ for some $\lambda,\lambda^*\in [n]$.
Then $\dist(s_i^{j^*},x)=\text{min}(|P(s_i^{j^*},p_i^h)|+\ell_x,\dist(s_i^{j^*},s_i^j)+|P(s_i^j,p_i^h)|-\ell_x)=\text{min}(20(n+1)+\ell_x,60(n+1)-10\lambda-10\lambda^*-\ell_x)$.
$\dist(s_i^{j^*},y)=\text{min }(|P(s_i^{j^*},p_i^h)|+2+\ell_y,|P(s_i^{j^*},p_i^{3-h})|+1+|P^{3-h}(i,j,p_i^h)|/2+|P(q_i^h,\text{mid}(P^{3-h}(i,j,p_i^{h})))|-\ell_y)=
\text{min }(20(n+1)+2+\ell_y,60(n+1)-\ell_y)$.
Thus every pair $\{x,y\}$ such that $\ell_x=\ell_y$ is resolved by $s_i^{j^*}$.
As a result, every pair $\{x,y\}$ is resolved by $f^h(i,j,a_1)$ or $s_i^{j^*}$.

Next we show that every pair $\{x,y\}\in P(s_i^j,p_i^h)\times P(q_i^h,\text{mid}(P^{3-h}(i,j',p_i^{h})))$ for $i\in [n],j,j'\in [m]$ and $h\in \{1,2\}$ such that $j\neq j'$ is resolved by $S'$.
We fix arbitrary integers $i\in [n],j,j'\in [m]$ and $h\in \{1,2\}$ such that $j\neq j'$.
Suppose that $x\in P(s_i^j,p_i^h)$ and $y\in P(q_i^h,\text{mid}(P^{3-h}(i,j',p_i^h)))$.
We define $\ell_x$ and $\ell_y$ in a similar way to that of $\ell_x$ and $\ell_y$ in the first paragraph.
$\dist(f^{mid}(i,j',3-h),y)=1+|P(q_i^h,\text{mid}(P^{3-h}(i,j',p_i^{h})))|-\ell_y=30(n+1)-\ell_y$.
$\dist(f^{mid}(i,j',3-h),x)=\text{min }(1+|P^{3-h}(i,j',p_i^h)|/2+|P(s_i^{j'},p_i^h)|-1+\ell_x,2+|P^{3-h}(i,j',p_i^h)|/2+|P(s_i^{j'},p_i^{3-h})|+|P(s_i^{j'},p_i^h)|-\ell_x)=
\text{min }(30(n+1)+\ell_x,2+50(n+1)-\ell_x)$.
Thus every pair $\{x,y\}$ is resolved by $f^{mid}(i,j',3-h)$ unless $\ell_x=\ell_y=0$, i.e. except the pair $\{p_i^h,q_i^h\}$.
According to Lemma~\ref{forcedSet}, $\{p_i^h,q_i^h\}$ is resolved by $S'$.
Thus every pair $\{x,y\}$ is resolved by $S'$.

Then we show that every pair $\{x,y\}\in P(s_i^j,p_i^h)\times P(q_i^{3-h},\text{mid}(P^h(i,j',p_i^{3-h})))$ for $i\in [n],j,j'\in [m]$ and $h\in \{1,2\}$ is resolved by $S'$.
We fix arbitrary integers $i\in [n],j,j'\in [m]$ and $h\in \{1,2\}$.
Suppose that $x\in P(s_i^j,p_i^h)$ and $y\in P(q_i^{3-h},\text{mid}(P^h(i,j',p_i^{3-h})))$.
We define $\ell_x$ and $\ell_y$ in a similar way to that of $\ell_x$ and $\ell_y$ in the first paragraph.
Let $j^*\in [m]$ be an integer such that $j^*\neq j$.
Then $\dist(f^{3-h}(i,j^*,a_1),y)=3+\ell_y$.
$\dist(f^{3-h}(i,j^*,a_1),x)=\text{min }(1+|P^{3-h}(i,j,a_1)|+|P(s_i^j,p_i^h)|-\ell_x,3+|P(\pi_i^{3-h},a_1)|+|P(\pi_i^h,a_1)|+\ell_x)=
\text{min }(40(n+1)+1-\ell_x,20(n+1)+3+\ell_x)\geq 20(n+1)+2$ if $x\neq s_i^j$ and $\dist(f^{3-h}(i,j^*,a_1),s_i^j)=3+20(n+1)$.
Thus any pair $\{x,y\}$ such that $\ell_y<20(n+1)-1$ is resolved by $f^{3-h}(i,j^*,a_1)$.
For the pair $\{x,y\}$ such that $20(n+1)-1\leq \ell_y\leq 30(n+1)-1$,
$\dist(f^{mid}(i,j',h),y)=1+|P(q_i^{3-h},\text{mid}(P^h(i,j',p_i^{3-h})))|-\ell_y=30(n+1)-\ell_y\leq 10(n+1)+1$.
If $j'=j$, then $\dist(f^{mid}(i,j',h),x)=\text{min }(2+|P^h(i,j,p_i^{3-h})|/2+\ell_x,2+|P^h(i,j,p_i^{3-h})|/2+|P(s_i^j,p_i^h)|-\ell_x)=
\text{min }(2+10(n+1)+\ell_x,2+30(n+1)-\ell_x)\geq 10(n+1)+2$.
If $j'\neq j$, then $\dist(f^{mid}(i,j',h),x)=2+|P^h(i,j,p_i^{3-h})|/2+\ell_x\geq 10(n+1)+2$.
As a result, every pair $\{x,y\}$ is resolved by $f^{3-h}(i,j^*,a_1)$ or $f^{mid}(i,j',h)$.

Finally we show that every pair $\{x,y\}\in P(s_i^j,p_i^h)\times P(q_{i'}^{h'},\text{mid}(P^{3-h'}(i',j',p_{i'}^{h'})))$ for $i,i'\in [n],j,j'\in [m]$ and $h,h'\in \{1,2\}$ such that $i\neq i'$ is resolved by $S'$.
We fix arbitrary integers $i,i'\in [n],j,j'\in [m]$ and $h,h'\in \{1,2\}$ such that $i\neq i'$.
Suppose that $x\in P(s_i^j,p_i^h)$ and $y\in P(q_{i'}^{h'},\text{mid}(P^{3-h'}(i',j',p_{i'}^{h'})))$.
We define $\ell_x$ and $\ell_y$ in a similar way to that of $\ell_x$ and $\ell_y$ in the first paragraph.
Then $\dist(f^h(i,j,a_1),x)=3+\ell_x$ if $x\neq s_i^j$ and $\dist(f^h(i,j,a_1),s_i^j)=2+20(n+1)$.
$\dist(f^h(i,j,a_1),y)=\text{min }(3+|P(\pi_i^h,a_1)|+|P(\pi_{i'}^{h'},a_1)|+\ell_y,2+|P(\pi_i^h,a_1)|+|P(\pi_{i'}^{3-h'},a_1)|+|P^{3-h'}(i',j',p_{i'}^{h'})|/2+|P(q_{i'}^{h'},\text{mid}(P^{3-h'}(i',j',p_{i'}^{h'})))|-\ell_y)
=\text{min }(3+20(n+1)+\ell_y,1+60(n+1)-\ell_y)\geq 3+20(n+1)>\dist(f^h(i,j,a_1),x)$.
Thus every pair $\{x,y\}$ is resolved by $f^h(i,j,a_1)$.
This completes the proof for the lemma.
\end{proof}

\begin{lemma} \label{UxS}
Every pair $\{x,y\}\in \bigcup_{i\in [n]}U_i\times \bigcup_{i\in [n]}S_i$ is resolved by $S'$.
\end{lemma}
\begin{proof}
We show that every pair $\{x,y\}\in P(s_i^j,p_i^h)\times (P(\pi_{i'}^{h'},a_r)\cup P(\pi_{i'}^{h'},c_r))$ for $i,i'\in [n],j\in [m],h'\in \{1,2\}$ and $r\in \{1,2,3\}$ is resolved by $S'$.
We fix arbitrary integers $i,i'\in [n],j\in [m]$,$h,h'\in \{1,2\}$ and $r\in \{1,2,3\}$.
Suppose that $x\in P(s_i^j,p_i^h)$ and $y\in P(\pi_{i'}^{h'},a_r)$.
For a vertex $x\in P(s_i^j,p_i^h)$, let $P(x,p_i^h)$ be the subpath of $P(s_i^j,p_i^h)$ from $x$ to $p_i^h$ and $|P(x,p_i^h)|=\ell_x$.
For a vertex $y\in P(\pi_{i'}^{h'},a_r)$, let $P(y,\pi_{i'}^{h'})$ be the subpath of $P(\pi_{i'}^{h'},a_r)$ from $y$ to $\pi_{i'}^{h'}$ and $|P(y,\pi_{i'}^{h'})|=\ell_y$.
Then $\dist(f(s_i^j,a_r),y)=2+|P(\pi_{i'}^{h'},a_r)|-\ell_y$.
Suppose that $|P(s_i^j,a_r)|=20(n+1)+10p$ for some $p\in [n]$.
$\dist(f(s_i^j,a_r),x)=\text{min }(3+|P(\pi_i^h,a_r)|+\ell_x,|P(s_i^j,a_r)|+|P(s_i^j,p_i^h)|-\ell_x)=\text{min }(3+10(n+1)+\ell_x,40(n+1)+10p-\ell_x)\geq 3+10(n+1)>\dist(f(s_i^j,a_r),y)$.
Thus every pair $\{x,y\}$ is resolved by $f(s_i^j,a_r)$.
Similarly, we can show that every vertex pair of $P(s_i^j,p_i^h)\times P(\pi_{i'}^{h'},c_r)$ for $i,i'\in [n],j\in [m],h'\in \{1,2\}$ and $r\in \{1,2,3\}$ is resolved by $f(s_i^j,c_r)$.
This completes the proof for the lemma.
\end{proof}

\begin{lemma} \label{UxH}
Every pair $\{x,y\}\in \bigcup_{i\in [n]}U_i\times \bigcup_{i\in [n]}H_i$ is resolved by $S'$.
\end{lemma}
\begin{proof}
First we show that every pair $\{x,y\}\in P(s_i^j,p_i^h)\times (P(s_i^j,a_r)\cup P(s_i^j,c_r))$ for $i\in [n],j\in [m],h\in \{1,2\}$ and $r\in \{1,2,3\}$ is resolved by $S'$.
We fix arbitrary integers $i\in [n],j\in [m],h\in \{1,2\}$ and $r\in \{1,2,3\}$.
Suppose that $x\in P(s_i^j,p_i^h)$ and $y\in P(s_i^j,a_r)$.
For a vertex $x\in P(s_i^j,p_i^h)$, let $P(s_i^j,x)$ be the subpath of $P(s_i^j,p_i^h)$ from $s_i^j$ to $x$ and $|P(s_i^j,x)|=\ell_x$.
For a vertex $y\in P(s_i^j,a_r)$, let $P(s_i^j,y)$ be the subpath of $P(s_i^j,a_r)$ from $s_i^j$ to $y$ and $|P(s_i^j,y)|=\ell_y$.
Let $|P(s_i^j,a_r)|=20(n+1)+10\lambda$ for some $\lambda\in [n]$.
Then $\dist(f(s_i^j,a_r),x)=\text{min }(3+|P(\pi_i^h,a_r)|+|P(s_i^j,p_i^h)|-\ell_x,|P(s_i^j,a_r)|+\ell_x)=\text{min }(3+30(n+1)-\ell_x,20(n+1)+10\lambda+\ell_x)$.
$\dist(f(s_i^j,a_r),y)=|P(s_i^j,a_r)|-\ell_y=20(n+1)+10\lambda-\ell_y$.
$\dist(f(\pi_i^h,a_r),x)=\text{min }(1+|P(\pi_i^h,a_r)|+|P(s_i^j,p_i^h)|-\ell_x,2+|P(s_i^j,a_r)|+\ell_x)=\text{min }(1+30(n+1)-\ell_x,2+20(n+1)+10\lambda+\ell_x)$.
$\dist(f(\pi_i^h,a_r),y)=2+|P(s_i^j,a_r)|-\ell_y=2+20(n+1)+10\lambda-\ell_y$.
For the pair $\{x,y\}$ which is not resolved by $f(s_i^j,a_r)$,
it satisfies that $\dist(f(s_i^j,a_r),x)=\dist(f(s_i^j,a_r),y)=3+|P(\pi_i^h,a_r)|+|P(s_i^j,p_i^h)|-\ell_x$.
Thus $\dist(f(\pi_i^h,a_r),x)<\dist(f(s_i^j,a_r),x)=\dist(f(s_i^j,a_r),y)<\dist(f(\pi_i^h,a_r),y)$.
It follows that every pair $\{x,y\}$ is resolved by $f(s_i^j,a_r)$ or $f(\pi_i^h,a_r)$.
Similarly, we can show that every vertex pair of $P(s_i^j,p_i^h)\times P(s_i^j,c_r)$ is resolved by $f(s_i^j,c_r)$ or $f(\pi_i^h,c_r)$.

Next we show that every vertex pair of $P(s_i^j,p_i^h)\times P(s_i^j,b_r)$ for $i\in [n],j\in [m],h\in \{1,2\}$ and $r\in \{1,2,3\}$ is resolved by $S'$.
We fix arbitrary integers $i\in [n],j\in [m],h\in \{1,2\}$ and $r\in \{1,2,3\}$.
Suppose that $x\in P(s_i^j,p_i^h)\setminus \{s_i^j\}$ and $y\in P(s_i^j,b_r)\setminus \{s_i^j\}$.
We define $\ell_x$ and $\ell_y$ in a similar way to that of $\ell_x$ and $\ell_y$ in the first paragraph.
Then $\dist(f^{mid}(i,j,3-h),x)=|P^{3-h}(i,j,p_i^h)|/2+\ell_x=10(n+1)+\ell_x$ and
$\dist(f^{mid}(i,j,3-h),y)=2+|P^{3-h}(i,j,p_i^h)|/2+\ell_y=2+10(n+1)+\ell_y$.
For the vertex pair $\{x,y\}$ which is not resolved by $f^{mid}(i,j,3-h)$, i.e. $\ell_x=2+\ell_y$,
$\dist(f^{mid}(i,j,h),x)=2+|P^h(i,j,p_i^{3-h})|/2+\ell_x=2+10(n+1)+\ell_x>\dist(f^{mid}(i,j,h),y)=2+|P^h(i,j,p_i^{3-h})|/2+\ell_y=10(n+1)+2+\ell_y=10(n+1)+\ell_x$.
Thus every pair $\{x,y\}$ is resolved by $f^{mid}(i,j,3-h)$ or $f^{mid}(i,j,h)$.

Then we show that every pair $\{x,y\}\in P(s_i^j,p_i^h)\times (P(s_{i'}^{j'},a_r)\cup P(s_{i'}^{j'},c_r))$ for $i,i'\in [n],j,j'\in [m],h\in \{1,2\}$ and $r\in \{1,2,3\}$ such that $i\neq i'$ or $j\neq j'$ is resolved by $S'$.
We fix arbitrary integers $i,i'\in [n],j,j'\in [m],h\in \{1,2\}$ and $r\in \{1,2,3\}$ such that $i\neq i'$ or $j\neq j'$.
Suppose that $x\in P(s_i^j,p_i^h)$ and $y\in P(s_{i'}^{j'},a_r)$.
We define $\ell_x$ and $\ell_y$ in a similar way to that of $\ell_x$ and $\ell_y$ in the first paragraph.
Let $|P(s_i^j,a_r)|=20(n+1)+10\lambda$ and $|P(s_{i'}^{j'},a_r)|=20(n+1)+10\lambda'$ for some $\lambda,\lambda'\in [n]$.
Then $\dist(f(s_{i'}^{j'},a_r),x)=\dist(f(\pi_{i'}^{3-h},a_r),x)=\text{min }(3+|P(\pi_i^h,a_r)|+|P(s_i^j,p_i^h)|-\ell_x,2+|P(s_i^j,a_r)|+\ell_x)=\text{min }(3+30(n+1)-\ell_x,2+20(n+1)+10\lambda+\ell_x)$.
$\dist(f(s_{i'}^{j'},a_r),y)=\dist(f(\pi_{i'}^{3-h},a_r),y)-2=|P(s_{i'}^{j'},a_r)|-\ell_y$ if $y\neq a_r$ and
$\dist(f(s_{i'}^{j'},a_r),a_r)=\dist(f(\pi_{i'}^{3-h},a_r),a_r)=2$.
It follows that every pair $\{x,y\}$ is resolved by $f(s_{i'}^{j'},a_r)$ or $f(\pi_{i'}^{3-h},a_r)$.
Similarly, we can show that every vertex pair of $P(s_i^j,p_i^h)\times P(s_{i'}^{j'},c_r)$ is resolved by $f(s_{i'}^{j'},c_r)$ or $f(\pi_{i'}^{3-h},c_r)$.

Finally we show that every vertex pair of $P(s_i^j,p_i^h)\times P(s_{i'}^{j'},b_r)$ for $i,i'\in [n],j,j'\in [m],h\in \{1,2\}$ and $r\in \{1,2,3\}$ such that $i\neq i'$ or $j\neq j'$ is resolved by $S'$.
We fix arbitrary integers $i,i'\in [n],j,j'\in [m],h\in \{1,2\}$ and $r\in \{1,2,3\}$ such that $i\neq i'$ or $j\neq j'$.
Suppose that $x\in P(s_i^j,p_i^h)$ and $y\in P(s_{i'}^{j'},b_r)$.
We define $\ell_x$ and $\ell_y$ in a similar way to that of $\ell_x$ and $\ell_y$ in the first paragraph.
Let $|P(s_i^j,b_r)|=20(n+1)+5\lambda+1$ and $|P(s_{i'}^{j'},b_r)|=20(n+1)+5\lambda'+1$ and for some $\lambda,\lambda'\in [n]$.
There are two cases.
Case 1: $i=i'$ and $j\neq j'$.
$\dist(f^{mid}(i,j,3-h),x)=|P^{3-h}(i,j,p_i^h)|/2+\ell_x=10(n+1)+\ell_x$ if $x\neq s_i^j$ and $\dist(f^{mid}(i,j,3-h),s_i^j)=2+10(n+1)$.
$\dist(f^{mid}(i,j,3-h),y)=\text{min }(2+|P^{3-h}(i,j,p_i^h)|/2+|P(s_{i'}^{j'},p_i^{3-h})|+\ell_y,2+|P^{3-h}(i,j,p_i^h)|/2+|P(s_i^j,b_r)|+|P(s_{i'}^{j'},b_r)|-\ell_y)
=\text{min }(2+30(n+1)+\ell_y,4+50(n+1)+5\lambda+5\lambda'-\ell_y)\geq 2+30(n+1)>\dist(f^{mid}(i,j,3-h),x)$.
Thus in this case, every pair $\{x,y\}$ is resolved by $f^{mid}(i,j,3-h)$.
Case 2: $i\neq i'$.
$\dist(f^{mid}(i,j,3-h),s_{i'}^{j'})=1+|P^{3-h}(i,j,p_i^h)|/2+\text{min}_{r\in \{1,2,3\}}(|P(\pi_i^{3-h},c_r)|+|P(s_{i'}^{j'},c_r)|)$.
$\dist(f^{mid}(i,j,3-h),y)=\text{min }(\dist(f^{mid}(i,j,3-h),s_{i'}^{j'})+\ell_y,2+|P^{3-h}(i,j,p_i^h)|/2+|P(s_i^j,b_r)|+|P(s_{i'}^{j'},b_r)|-\ell_y)
=\text{min }(1+40(n+1)-10\lambda'+\ell_y,4+50(n+1)+5\lambda+5\lambda'-\ell_y)>30(n+1)+5>\dist(f^{mid}(i,j,3-h),x)$.
Thus in this case, every pair $\{x,y\}$ is resolved by $f^{mid}(i,j,3-h)$.
It follows that every pair $\{x,y\}$ is resolved by $f^{mid}(i,j,3-h)$.
This completes the proof for the lemma.
\end{proof}

\begin{lemma}  \label{UxR}
Every pair $\{x,y\}\in \bigcup_{i\in [n]}U_i\times \bigcup_{r\in \{1,2,3\}}R_r$ is resolved by $S'$.
\end{lemma}
\begin{proof}
First we show that every pair $\{x,y\}\in P(s_i^j,p_i^h)\times (P(u_r^{i'},a_r)\cup P(v_r^{i'},a_r))$ for $i,i'\in [n],j\in [m],h\in \{1,2\}$ and $r\in \{1,2,3\}$ is resolved by $S'$.
We fix arbitrary integers $i,i'\in [n],j\in [m],h\in \{1,2\}$ and $r\in \{1,2,3\}$.
Suppose that $x\in P(s_i^j,p_i^h)$, $y\in P(u_r^{i'},a_r)$.
For a vertex $x\in P(s_i^j,p_i^h)$, let $P(s_i^j,x)$ be the subpath of $P(s_i^j,p_i^h)$ from $s_i^j$ to $x$ and $|P(s_i^j,x)|=\ell_x$.
For a vertex $y\in P(a_r,u_r^{i'})$, let $P(u_r^{i'},y)$ be the subpath of $P(a_r,u_r^{i'})$ from $u_r^{i'}$ to $y$ and $|P(u_r^{i'},y)|=\ell_{y}$.
Let $|P(s_i^j,a_r)|=20(n+1)+10\lambda$ for some $\lambda\in [n]$ and $|P(u_r^{i'},a_r)|=20(n+1)-10i'$.
There are two cases.
Case 1: $\lambda\leq i'$.
$\dist(f^1(u_r^{i'},v_r^{i'}),x)=\text{min }(1+|P(a_r,u_r^{i'})|+|P(s_i^j,a_r)|+\ell_x,1+|P(a_r,u_r^{i'})|+|P(\pi_i^h,a_r)|+|P(s_i^j,\pi_i^h)|-\ell_x)=
\text{min }(40(n+1)-10(i'-\lambda)+1+\ell_x,50(n+1)+1-10i'-\ell_x)$.
$\dist(f^1(u_r^{i'},v_r^{i'}),y)=1+\ell_{y}$ if $y\neq u_r^{i'}$ and $\dist(f^1(u_r^{i'},v_r^{i'}),y)=2$ if $y=u_r^{i'}$.
Thus $\dist(f^1(u_r^{i'},v_r^{i'}),x)\geq 30(n+1)-10i'+1>\dist(f^1(u_r^{i'},v_r^{i'}),y)$.
Case 2: $\lambda>i'$.
$\dist(f^1(u_r^{i'},v_r^{i'}),x)=\text{min }(1+|P(c_r,u_r^{i'})|+|P(s_i^j,c_r)|+\ell_x,1+|P(a_r,u_r^{i'})|+|P(\pi_i^h,a_r)|+|P(s_i^j,\pi_i^h)|-\ell_x)=
\text{min }(40(n+1)-10(\lambda-i')+1+\ell_x,50(n+1)+1-10i'-\ell_x)$.
$\dist(f^1(u_r^{i'},v_r^{i'}),y)=1+\ell_{y}$ if $y\neq u_r^{i'}$ and $\dist(f^1(u_r^{i'},v_r^{i'}),y)=2$ if $y=u_r^{i'}$.
Thus $\dist(f^1(u_r^{i'},v_r^{i'}),x)\geq 30(n+1)-10i'+1>\dist(f^1(u_r^{i'},v_r^{i'}),y)$.
Thus every pair $\{x,y\}$ is resolved by $f^1(u_r^{i'},v_r^{i'})$.
Similarly, we can show that every vertex pair of $P(s_i^j,p_i^h)\times P(v_r^{i'},a_r)$ is resolved by $f^1(u_r^{i'},v_r^{i'})$.

Next we show that every pair $\{x,y\}\in P(s_i^j,p_i^h)\times (P(u_r^{i'},b_r)\cup P(v_r^{i'},b_r))$ for $i,i'\in [n],j\in [m],h\in \{1,2\}$ and $r\in \{1,2,3\}$ is resolved by $S'$.
We fix arbitrary integers $i,i'\in [n],j\in [m],h\in \{1,2\}$ and $r\in \{1,2,3\}$.
Suppose that $x\in P(s_i^j,p_i^h)$, $y\in P(u_r^{i'},b_r)$.
We define $\ell_x$ and $\ell_y$ in a similar way to that of $\ell_x$ and $\ell_y$ in the first paragraph.
Let $|P(s_i^j,a_r)|=20(n+1)+10\lambda$ for some $\lambda\in [n]$ and $|P(u_r^{i'},a_r)|=20(n+1)-10i'$.
There are two cases.
Case 1: $\lambda\leq i'$.
$\dist(f^1(u_r^{i'},v_r^{i'}),x)=\text{min }(1+|P(a_r,u_r^{i'})|+|P(s_i^j,a_r)|+\ell_x,1+|P(a_r,u_r^{i'})|+|P(\pi_i^h,a_r)|+|P(s_i^j,\pi_i^h)|-\ell_x)$.
$\dist(f^2(u_r^{i'},v_r^{i'}),x)=\dist(f^1(u_r^{i'},v_r^{i'}),x)-1=\text{min }(|P(a_r,u_r^{i'})|+|P(s_i^j,a_r)|+\ell_x,|P(a_r,u_r^{i'})|+|P(\pi_i^h,a_r)|+|P(s_i^j,\pi_i^h)|-\ell_x)$.
$\dist(f^1(u_r^{i'},v_r^{i'}),y)=\dist(f^2(u_r^{i'},v_r^{i'}),y)=2+\ell_{y}$.
In this case, for a pair $\{x,y\}$ which is not resolved by $f^1(u_r^{i'},v_r^{i'})$, $\dist(f^2(u_r^{i'},v_r^{i'}),x)=\dist(f^1(u_r^{i'},v_r^{i'}),x)-1<\dist(f^1(u_r^{i'},v_r^{i'}),y)=\dist(f^2(u_r^{i'},v_r^{i'}),y)$.
Case 2: $\lambda>i'$.
$\dist(f^1(u_r^{i'},v_r^{i'}),x)=\text{min }(1+|P(c_r,u_r^{i'})|+|P(s_i^j,c_r)|+\ell_x,1+|P(a_r,u_r^{i'})|+|P(\pi_i^h,a_r)|+|P(s_i^j,\pi_i^h)|-\ell_x)$.
$\dist(f^2(u_r^{i'},v_r^{i'}),x)=\dist(f^1(u_r^{i'},v_r^{i'}),x)-1=\text{min }(|P(c_r,u_r^{i'})|+|P(s_i^j,c_r)|+\ell_x,|P(a_r,u_r^{i'})|+|P(\pi_i^h,a_r)|+|P(s_i^j,\pi_i^h)|-\ell_x)$.
$\dist(f^1(u_r^{i'},v_r^{i'}),y)=\dist(f^2(u_r^{i'},v_r^{i'}),y)=2+\ell_{y}$.
Similar to Case 1, in this case, for a pair $\{x,y\}$ which is not resolved by $f^1(u_r^{i'},v_r^{i'})$, $\dist(f^2(u_r^{i'},v_r^{i'}),x)<\dist(f^2(u_r^{i'},v_r^{i'}),y)$.
Thus every pair $\{x,y\}$ is resolved by $f^1(u_r^{i'},v_r^{i'})$ or $f^2(u_r^{i'},v_r^{i'})$.
Similarly, we can show that every vertex pair of $P(s_i^j,p_i^h)\times P(v_r^{i'},b_r)$ is resolved by $f^1(u_r^{i'},v_r^{i'})$ or $f^2(u_r^{i'},v_r^{i'})$.

Finally we show that every pair $\{x,y\}\in P(s_i^j,p_i^h)\times (P(u_r^{i'},c_r)\cup P(v_r^{i'},c_r))$ for $i,i'\in [n],j\in [m],h\in \{1,2\}$ and $r\in \{1,2,3\}$ is resolved by $S'$.
We fix arbitrary integers $i,i'\in [n],j\in [m],h\in \{1,2\}$ and $r\in \{1,2,3\}$.
Suppose that $x\in P(s_i^j,p_i^h)$, $y\in P(u_r^{i'},c_r)$.
We define $\ell_x$ and $\ell_y$ in a similar way to that of $\ell_x$ and $\ell_y$ in the first paragraph.
Let $|P(s_i^j,c_r)|=20(n+1)-10\lambda$ for some $\lambda\in [n]$ and $|P(u_r^{i'},c_r)|=20(n+1)+10i'$.
Then $\dist(f(s_i^j,c_r),x)=\text{min }(3+|P(\pi_i^h,c_r)|+|P(s_i^j,p_i^h)|-\ell_x,|P(s_i^j,c_r)|+\ell_x)=\text{min }(3+30(n+1)-\ell_x,20(n+1)-10\lambda+\ell_x)$.
$\dist(f(\pi_i^h,c_r),x)=\text{min }(1+|P(\pi_i^h,c_r)|+|P(s_i^j,p_i^h)|-\ell_x,2+|P(s_i^j,c_r)|+\ell_x)=\text{min }(1+30(n+1)-\ell_x,2+20(n+1)-10\lambda+\ell_x)$.
$\dist(f(\pi_i^{3-h},c_r),x)=\text{min }(3+|P(\pi_i^h,c_r)|+|P(s_i^j,p_i^h)|-\ell_x,2+|P(s_i^j,c_r)|+\ell_x)=\text{min }(3+30(n+1)-\ell_x,2+20(n+1)-10\lambda+\ell_x)$.
$\dist(f(s_i^j,c_r),y)=\dist(f(\pi_i^h,c_r),y)=\dist(f(\pi_i^{3-h},c_r),y)=2+|P(c_r,u_r^{i'})|-\ell_y=2+20(n+1)+10i'-\ell_y$.
For a pair $\{x,y\}$ which is not resolved by $f(s_i^j,c_r)$, either $f(\pi_i^h,c_r)$ or $f(\pi_i^{3-h},c_r)$ resolves it.
Thus every pair $\{x,y\}$ is resolved by $f(s_i^j,c_r)$, $f(\pi_i^h,c_r)$ or $f(\pi_i^{3-h},c_r)$.
Similarly, we can show that every vertex pair of $P(s_i^j,p_i^h)\times P(v_r^{i'},c_r)$ is resolved by $f(s_i^j,c_r)$, $f(\pi_i^h,c_r)$ or $f(\pi_i^{3-h},c_r)$.
This completes the proof for the lemma.
\end{proof}

\begin{lemma}  \label{PixH}
Every pair $\{x,y\}\in \bigcup_{i\in [n]}\Pi_i\times \bigcup_{i\in [n]}H_i$ is resolved by $S'$.
\end{lemma}
\begin{proof}
We show that every pair $\{x,y\}\in (P^{h}(i,j,a_r)\cup P^{h}(i,j,b_r)\cup P^{h}(i,j,c_r)\cup P^{h}(i,j,p_i^{3-h}))\times (P(s_{i'}^{j'},a_{r'})\cup P(s_{i'}^{j'},b_{r'})\cup P(s_{i'}^{j'},c_{r'}))$ for $i,i'\in [n],j,j'\in [m],h\in \{1,2\}$ and $r,r'\in \{1,2,3\}$ is resolved by $S'$.
We fix arbitrary integers $i,i'\in [n],j,j'\in [m],h\in \{1,2\}$ and $r,r'\in \{1,2,3\}$.
Suppose that $x_1\in P^{h}(i,j,a_r)$, $x_2\in P^{h}(i,j,b_r)$, $x_3\in P^{h}(i,j,c_r)$ and $x_4\in P^{h}(i,j,p_i^{3-h})$.
Suppose that $y_1\in P(s_{i'}^{j'},a_{r'})$, $y_2\in P(s_{i'}^{j'},b_{r'})$ and $y_3\in P(s_{i'}^{j'},c_{r'})$.
For a vertex $x_{\mu}$ for $\mu\in \{1,2,3,4\}$, let $\ell_{x_{\mu}}=\dist(\pi_i^h,x_{\mu})$.
For a vertex $y_{\nu}$ for $\nu\in \{1,2,3\}$, let $\ell_{y_{\nu}}=\dist(s_i^j,y_{\nu})$.
Let $|P(s_i^j,a_{r'})|=20(n+1)+10\lambda$ and $|P(s_{i'}^{j'},a_{r'})|=20(n+1)+10\lambda'$ for some $\lambda,\lambda'\in [n]$.
There are three cases.
Case 1: $s_i^j=s_{i'}^{j'}$ and $r'=r$.
Then $\dist(f(s_{i'}^{j'},a_{r'}),x_1)=\text{min }(2+|P(\pi_i^h,a_r)|+\ell_{x_1},|P(s_i^j,a_r)|+|P^{h}(i,j,a_r)|-1-\ell_{x_1})=\text{min }(2+10(n+1)+\ell_{x_1},40(n+1)+10\lambda-1-\ell_{x_1})$.
$\dist(f(s_{i'}^{j'},a_{r'}),y_1)=|P(s_i^j,a_r)|-\ell_{y_1}=20(n+1)+10\lambda-\ell_{y_1}$ if $y_1\neq a_r$ and $\dist(f(s_{i'}^{j'},a_{r'}),a_r)=2$.
$\dist(f(\pi_i^{3-h},a_{r'}),x_1)=\text{min }(2+|P(\pi_i^h,a_r)|+\ell_{x_1},1+|P(s_i^j,a_r)|+|P^{h}(i,j,a_r)|-\ell_{x_1})=\text{min }(2+10(n+1)+\ell_{x_1},40(n+1)+10\lambda+1-\ell_{x_1})$.
$\dist(f(\pi_i^{3-h},a_{r'}),y_1)=2+|P(s_i^j,a_r)|-\ell_{y_1}=20(n+1)+10\lambda+2-\ell_{y_1}$.
Let $\gamma\in P^h(i,j,a_r)$ be the vertex such that $\dist(\gamma,\pi_i^h)=20(n+1)-1$.
Obviously the pair $\{s_i^j,\gamma\}$ is resolved by $f^h(i,j,a_r)$.
For the pair $\{x_1,y_1\}$ which is not resolved by $f(s_{i'}^{j'},a_{r'})$ and $y_1\neq s_i^j$,
it satisfies that $\dist(f(s_{i'}^{j'},a_{r'}),x_1)=\dist(f(s_{i'}^{j'},a_{r'}),y_1)=\dist(f(\pi_i^{3-h},a_{r'}),x_1)<\dist(f(\pi_i^{3-h},a_{r'}),y_1)$.
Thus in this case, every pair $\{x_1,y_1\}$ is resolved by $f(s_{i'}^{j'},a_{r'})$, $f(\pi_i^{3-h},a_{r'})$ or $f^h(i,j,a_r)$.
Case 2: $s_i^j\neq s_{i'}^{j'}$ and $r'=r$.
Then $\dist(f(s_{i'}^{j'},a_{r'}),x_1)=\text{min }(2+|P(\pi_i^h,a_{r'})|+\ell_{x_1},1+|P(s_i^j,a_{r'})|+|P^{h}(i,j,a_{r'})|-\ell_{x_1})=\text{min }(2+10(n+1)+\ell_{x_1},40(n+1)+10\lambda+1-\ell_{x_1})$.
$\dist(f(s_{i'}^{j'},a_{r'}),y_1)=|P(s_{i'}^{j'},a_{r'})|-\ell_{y_1}=20(n+1)+10\lambda'-\ell_{y_1}$ if $y_1\neq a_{r'}$ and $\dist(f(s_{i'}^{j'},a_{r'}),a_{r'})=2$.
$\dist(f(\pi_i^{3-h},a_{r'}),x_1)=\text{min }(2+|P(\pi_i^h,a_{r'})|+\ell_{x_1},|P(s_i^j,a_{r'})|+|P^{h}(i,j,a_{r'})|+1-\ell_{x_1})=\text{min }(2+10(n+1)+\ell_{x_1},40(n+1)+10\lambda+1-\ell_{x_1})$.
$\dist(f(\pi_i^{3-h},a_{r'}),y_1)=2+|P(s_{i'}^{j'},a_{r'})|-\ell_{y_1}=20(n+1)+10\lambda'+2-\ell_{y_1}$.
For the pair $\{x_1,y_1\}$ which is not resolved by $f(s_{i'}^{j'},a_{r'})$,
it satisfies that $\dist(f(s_{i'}^{j'},a_{r'}),x_1)=\dist(f(s_{i'}^{j'},a_{r'}),y_1)=\dist(f(\pi_i^{3-h},a_{r'}),x_1)<\dist(f(\pi_i^{3-h},a_{r'}),y_1)$.
Thus in this case, every pair $\{x_1,y_1\}$ is resolved by $f(s_{i'}^{j'},a_{r'})$ or $f(\pi_i^{3-h},a_{r'})$.
Case 3: $s_i^j\neq s_{i'}^{j'}$ and $r'\neq r$.
Then $\dist(f(s_{i'}^{j'},a_{r'}),x_1)=\text{min }(2+|P(\pi_i^h,a_{r'})|+\ell_{x_1},2+|P(s_i^j,a_{r'})|+1+|P^{h}(i,j,a_{r'})|-\ell_{x_1})=\text{min }(2+10(n+1)+\ell_{x_1},40(n+1)+10\lambda+3-\ell_{x_1})$.
$\dist(f(s_{i'}^{j'},a_{r'}),y_1)=|P(s_{i'}^{j'},a_{r'})|-\ell_{y_1}=20(n+1)+10\lambda'-\ell_{y_1}$ if $y_1\neq a_{r'}$ and $\dist(f(s_{i'}^{j'},a_{r'}),a_{r'})=2$.
$\dist(f(\pi_i^{3-h},a_{r'}),x_1)=\text{min }(2+|P(\pi_i^h,a_{r'})|+\ell_{x_1},2+|P(s_i^j,a_{r'})|+1+|P^{h}(i,j,a_{r'})|-\ell_{x_1})=\text{min }(2+10(n+1)+\ell_{x_1},40(n+1)+10\lambda+3-\ell_{x_1})$.
$\dist(f(\pi_i^{3-h},a_{r'}),y_1)=2+|P(s_{i'}^{j'},a_{r'})|-\ell_{y_1}=20(n+1)+10\lambda'+2-\ell_{y_1}$.
For the pair $\{x_1,y_1\}$ which is not resolved by $f(s_{i'}^{j'},a_{r'})$,
it satisfies that $\dist(f(s_{i'}^{j'},a_{r'}),x_1)=\dist(f(s_{i'}^{j'},a_{r'}),y_1)=\dist(f(\pi_i^{3-h},a_{r'}),x_1)<\dist(f(\pi_i^{3-h},a_{r'}),y_1)$.
Thus in this case, every pair $\{x_1,y_1\}$ is resolved by $f(s_{i'}^{j'},a_{r'})$ or $f(\pi_i^{3-h},a_{r'})$.
It follows that every pair $\{x_1,y_1\}$ is resolved by $f^h(i,j,a_r)$, $f(s_{i'}^{j'},a_{r'})$ or $f(\pi_i^{3-h},a_{r'})$.
In a similar way, we can show that every vertex pair $\{x_2,y_1\}$, $\{x_3,y_1\}$ and $\{x_4,y_1\}$ are resolved by $S'$.
Also in a similar way, we can show that every vertex pair $\{x_1,y_3\}$, $\{x_2,y_3\}$, $\{x_3,y_3\}$ and $\{x_4,y_3\}$ are resolved by $f(s_{i'}^{j'},c_{r'})$, $f(\pi_i^{3-h},c_{r'})$ or $f^h(i,j,c_r)$.
For a pair $\{x_1,y_2\}$, $\dist(f^h(i,j,a_r),x_1)=\ell_{x_1}$ if $x_1\neq \pi_i^h$ and $\dist(f^h(i,j,a_r),\pi_i^h)=2$.
$\dist(f^h(i,j,a_r),b_{r'})=2+|P(\pi_i^h,a_{r'})|+|P(a_{r'},v_{r'}^n)|+|P(b_{r'},v_{r'}^n)|+|P(s_{i'}^{j'},b_{r'})|-\ell_{y_2}$.
There are three cases.
Case 1: $i=i'$ and $j=j'$.
$\dist(f^h(i,j,a_r),y_2)=\text{min }(|P^h(i,j,a_r)|+1+\ell_{y_2},\dist(f^h(i,j,a_r),b_{r'})+|P(s_{i'}^{j'},b_{r'})|-\ell_{y_2})=
\text{min }(20(n+1)+1+\ell_{y_2},55n+5\lambda+71-\ell_{y_2})>\dist(f^h(i,j,a_r),x_1)$.
Thus in this case, every pair $\{x_1,y_2\}$ is resolved by $f^h(i,j,a_r)$.
Case 2: $i'=i$ and $j\neq j'$.
Then $\dist(f^h(i,j,a_r),y_2)=\text{min }(2+|P^h(i,j,b_{r'})|+\ell_{y_2}-1,\dist(f^h(i,j,a_r),b_{r'})+|P(s_{i'}^{j'},b_{r'})|-\ell_{y_2})=
\text{min }(20(n+1)+1+\ell_{y_2},55n+5\lambda+71-\ell_{y_2})$ if $y_2\neq s_{i'}^{j'}$ and $\dist(f^h(i,j,a_r),s_{i'}^{j'})=3+20(n+1)$.
Thus $\dist(f^h(i,j,a_r),y_2)\geq 20(n+1)+2>\dist(f^h(i,j,a_r),x_1)$.
In this case, every pair $\{x_1,y_2\}$ is resolved by $f^h(i,j,a_r)$.
Case 3: $i'\neq i$.
$\dist(f^h(i,j,a_r),s_{i'}^{j'})=\text{min}_{d\in \{1,2,3\}}(2+|P(\pi_i^h,c_d)|+|P(s_{i'}^{j'},c_d)|)$.
$\dist(f^h(i,j,a_r),y_2)=\text{min }(\dist(f^h(i,j,a_r),s_{i'}^{j'})+\ell_{y_2},\dist(f^h(i,j,a_r),b_{r'})+|P(s_{i'}^{j'},b_{r'})|-\ell_{y_2})> 20(n+1)\geq \dist(f^h(i,j,a_r),x_1)$.
Thus in this case, every pair $\{x_1,y_2\}$ is resolved by $f^h(i,j,a_r)$.
In a similar way, we can show that every vertex pair $\{x_2,y_2\}$, $\{x_3,y_2\}$ and $\{x_4,y_2\}$ are resolved by $S'$.
This completes the proof for the lemma.
\end{proof}

\begin{lemma}  \label{PixL}
Every pair $\{x,y\}\in \bigcup_{i\in [n]}\Pi_i\times \bigcup_{i\in [n]}L_i$ is resolved by $S'$.
\end{lemma}
\begin{proof}
First we show that every pair $\{x,y\}\in (P^{h}(i,j,a_r)\cup P^{h}(i,j,b_r)\cup P^{h}(i,j,c_r)\cup P^{h}(i,j,p_i^{3-h}))\times P(q_i^h,\text{mid}(P^{3-h}(i,j',p_i^{h})))$ for $i\in [n],j,j'\in [m],h\in \{1,2\}$ and $r\in \{1,2,3\}$ is resolved by $S'$.
We fix arbitrary integers $i\in [n],j,j'\in [m],h\in \{1,2\}$ and $r\in \{1,2,3\}$.
Suppose that $x_1\in P^{h}(i,j,a_r)$, $x_2\in P^{h}(i,j,b_r)$, $x_3\in P^{h}(i,j,c_r)$ and $x_4\in P^{h}(i,j,p_i^{3-h})$.
Suppose that $y\in P(q_i^h,\text{mid}(P^{3-h}(i,j',p_i^{h})))$.
For a vertex $x_{\mu}$ for $\mu\in \{1,2,3,4\}$, let $\ell_{x_{\mu}}=\dist(\pi_i^h,x_{\mu})$.
For a vertex $y$, let $\ell_y=\dist(q_i^h,y)$.
Then $\dist(f^h(i,j,a_r),x_1)=\ell_{x_1}$ if $x_1\neq \pi_i^h$ and $\dist(f^h(i,j,a_r),\pi_i^h)=2$.
$\dist(f^h(i,j,a_r),y)=\dist(f^h(i,j,b_r),y)=\dist(f^h(i,j,c_r),y)=\dist(f^h(i,j,p_i^{3-h}),y)=3+\ell_y$.
For the pair $\{x_1,y\}$ that is not resolved by $f^h(i,j,a_r)$, $\dist(f^h(i,j,b_r),y)=\dist(f^h(i,j,a_r),y)=\dist(f^h(i,j,a_r),x_1)=\dist(f^h(i,j,b_r),x_1)-2<\dist(f^h(i,j,b_r),x_1)$.
Thus every pair $\{x_1,y\}$ is resolved by $f^h(i,j,a_r)$ or $f^h(i,j,b_r)$.
In a similar way, we can show that every vertex pair $\{x_2,y\}$, $\{x_3,y\}$ and $\{x_4,y\}$ are resolved by $S'$.

Next we show that every pair $\{x,y\}\in P(q_i^{3-h},\text{mid}(P^{h}(i,j',p_i^{3-h})))\times (P^{h}(i,j,a_r)\cup P^{h}(i,j,b_r)\cup P^{h}(i,j,c_r)\cup P^{h}(i,j,p_i^{3-h}))$ for $i\in [n],j,j'\in [m],h\in \{1,2\}$ and $r\in \{1,2,3\}$ is resolved by $S'$.
We fix arbitrary integers $i\in [n],j,j'\in [m],h\in \{1,2\}$ and $r\in \{1,2,3\}$.
Suppose that $x_1\in P^{h}(i,j,a_r)$, $x_2\in P^{h}(i,j,b_r)$, $x_3\in P^{h}(i,j,c_r)$ $x_4\in P^{h}(i,j,p_i^{3-h})$ and $y\in P(q_i^{3-h},\text{mid}(P^{h}(i,j',p_i^{3-h})))$.
For a vertex $x_{\mu}$ for $\mu\in \{1,2,3,4\}$, let $\ell_{x_{\mu}}=\dist(\pi_i^h,x_{\mu})$.
For a vertex $y$, let $\ell_y=\dist(q_i^{3-h},y)$.
Then $\dist(f^h(i,j,a_r),x_1)=\ell_{x_1}$ if $x_1\neq \pi_i^h$ and $\dist(f^h(i,j,a_r),\pi_i^h)=2$.
$\dist(f^h(i,j,a_r),y)=\dist(f^h(i,j,b_r),y)=\dist(f^h(i,j,c_r),y)=
\text{min }(2+|P^h(i,j',p_i^{3-h})|/2+|P(q_i^{3-h},\text{mid}(P^{h}(i,j',p_i^{3-h})))|-\ell_y,2+|P(\pi_i^h,a_r)|+|P(\pi_i^{3-h},a_r)|+1+\ell_y)=
\text{min }(1+40(n+1)-\ell_y,3+20(n+1)+\ell_y)$.
For the pair $\{x_1,y\}$ that is not resolved by $f^h(i,j,a_r)$, $\dist(f^h(i,j,b_r),y)=\dist(f^h(i,j,a_r),y)=\dist(f^h(i,j,a_r),x_1)=\dist(f^h(i,j,b_r),x_1)-2<\dist(f^h(i,j,b_r),x_1)$.
Thus every pair $\{x_1,y\}$ is resolved by $f^h(i,j,a_r)$ or $f^h(i,j,b_r)$.
In a similar way, we can show that every vertex pair $\{x_2,y\}$ and $\{x_3,y\}$ are resolved by $S'$.
For the pair $\{x_4,y\}$, there are two cases.
Case 1: $j'\neq j$.
In this case, the analysis is similar to that of $\{x_1,y\}$ above and every pair $\{x_4,y\}$ is resolved by $f^h(i,j,p_i^{3-h})$ or $f^h(i,j,a_r)$.
Case 2: $j'=j$.
$\dist(f^h(i,j,p_i^{3-h}),x_4)=\ell_{x_4}$ if $x_4\neq \pi_i^h$ and $\dist(f^h(i,j,p_i^{3-h}),\pi_i^h)=2$.
$\dist(f^h(i,j,p_i^{3-h}),y)=
\text{min }(|P^h(i,j',p_i^{3-h})|/2+|P(q_i^{3-h},\text{mid}(P^{h}(i,j',p_i^{3-h})))|-\ell_y,2+|P(\pi_i^h,a_r)|+|P(\pi_i^{3-h},a_r)|+1+\ell_y)=
\text{min }(40(n+1)-1-\ell_y,3+20(n+1)+\ell_y)$.
For the pair $\{x_4,y\}$ which is not resolved by $f^h(i,j,p_i^{3-h})$, it satisfies that $\dist(f^h(i,j,p_i^{3-h}),x_4)=\dist(f^h(i,j,p_i^{3-h}),y)=40(n+1)-1-\ell_y=\ell_{x_4}$, i.e.
$\dist(f^{mid}(i,j,h),x_4)=\dist(f^{mid}(i,j,h),y)$ and $10(n+1)<\ell_{x_4}\leq 20(n+1), 20(n+1)-1\leq \ell_y<30(n+1)-1$.
For such pairs, $\dist(f^{ecc}(i,j,3-h,r),x_4)=2+30(n+1)-\ell_{x_4}<2+20(n+1)<\dist(f^{ecc}(i,j,3-h,r),y)=50(n+1)+1-\ell_y$.
Thus in this case, every pair $\{x_4,y\}$ is resolved by $f^h(i,j,p_i^{3-h})$ or $f^{ecc}(i,j,3-h,r)$.

Finally we show that every pair $\{x,y\}\in (P^{h}(i,j,a_r)\cup P^{h}(i,j,b_r)\cup P^{h}(i,j,c_r)\cup P^{h}(i,j,p_i^{3-h}))\times P(q_{i'}^{h'},\text{mid}(P^{3-h'}(i',j',p_{i'}^{h'})))$ for $i,i'\in [n],j,j'\in [m],h,h'\in \{1,2\}$ and $r\in \{1,2,3\}$ such that $i\neq i'$ is resolved by $S'$.
We fix arbitrary integers $i,i'\in [n],j,j'\in [m],h,h'\in \{1,2\}$ and $r\in \{1,2,3\}$ such that $i\neq i'$.
Suppose that $x_1\in P^{h}(i,j,a_r)$, $x_2\in P^{h}(i,j,b_r)$, $x_3\in P^{h}(i,j,c_r)$ and $x_4\in P^{h}(i,j,p_i^{3-h})$.
Suppose that $y\in P(q_{i'}^{3-h'},\text{mid}(P^{h'}(i',j',p_{i'}^{3-h'})))$.
For a vertex $x_{\mu}$ for $\mu\in \{1,2,3,4\}$, let $\ell_{x_{\mu}}=\dist(\pi_i^h,x_{\mu})$.
For a vertex $y$, let $\ell_y=\dist(q_{i'}^{3-h'},y)$.
For the pair $\{x_1,y\}$, $\dist(f^h(i,j,a_r),x_1)=\ell_{x_1}$ if $x_1\neq \pi_i^h$ and $\dist(f^h(i,j,a_r),\pi_i^h)=2$.
$\dist(f^h(i,j,a_r),y)=\text{min }(2+|P(\pi_i^h,a_r)|+|P(\pi_{i'}^{h'},a_r)|+1+\ell_y,2+|P(\pi_i^h,a_r)|+|P(\pi_{i'}^{3-h'},a_r)|+|P^{3-h'}(i',j',p_{i'}^{h'})|/2+|P(q_{i'}^{h'},\text{mid}(P^{3-h'}(i',j',p_{i'}^{h'})))|-\ell_y)=
\text{min }(3+20(n+1)+\ell_y,1+60(n+1)-\ell_y)\geq 3+20(n+1)>\dist(f^h(i,j,a_r),x_1)$.
Thus every pair $\{x_1,y\}$ is resolved by $f^h(i,j,a_r)$.
Similarly, we can show that every pair $\{x_2,y\}$, $\{x_3,y\}$ and $\{x_4,y\}$ are resolved by $f^h(i,j,b_r)$, $f^h(i,j,c_r)$ and $f^h(i,j,p_i^{3-h})$ respectively.
This completes the proof for the lemma.
\end{proof}

\begin{lemma}  \label{PixS}
Every pair $\{x,y\}\in \bigcup_{i\in [n]}\Pi_i\times \bigcup_{i\in [n]}S_i$ is resolved by $S'$.
\end{lemma}
\begin{proof}
We show that every pair $\{x,y\}\in (P^{h}(i,j,a_r)\cup P^{h}(i,j,b_r)\cup P^{h}(i,j,c_r)\cup P^{h}(i,j,p_i^{3-h}))\times (P(\pi_{i'}^{h'},a_{r'})\cup P(\pi_{i'}^{h'},c_{r'}))$ for $i,i'\in [n],j\in [m],h,h'\in \{1,2\}$ and $r,r'\in \{1,2,3\}$ is resolved by $S'$.
We fix arbitrary integers $i,i'\in [n],j\in [m],h,h'\in \{1,2\}$ and $r,r'\in \{1,2,3\}$.
Suppose that $x_1\in P^{h}(i,j,a_r)$, $x_2\in P^{h}(i,j,b_r)$, $x_3\in P^{h}(i,j,c_r)$ and $x_4\in P^{h}(i,j,p_i^{3-h})$.
Suppose that $y_1\in P(\pi_{i'}^{h'},a_{r'})$ and $y_2\in P(\pi_{i'}^{h'},c_{r'})$.
For a vertex $x_{\mu}$ for $\mu\in \{1,2,3,4\}$, let $\ell_{x_{\mu}}=\dist(\pi_i^h,x_{\mu})$.
For a vertex $y_1$, let $\ell_{y_1}=\dist(a_{r'},y_1)$.
For a vertex $y_2$, let $\ell_{y_2}=\dist(c_{r'},y_2)$.
Let $|P(s_i^j,a_{r'})|=20(n+1)+10\lambda$ for some $\lambda\in [n]$.
For a pair $\{x_1,y_1\}$, $\dist(f(s_i^j,a_{r'}),y_1)=2+\ell_{y_1}$.
$\dist(f(s_i^j,a_{r'}),x_1)=\text{min }(2+|P(\pi_i^h,a_{r'})|+\ell_{x_1},|P(s_i^j,a_{r'})|-1+|P^{h}(i,j,a_r)|-\ell_{x_1})=
\text{min }(2+10(n+1)+\ell_{x_1},40(n+1)+10\lambda-1-\ell_{x_1})$ if $r=r'$ and $\dist(f(s_i^j,a_{r'}),x_1)=\text{min }(2+|P(\pi_i^h,a_{r'})|+\ell_{x_1},|P(s_i^j,a_{r'})|+1+|P^{h}(i,j,a_r)|-\ell_{x_1})=
\text{min }(2+10(n+1)+\ell_{x_1},40(n+1)+10\lambda+1-\ell_{x_1})\geq 2+10(n+1)\geq \dist(f(s_i^j,a_{r'}),y_1)$ if $r\neq r'$.
It follows that $\dist(f(s_i^j,a_{r'}),x_1)\geq 2+10(n+1)\geq \dist(f(s_i^j,a_{r'}),y_1)$.
$\dist(f(s_i^j,a_{r'}),x_1)=\dist(f(s_i^j,a_{r'}),y_1)$ only when $x_1=\pi_i^h$ and $y_1=\pi_{i'}^{h'}$.
If $i'\neq i$ or $h'\neq h$, obviously the pair $\{\pi_i^h,\pi_{i'}^{h'}\}$ is resolved by $f^h(i,j,a_r)$.
Thus every pair $\{x_1,y_1\}$ is resolved by $f(s_i^j,a_{r'})$ or $f^h(i,j,a_r)$.
In a similar way, we can show that every vertex pair $\{x_2,y_1\}$, $\{x_3,y_1\}$ and $\{x_4,y_1\}$ are resolved by $f(s_i^j,a_{r'})$ or $f^h(i,j,a_r)$.
Also in a similar way, we can show that every vertex pair $\{x_1,y_2\}$, $\{x_2,y_2\}$, $\{x_3,y_2\}$ and $\{x_4,y_2\}$ are resolved by $f(s_i^j,c_{r'})$ or $f^h(i,j,a_r)$.
This completes the proof for the lemma.
\end{proof}

\begin{lemma} \label{PixR}
Every pair $\{x,y\}\in \bigcup_{i\in [n]}\Pi_i\times \bigcup_{r\in \{1,2,3\}}R_r$ is resolved by $S'$.
\end{lemma}
\begin{proof}
We show that every pair $\{x,y\}\in (P^{h}(i,j,a_r)\cup P^{h}(i,j,b_r)\cup P^{h}(i,j,c_r)\cup P^{h}(i,j,p_i^{3-h}))\times (P(u_{r'}^{i'},a_{r'})\cup P(v_{r'}^{i'},a_{r'})\cup P(u_{r'}^{i'},c_{r'})\cup P(v_{r'}^{i'},c_{r'}))$ for $i,i'\in [n],j\in [m],h\in \{1,2\}$ and $r,r'\in \{1,2,3\}$ is resolved by $S'$.
We fix arbitrary integers $i,i'\in [n],j\in [m],h\in \{1,2\}$ and $r,r'\in \{1,2,3\}$.
Suppose that $x_1\in P^{h}(i,j,a_r)$, $x_2\in P^{h}(i,j,b_r)$, $x_3\in P^{h}(i,j,c_r)$ and $x_4\in P^{h}(i,j,p_i^{3-h})$.
Suppose that $y_1\in P(u_{r'}^{i'},a_{r'})$,  $y_2\in P(v_{r'}^{i'},a_{r'})$, $z_1\in P(u_{r'}^{i'},c_{r'})$ and $z_2\in P(v_{r'}^{i'},c_{r'})$.
For a vertex $x_{\mu}$ for $\mu\in \{1,2,3,4\}$, let $\ell_{x_{\mu}}=\dist(\pi_i^h,x_{\mu})$.
For a vertex $y_{\nu}$ for $\nu\in \{1,2\}$, let $\ell_{y_{\nu}}=\dist(a_{r'},y_{\nu})$.
For a vertex $z_{\eta}$ for $\eta\in \{1,2\}$, let $\ell_{y_{\eta}}=\dist(c_{r'},z_{\eta})$.
Then $|P(u_{r'}^{i'},a_{r'})|=|P(v_{r'}^{i'},a_{r'})|=20(n+1)-10i'$ and $|P(u_{r'}^{i'},c_{r'})|=|P(v_{r'}^{i'},c_{r'})|=20(n+1)+10i'$.
For a pair $\{x_1,y_1\}$, $\dist(f^h(i,j,a_r),x_1)=\ell_{x_1}$ and $\dist(f^h(i,j,a_r),\pi_i^h)=2$.
$\dist(f^h(i,j,a_r),y_1)=\dist(f^h(i,j,b_r),y_1)=\dist(f^h(i,j,c_r),y_1)=\dist(f^h(i,j,p_i^{3-h}),y_1)=2+|P(\pi_i^h,a_{r'})|+\ell_{y_1}=2+10(n+1)+\ell_{y_1}$.
For the pair $\{x_1,y_1\}$ that is not resolved by $f^h(i,j,a_r)$, $\dist(f^h(i,j,a_r),x_1)=\dist(f^h(i,j,a_r),y_1)=\dist(f^h(i,j,b_r),y_1)=\dist(f^h(i,j,b_r),x_1)-2$.
Thus every pair $\{x_1,y_1\}$ is resolved by $f^h(i,j,a_r)$ or $f^h(i,j,b_r)$.
In a similar way, we can show that every vertex pair $\{x_{\mu},y_{\nu}\}$ for $\mu\in \{1,2,3,4\},\nu\in \{1,2\}$ is resolved by $S'$.
For a pair $\{x_1,z_1\}$, $\dist(f^h(i,j,a_r),x_1)=\ell_{x_1}$ and $\dist(f^h(i,j,a_r),\pi_i^h)=2$.
$\dist(f^h(i,j,a_r),z_1)=\dist(f^h(i,j,b_r),z_1)=
\text{min }(2+|P(\pi_i^h,c_{r'})|+\ell_{z_1},2+|P(\pi_i^h,a_{r'})|+|P(a_{r'},u_{r'}^{i'})|+|P(c_{r'},u_{r'}^{i'})|-2-\ell_{z_1})=
\text{min }(2+10(n+1)+\ell_{z_1},50(n+1)-\ell_{z_1})$ if $|P(c_{r'},u_{r'}^{i'})|-\ell_{z_1}\geq 2$.
$\dist(f^h(i,j,a_r),z_1)=\dist(f^h(i,j,b_r),z_1)=2+|P(\pi_i^h,a_{r'})|+|P(a_{r'},u_{r'}^{i'})|+|P(c_{r'},u_{r'}^{i'})|-\ell_{z_1}$ if $|P(c_{r'},u_{r'}^{i'})|-\ell_{z_1}<2$.
For the pair $\{x_1,z_1\}$ that is not resolved by $f^h(i,j,a_r)$, $\dist(f^h(i,j,a_r),x_1)=\dist(f^h(i,j,a_r),z_1)=\dist(f^h(i,j,b_r),z_1)=\dist(f^h(i,j,b_r),x_1)-2$.
Thus every pair $\{x_1,z_1\}$ is resolved by $f^h(i,j,a_r)$ or $f^h(i,j,b_r)$.
In a similar way, we can show that every vertex pair $\{x_{\mu},z_{\nu}\}$ for $\mu\in \{1,2,3,4\},\nu\in \{1,2\}$ is resolved by $S'$.

Then we show that every pair $\{x,y\}\in (P^{h}(i,j,a_r)\cup P^{h}(i,j,b_r)\cup P^{h}(i,j,c_r)\cup P^{h}(i,j,p_i^{3-h}))\times (P(u_{r'}^{i'},b_{r'})\cup P(v_{r'}^{i'},b_{r'}))$ for $i,i'\in [n],j\in [m],h\in \{1,2\}$ and $r,r'\in \{1,2,3\}$ is resolved by $S'$.
We fix arbitrary integers $i,i'\in [n],j\in [m],h\in \{1,2\}$ and $r,r'\in \{1,2,3\}$.
Suppose that $x_1\in P^{h}(i,j,a_r)$, $x_2\in P^{h}(i,j,b_r)$, $x_3\in P^{h}(i,j,c_r)$ and $x_4\in P^{h}(i,j,p_i^{3-h})$.
Suppose that $w_1\in P(u_{r'}^{i'},b_{r'})$ and $w_2\in P(v_{r'}^{i'},b_{r'})$.
For a vertex $x_{\mu}$ for $\mu\in \{1,2,3,4\}$, let $\ell_{x_{\mu}}=\dist(\pi_i^h,x_{\mu})$.
For a vertex $w_{\nu}$ for $\nu\in \{1,2\}$, let $\ell_{w_{\nu}}=\dist(b_{r'},w_{\nu})$.
Then let $|P(s_i^j,a_{r'})|=20(n+1)+10\lambda$ for some $\lambda\in [n]$, $|P(u_{r'}^{i'},b_{r'})|=20(n+1)-5i'-1$ and $|P(v_{r'}^{i'},b_{r'})|=20(n+1)-5i'-2$.
For a pair $\{x_1,w_1\}$, $\dist(f^1(u_{r'}^{i'},v_{r'}^{i'}),w_1)=\dist(f^2(u_{r'}^{i'},v_{r'}^{i'}),w_1)=2+|P(u_{r'}^{i'},b_{r'})|-\ell_{w_1}$.
For the distance between $f^{\eta}(u_{r'}^{i'},v_{r'}^{i'})$ and $x_1$ for $\eta\in \{1,2\}$, there are two cases.
Case 1: $\lambda\leq i'$.
$\dist(f^1(u_{r'}^{i'},v_{r'}^{i'}),x_1)=\dist(f^2(u_{r'}^{i'},v_{r'}^{i'}),x_1)+1=\text{min }(1+|P(a_{r'},u_{r'}^{i'})|+|P(s_i^j,a_{r'})|-1+|P^h(i,j,a_r)|-\ell_{x_1},1+|P(a_{r'},u_{r'}^{i'})|+|P(\pi_i^h,a_{r'})|+\ell_{x_1})$ if $r=r'$.
$\dist(f^1(u_{r'}^{i'},v_{r'}^{i'}),x_1)=\dist(f^2(u_{r'}^{i'},v_{r'}^{i'}),x_1)+1=\text{min }(1+|P(a_{r'},u_{r'}^{i'})|+|P(s_i^j,a_{r'})|+1+|P^h(i,j,a_r)|-\ell_{x_1},1+|P(a_{r'},u_{r'}^{i'})|+|P(\pi_i^h,a_{r'})|+\ell_{x_1})$ if $r\neq r'$.
In this case, for the pair $\{x_1,w_1\}$ which is not resolved by $f^1(u_{r'}^{i'},v_{r'}^{i'})$, $\dist(f^2(u_{r'}^{i'},v_{r'}^{i'}),x_1)=\dist(f^1(u_{r'}^{i'},v_{r'}^{i'}),x_1)-1<\dist(f^1(u_{r'}^{i'},v_{r'}^{i'}),y)=\dist(f^2(u_{r'}^{i'},v_{r'}^{i'}),y)$.
Case 2: $\lambda>i'$.
$\dist(f^1(u_{r'}^{i'},v_{r'}^{i'}),x_1)=\dist(f^2(u_{r'}^{i'},v_{r'}^{i'}),x_1)+1=\text{min }(1+|P(c_{r'},u_{r'}^{i'})|+|P(s_i^j,c_{r'})|+1+|P^h(i,j,a_r)|-\ell_{x_1},1+|P(a_{r'},u_{r'}^{i'})|+|P(\pi_i^h,a_{r'})|+\ell_{x_1})$.
In this case, for the pair $\{x_1,w_1\}$ which is not resolved by $f^1(u_{r'}^{i'},v_{r'}^{i'})$, $\dist(f^2(u_{r'}^{i'},v_{r'}^{i'}),x_1)=\dist(f^1(u_{r'}^{i'},v_{r'}^{i'}),x_1)-1<\dist(f^1(u_{r'}^{i'},v_{r'}^{i'}),y)=\dist(f^2(u_{r'}^{i'},v_{r'}^{i'}),y)$.
In a similar way, we can show that every vertex pair $\{x_{\mu},w_{\nu}\}$ for $\mu\in \{1,2,3,4\},\nu\in \{1,2\}$ is resolved by $S'$.
This completes the proof for the lemma.
\end{proof}

\begin{lemma} \label{LxH}
Every pair $\{x,y\}\in \bigcup_{i\in [n]}L_i\times \bigcup_{i\in [n]}H_i$ is resolved by $S'$.
\end{lemma}
\begin{proof}
First we show that every pair $\{x,y\}\in P(q_i^h,\text{mid}(P^{3-h}(i,j,p_i^{h})))\times (P(s_{i'}^{j'},a_r)\cup P(s_{i'}^{j'},c_r))$ for $i,i'\in [n],j,j'\in [m],h\in \{1,2\}$ and $r\in \{1,2,3\}$ is resolved by $S'$.
We fix arbitrary integers $i,i'\in [n],j,j'\in [m],h\in \{1,2\}$ and $r\in \{1,2,3\}$.
Suppose that $x\in P(q_i^h,\text{mid}(P^{3-h}(i,j,p_i^{h})))$ and $y\in P(s_{i'}^{j'},a_r)$.
For a vertex $x\in P(q_i^h,\text{mid}(P^{3-h}(i,j,p_i^{h})))$, let $P(q_i^h,x)$ be the subpath of $P(q_i^h,\text{mid}(P^{3-h}(i,j,p_i^{h})))$ from $q_i^h$ to $x$ and $|P(q_i^h,x)|=\ell_x$.
For a vertex $y\in P(s_{i'}^{j'},a_r)$, let $P(a_r,y)$ be the subpath of $P(s_{i'}^{j'},a_r)$ from $a_r$ to $y$ and $|P(a_r,y)|=\ell_y$.
Then $\dist(f(\pi_i^h,a_r),x)=\text{min }(|P(\pi_i^h,a_r)|+1+\ell_x,2+|P(\pi_i^{3-h},a_r)|+|P^{3-h}(i,j,p_i^h)|/2+|P(q_i^h,\text{mid}(P^{3-h}(i,j,p_i^{h})))|-\ell_x)=
\text{min }(10(n+1)+1+\ell_x,50(n+1)+1-\ell_x)$.
$\dist(f(\pi_i^{3-h},a_r),x)=\text{min }(2+|P(\pi_i^h,a_r)|+1+\ell_x,|P(\pi_i^{3-h},a_r)|+|P^{3-h}(i,j,p_i^h)|/2+|P(q_i^h,\text{mid}(P^{3-h}(i,j,p_i^{h})))|-\ell_x)=
\text{min }(10(n+1)+3+\ell_x,50(n+1)-1-\ell_x)$.
$\dist(f(\pi_i^h,a_r),y)=\dist(f(\pi_i^{3-h},a_r),y)=2+\ell_y$.
For a pair $\{x,y\}$ which is not resolved by $f(\pi_i^h,a_r)$, there are two cases.
Case 1: $\dist(f(\pi_i^{3-h},a_r),x)=\dist(f(\pi_i^{3-h},a_r),y)=10(n+1)+1+\ell_x=2+\ell_y$.
In this case, $\dist(f(\pi_i^{3-h},a_r),x)=10(n+1)+3+\ell_x>2+\ell_y=\dist(f(\pi_i^{3-h},a_r),y)$.
Case 2: $\dist(f(\pi_i^{3-h},a_r),x)=\dist(f(\pi_i^{3-h},a_r),y)=50(n+1)+1-\ell_x=2+\ell_y$.
In this case, $\dist(f(\pi_i^{3-h},a_r),x)=50(n+1)-1-\ell_x<2+\ell_y=\dist(f(\pi_i^{3-h},a_r),y)$.
It follows that every pair $\{x,y\}$ is resolved by $f(\pi_i^h,a_r)$ or $f(\pi_i^{3-h},a_r)$.
Similarly we can show that every pair of $P(q_i^h,\text{mid}(P^{3-h}(i,j,p_i^{h})))\times P(s_{i'}^{j'},c_r)$ is resolved by $f(\pi_i^h,c_r)$ or $f(\pi_i^{3-h},c_r)$.

Then we show that every pair $\{x,y\}\in P(q_i^h,\text{mid}(P^{3-h}(i,j,p_i^{h})))\times (P(s_{i'}^{j'},b_r)\setminus \{s_{i'}^{j'}\})$ for $i,i'\in [n],j,j'\in [m],h\in \{1,2\}$ and $r\in \{1,2,3\}$ is resolved by $S'$.
We fix arbitrary integers $i,i'\in [n],j,j'\in [m],h\in \{1,2\}$ and $r\in \{1,2,3\}$.
Suppose that $x\in P(q_i^h,\text{mid}(P^{3-h}(i,j,p_i^{h})))$ and $y\in P(s_{i'}^{j'},b_r)\setminus \{s_{i'}^{j'}\}$.
We define $\ell_x$ and $\ell_y$ in a similar way to that of $\ell_x$ and $\ell_y$ in the first paragraph.
Suppose that $s_{i'}^{j'}$ resolves the pair $\{u_r^{i_r},v_r^{i_r}\}$ for some $i_r\in [n]$, i.e. $|P(a_r,u_r^{i_r})|=20(n+1)-10i_r$.
Then $\dist(f^1(u_r^{i_r},v_r^{i_r}),x)=\dist(f^2(u_r^{i_r},v_r^{i_r}),x)+1=\text{min }(1+|P(a_r,u_r^{i_r})|+|P(\pi_i^h,a_r)|+1+\ell_x,1+|P(a_r,u_r^{i_r})|+|P(\pi_i^{3-h},a_r)|+|P^{3-h}(i,j,p_i^h)|/2+|P(q_i^h,\text{mid}(P^{3-h}(i,j,p_i^{h})))|-\ell_x)=
\text{min }(2+30(n+1)-10i_r+\ell_x,70(n+1)-10i_r-\ell_x)$.
$\dist(f^1(u_r^{i_r},v_r^{i_r}),y)=\dist(f^2(u_r^{i_r},v_r^{i_r}),y)=2+|P(b_r,v_r^{i_r})|+\ell_y=20(n+1)-5i_r+\ell_y$.
Thus for a vertex pair $\{x,y\}$ which is not resolved by $f^1(u_r^{i_r},v_r^{i_r})$, $\dist(f^1(u_r^{i_r},v_r^{i_r}),x)=\dist(f^1(u_r^{i_r},v_r^{i_r}),y)=\dist(f^2(u_r^{i_r},v_r^{i_r}),y)>\dist(f^2(u_r^{i_r},v_r^{i_r}),x)$.
It follows that every pair $\{x,y\}$ is resolved by $f^1(u_r^{i_r},v_r^{i_r})$ or $f^2(u_r^{i_r},v_r^{i_r})$.
This completes the proof for the lemma.
\end{proof}

\begin{lemma} \label{LxS}
Every pair $\{x,y\}\in \bigcup_{i\in [n]}L_i\times \bigcup_{i\in [n]}S_i$ is resolved by $S'$.
\end{lemma}
\begin{proof}
We show that every pair $\{x,y\}\in P(q_i^h,\text{mid}(P^{3-h}(i,j,p_i^{h})))\times (P(\pi_{i'}^{h'},a_r)\cup P(\pi_{i'}^{h'},c_r))$ for $i,i'\in [n],j\in [m],h,h'\in \{1,2\}$ and $r\in \{1,2,3\}$ is resolved by $S'$.
We fix arbitrary integers $i,i'\in [n],j\in [m],h,h'\in \{1,2\}$ and $r\in \{1,2,3\}$.
Suppose that $x\in P(q_i^h,\text{mid}(P^{3-h}(i,j,p_i^{h})))$ and $y\in P(\pi_{i'}^{h'},a_r)$.
For a vertex $x\in P(q_i^h,\text{mid}(P^{3-h}(i,j,p_i^{h})))$, let $P(q_i^h,x)$ be the subpath of $P(q_i^h,\text{mid}(P^{3-h}(i,j,p_i^{h})))$ from $q_i^h$ to $x$ and $|P(q_i^h,x)|=\ell_x$.
For a vertex $y\in P(\pi_{i'}^{h'},a_r)$, let $P(a_r,y)$ be the subpath of $P(\pi_{i'}^{h'},a_r)$ from $a_r$ to $y$ and $|P(a_r,y)|=\ell_y$.
Then $\dist(f(s_i^j,a_r),y)=2+\ell_y\leq 2+10(n+1)$.
$\dist(f(s_i^j,a_r),x)=\text{min }(2+|P(\pi_i^h,a_r)|+1+\ell_x,2+|P(\pi_i^{3-h},a_r)|+|P^{3-h}(i,j,p_i^h)|/2+|P(q_i^h,\text{mid}(P^{3-h}(i,j,p_i^{h})))|-\ell_x)=
\text{min }(10(n+1)+3+\ell_x,50(n+1)+1-\ell_x)\geq 3+10(n+1)> \dist(f(s_i^j,a_r),x)$.
Thus every pair $\{x,y\}$ is resolved by $f(s_i^j,a_r)$.
Similarly we can show that every pair of $P(q_i^h,\text{mid}(P^{3-h}(i,j,p_i^{h})))\times P(\pi_{i'}^{h'},c_r)$ is resolved by $f(s_i^j,c_r)$.
This completes the proof for the lemma.
\end{proof}

\begin{lemma} \label{LxR}
Every pair $\{x,y\}\in \bigcup_{i\in [n]}L_i\times \bigcup_{r\in \{1,2,3\}}R_r$ is resolved by $S'$.
\end{lemma}
\begin{proof}
We show that every pair $\{x,y\}\in P(q_i^h,\text{mid}(P^{3-h}(i,j,p_i^{h})))\times (P(u_r^{i'},a_r)\cup P(v_r^{i'},a_r)\cup P(u_r^{i'},b_r)\cup P(v_r^{i'},b_r)\cup P(u_r^{i'},c_r)\cup P(v_r^{i'},c_r))$ for $i,i'\in [n],j\in [m],h\in \{1,2\}$ and $r\in \{1,2,3\}$ is resolved by $S'$.
We fix arbitrary integers $i,i'\in [n],j\in [m],h\in \{1,2\}$ and $r\in \{1,2,3\}$.
Suppose that $x\in P(q_i^h,\text{mid}(P^{3-h}(i,j,p_i^{h})))$, $y_1\in P(u_r^{i'},a_r)$, $y_2\in P(v_r^{i'},a_r)$,
$z_1\in P(u_r^{i'},b_r)$, $z_2\in P(v_r^{i'},b_r)$, $w_1\in P(u_r^{i'},c_r)$, $w_2\in P(v_r^{i'},c_r)$.
For a vertex $x\in P(q_i^h,\text{mid}(P^{3-h}(i,j,p_i^{h})))$, let $P(q_i^h,x)$ be the subpath of $P(q_i^h,\text{mid}(P^{3-h}(i,j,p_i^{h})))$ from $q_i^h$ to $x$ and $|P(q_i^h,x)|=\ell_x$.
For a vertex $y_1\in P(u_r^{i'},a_r)$, let $P(y_1,u_r^{i'})$ be the subpath of $P(u_r^{i'},a_r)$ from $y_1$ to $u_r^{i'}$ and let $|P(y_1,u_r^{i'})|=\ell_{y_1}$.
For a vertex $y_2\in P(v_r^{i'},a_r)$, let $P(y_2,v_r^{i'})$ be the subpath of $P(v_r^{i'},a_r)$ from $y_2$ to $v_r^{i'}$ and let $|P(y_2,v_r^{i'})|=\ell_{y_2}$.
For a vertex $z_1\in P(u_r^{i'},b_r)$, let $P(z_1,u_r^{i'})$ be the subpath of $P(u_r^{i'},b_r)$ from $z_1$ to $u_r^{i'}$ and let $|P(z_1,u_r^{i'})|=\ell_{z_1}$.
For a vertex $z_2\in P(v_r^{i'},b_r)$, let $P(z_2,v_r^{i'})$ be the subpath of $P(v_r^{i'},b_r)$ from $z_2$ to $v_r^{i'}$ and let $|P(z_2,v_r^{i'})|=\ell_{z_2}$.
For a vertex $w_1\in P(u_r^{i'},c_r)$, let $P(w_1,u_r^{i'})$ be the subpath of $P(u_r^{i'},c_r)$ from $w_1$ to $u_r^{i'}$ and let $|P(w_1,u_r^{i'})|=\ell_{w_1}$.
For a vertex $w_2\in P(v_r^{i'},c_r)$, let $P(w_2,v_r^{i'})$ be the subpath of $P(v_r^{i'},c_r)$ from $w_2$ to $v_r^{i'}$ and let $|P(w_2,v_r^{i'})|=\ell_{w_2}$.
For a pair $\{x,y_1\}$, $\dist(f^1(u_r^{i'},v_r^{i'}),y_1)=1+\ell_{y_1}$ if $y_1\neq u_r^{i'}$ and $\dist(f^1(u_r^{i'},v_r^{i'}),u_r^{i'})=2$.
$\dist(f^1(u_r^{i'},v_r^{i'}),x)=\text{min }(1+|P(a_r,u_r^{i'})|+|P(\pi_i^h,a_r)|+1+\ell_x,1+|P(a_r,u_r^{i'})|+|P(\pi_i^{3-h},a_r)|+|P^{3-h}(i,j,p_i^h)|/2+|P(q_i^h,\text{mid}(P^{3-h}(i,j,p_i^{h})))|-\ell_x)=
\text{min }(2+30(n+1)-10i'+\ell_x,70(n+1)-10i'-\ell_x)>1+20(n+1)-10i'\geq \dist(f^1(u_r^{i'},v_r^{i'}),y_1)$.
Thus every pair $\{x,y_1\}$ is resolved by $f^1(u_r^{i'},v_r^{i'})$.
Similarly, every pair $\{x,y_2\}$ is resolved by $f^1(u_r^{i'},v_r^{i'})$.
For a pair $\{x,z_1\}$, $\dist(f^2(u_r^{i'},v_r^{i'}),x)=\dist(f^1(u_r^{i'},v_r^{i'}),x)-1$.
$\dist(f^1(u_r^{i'},v_r^{i'}),z_1)=\dist(f^2(u_r^{i'},v_r^{i'}),z_1)=2+\ell_{z_1}$.
Thus for a pair $\{x,z_1\}$ that is not resolved by $f^1(u_r^{i'},v_r^{i'})$,
$\dist(f^1(u_r^{i'},v_r^{i'}),x)=\dist(f^1(u_r^{i'},v_r^{i'}),z_1)=\dist(f^2(u_r^{i'},v_r^{i'}),z_1)>\dist(f^2(u_r^{i'},v_r^{i'}),x)$.
It follows that every pair $\{x,z_1\}$ is resolved by $f^1(u_r^{i'},v_r^{i'})$ or $f^2(u_r^{i'},v_r^{i'})$.
Similarly, every pair $\{x,z_2\}$ is resolved by $f^1(u_r^{i'},v_r^{i'})$ or $f^2(u_r^{i'},v_r^{i'})$.
For a pair $\{x,w_1\}$, $\dist(f(\pi_i^h,c_r),w_1)=\dist(f(\pi_i^{3-h},c_r),w_1)=\dist(f(s_i^j,c_r),w_1)=2+|P(c_r,u_r^{i'})|-\ell_{w_1}=2+20(n+1)+10i'-\ell_{w_1}$.
$\dist(f(\pi_i^h,c_r),x)=\text{min }(|P(\pi_i^h,a_r)|+1+\ell_x,2+|P(\pi_i^{3-h},a_r)|+|P^{3-h}(i,j,p_i^h)|/2+|P(q_i^h,\text{mid}(P^{3-h}(i,j,p_i^{h})))|-\ell_x)=
\text{min }(10(n+1)+1+\ell_x,50(n+1)+1-\ell_x)$.
$\dist(f(\pi_i^{3-h},c_r),x)=\text{min }(2+|P(\pi_i^h,a_r)|+1+\ell_x,|P(\pi_i^{3-h},a_r)|+|P^{3-h}(i,j,p_i^h)|/2+|P(q_i^h,\text{mid}(P^{3-h}(i,j,p_i^{h})))|-\ell_x)=
\text{min }(10(n+1)+3+\ell_x,50(n+1)-1-\ell_x)$.
$\dist(f(s_i^j,c_r),x)=\text{min }(2+|P(\pi_i^h,a_r)|+1+\ell_x,2+|P(\pi_i^{3-h},a_r)|+|P^{3-h}(i,j,p_i^h)|/2+|P(q_i^h,\text{mid}(P^{3-h}(i,j,p_i^{h})))|-\ell_x)=
\text{min }(10(n+1)+3+\ell_x,50(n+1)+1-\ell_x)$.
For a pair $\{x,w_1\}$ which is not resolved by $f(s_i^j,c_r)$, either $f(\pi_i^h,c_r)$ or $f(\pi_i^{3-h},c_r)$ resolves it.
Thus every pair $\{x,w_1\}$ is resolved by $f(s_i^j,c_r)$, $f(\pi_i^h,c_r)$ or $f(\pi_i^{3-h},c_r)$.
Similarly, we can show that every pair $\{x,w_2\}$ is resolved by $f(s_i^j,c_r)$, $f(\pi_i^h,c_r)$ or $f(\pi_i^{3-h},c_r)$.
This completes the proof for the lemma.
\end{proof}

\begin{lemma} \label{SxH}
Every pair $\{x,y\}\in \bigcup_{i\in [n]}S_i\times \bigcup_{i\in [n]}H_i$ is resolved by $S'$.
\end{lemma}
\begin{proof}
We show that every pair $\{x,y\}\in (P(\pi_{i'}^h,a_{r'})\cup P(\pi_{i'}^h,c_{r'}))\times (P(s_i^j,a_r)\cup P(s_i^j,b_r)\cup P(s_i^j,c_r))$ for $i,i'\in [n],j\in [m],h\in \{1,2\}$ and $r,r'\in \{1,2,3\}$ is resolved by $S'$.
We fix arbitrary integers $i,i'\in [n],j\in [m],h\in \{1,2\}$ and $r,r'\in \{1,2,3\}$.
Suppose that $x_1\in P(\pi_{i'}^h,a_{r'})$, $x_2\in P(\pi_{i'}^h,c_{r'})$, $y_1\in P(s_i^j,a_r)$, $y_2\in P(s_i^j,a_r)$ and $y_3\in P(s_i^j,c_r)$.
For a vertex $x_1\in P(\pi_{i'}^h,a_{r'})$, let $P(\pi_{i'}^h,x_1)$ be the subpath of $P(\pi_{i'}^h,a_{r'})$ from $\pi_{i'}^h$ to $x_1$ and let $|P(\pi_{i'}^h,x_1)|=\ell_{x_1}$.
For a vertex $x_2\in P(\pi_{i'}^h,c_{r'})$, let $P(\pi_{i'}^h,x_2)$ be the subpath of $P(\pi_{i'}^h,c_{r'})$ from $\pi_{i'}^h$ to $x_2$ and let $|P(\pi_{i'}^h,x_2)|=\ell_{x_2}$.
For a vertex $y_1\in P(s_i^j,a_r)$, let $P(s_i^j,y_1)$ be the subpath of $P(s_i^j,a_r)$, from $s_i^j$ to $y_1$ and let $|P(s_i^j,y_1)|=\ell_{y_1}$.
For a vertex $y_2\in P(s_i^j,b_r)\setminus \{s_i^j\}$, let $P(s_i^j,y_2)$ be the subpath of $P(s_i^j,b_r)$, from $s_i^j$ to $y_2$ and let $|P(s_i^j,y_2)|=\ell_{y_2}$.
For a vertex $y_3\in P(s_i^j,c_r)$, let $P(s_i^j,y_3)$ be the subpath of $P(s_i^j,c_r)$, from $s_i^j$ to $y_3$ and let $|P(s_i^j,y_3)|=\ell_{y_3}$.
Let $|P(s_i^j,a_r)|=20(n+1)+10\lambda$ for some $\lambda\in [n]$.
For a vertex pair $\{x_1,y_1\}$, $\dist(f^h(i',j,a_r),x_1)=2+\ell_{x_1}$.
For the distance between $f^h(i',j,a_r)$ and $y_1$, there are two cases.
Case 1: $i'=i$.
Then $\dist(f^h(i',j,a_r),y_1)=\text{min }(|P^h(i,j,a_r)|+\ell_{y_1}-1,2+|P(\pi_i^h,a_r)|+|P(s_i^j,a_r)|-\ell_{y_1})=\text{min }(20(n+1)+\ell_{y_1}-1,30(n+1)+10\lambda+2-\ell_{y_1})$ if $y_1\neq s_i^j$ and
$\dist(f^h(i',j,a_r),s_i^j)=20(n+1)+1$.
Thus $\dist(f^h(i',j,a_r),y_1)\geq 2+10(n+1)\geq \dist(f^h(i',j,a_r),x_1)$.
$f^h(i',j,a_r)$ does not resolve $\{x_1,y_1\}$ only when $x_1=a_{r'}$ and $y_1=a_r$ with $r\neq r'$.
The pair $\{a_{r'},a_r\}$ is resolved by $f(\pi_{i'}^h,a_{r'})$.
Thus in this case, every pair $\{x_1,y_1\}$ is resolved by $f^h(i',j,a_r)$ or $f(\pi_{i'}^h,a_{r'})$.
Case 2: $i'\neq i$.
Then $\dist(f^h(i',j,a_r),s_i^j)=\text{min}_{d\in \{1,2,3\}}(2+|P(\pi_{i'}^h,c_{d})|+|P(s_i^j,c_{d})|)$.
$\dist(f^h(i',j,a_r),y_1)=\text{min }(\dist(f^h(i',j,a_r),s_i^j)+\ell_{y_1},2+|P(\pi_{i'}^h,a_r)|+|P(s_i^j,a_r)|-\ell_{y_1})\geq 2+10(n+1)\geq \dist(f^h(i',j,a_r),x_1)$.
$f^h(i',j,a_r)$ does not resolve $\{x_1,y_1\}$ only when $x_1=a_{r'}$ and $y_1=a_r$ with $r\neq r'$..
The pair $\{a_{r'},a_r\}$ is resolved by $f(\pi_{i'}^h,a_{r'})$.
It follows that every pair $\{x_1,y_1\}$ is resolved by $f^h(i',j,a_r)$ or $f(\pi_{i'}^h,a_{r'})$.
In a similar way, we can show that every pair $\{x_2,y_1\}$, $\{x_1,y_3\}$ and $\{x_2,y_3\}$ are resolved by $S'$.
For a vertex pair $\{x_1,y_2\}$, $\dist(f^h(i',j,b_r),x_1)=2+\ell_{x_1}$.
For the distance between $f^h(i',j,b_r)$ and $y_2$, there are two cases.
Case 1: $i'=i$.
Then $\dist(f^h(i',j,b_r),y_2)=\text{min }(|P^h(i',j,b_r)|+\ell_{y_2}-1,2+|P(\pi_i^h,a_r)|+|P(a_r,v_r^n)|+|P(b_r,v_r^n)|+|P(s_i^j,b_r)|-\ell_{y_2})\geq 20(n+1)>\dist(f^h(i',j,b_r),x_1)$.
Case 2: $i'\neq i$.
Then $\dist(f^h(i',j,b_r),s_i^j)=\text{min}_{d\in \{1,2,3\}}(2+|P(\pi_{i'}^h,c_{d})|+|P(s_i^j,c_{d})|)$.
$\dist(f^h(i',j,b_r),y_2)=\text{min }(\dist(f^h(i',j,b_r),s_i^j)+\ell_{y_2},2+|P(\pi_{i'}^h,a_r)|+|P(a_r,v_r^n)|+|P(b_r,v_r^n)|+|P(s_i^j,b_r)|-\ell_{y_2})>20(n+1)>\dist(f^h(i',j,b_r),x_1)$.
Thus in both cases, every pair $\{x_1,y_2\}$ is resolved by $f^h(i',j,b_r)$.
In a similar way, we can show that every pair $\{x_2,y_2\}$ is resolved by $f^h(i',j,b_r)$.
This completes the proof for the lemma.
\end{proof}

\begin{lemma} \label{SxR}
Every pair $\{x,y\}\in \bigcup_{i\in [n]}S_i\times \bigcup_{r\in \{1,2,3\}}R_r$ is resolved by $S'$.
\end{lemma}
\begin{proof}
We show that every pair $\{x,y\}\in (P(\pi_i^h,a_r)\cup P(\pi_i^h,c_r))\times (P(u_{r'}^{i'},a_{r'})\cup P(v_{r'}^{i'},a_{r'})\cup P(u_{r'}^{i'},b_{r'})\cup P(v_{r'}^{i'},b_{r'})\cup P(u_{r'}^{i'},c_{r'})\cup P(v_{r'}^{i'},c_{r'}))$ for $i,i'\in [n],h\in \{1,2\}$ and $r,r'\in \{1,2,3\}$ is resolved by $S'$.
We fix arbitrary integers $i,i'\in [n],j\in [m],h\in \{1,2\}$ and $r,r'\in \{1,2,3\}$.
Suppose that $x_1\in P(\pi_i^h,a_r)$, $x_2\in P(\pi_i^h,c_r)$, $y_1\in P(u_{r'}^{i'},a_{r'})$, $y_2\in P(v_{r'}^{i'},a_{r'})$,
$z_1\in P(u_{r'}^{i'},b_{r'})$, $z_2\in P(v_{r'}^{i'},b_{r'})$, $w_1\in P(u_{r'}^{i'},c_{r'})$, $w_2\in P(v_{r'}^{i'},c_{r'})$.
For a vertex $x_1\in P(\pi_i^h,a_r)$, let $P(\pi_i^h,x_1)$ be the subpath of $P(\pi_i^h,a_r)$ from $\pi_i^h$ to $x_1$ and let $|P(\pi_i^h,x_1)|=\ell_{x_1}$.
For a vertex $x_2\in P(\pi_i^h,c_r)$, let $P(\pi_i^h,x_2)$ be the subpath of $P(\pi_i^h,c_r)$ from $\pi_i^h$ to $x_2$ and let $|P(\pi_i^h,x_2)|=\ell_{x_2}$.
Then $\dist(f^h(i,j,a_r),x_1)=2+\ell_{x_1}\leq 2+10(n+1)$ and $\dist(f^h(i,j,a_r),x_2)=2+\ell_{x_2}\leq 2+10(n+1)$.
$\dist(f^h(i,j,a_r),a_{r'})=\dist(f^h(i,j,a_r),c_{r'})=2+10(n+1)$.
$\dist(f^h(i,j,a_r),b_{r'})=2+|P(\pi_i^h,a_{r'})|+|P(a_{r'},v_{r'}^n)|+|P(b_{r'},v_{r'}^n)|>2+10(n+1)$.
We see that any shortest path from $f^h(i,j,a_r)$ to a vertex of $\{y_1,y_2,z_1,z_2,w_1,w_2\}$ goes through $a_{r'},b_{r'}$ or $c_{r'}$.
Thus the distance from $f^h(i,j,a_r)$ to any vertex of $\{y_1,y_2,z_1,z_2,w_1,w_2\}$ is at least $2+10(n+1)$ and the equality holds only when $y_1=y_2=a_{r'}$ or $w_1=w_2=c_{r'}$.
Obviously $f(\pi_i^h,a_r)$ resolves the pairs $\{a_r,a_{r'}\}$ and $\{a_r,c_{r'}\}$ and $f(\pi_i^h,c_r)$ resolves the pairs $\{c_r,a_{r'}\}$ and $\{c_r,c_{r'}\}$ with $r\neq r'$.
As a result, every vertex pair of $\bigcup_{i\in [n]}S_i\times \bigcup_{r\in \{1,2,3\}}R_r$ is resolved by $f^h(i,j,a_r)$, $f(\pi_i^h,a_r)$ or $f(\pi_i^h,c_r)$.
This completes the proof for the lemma.
\end{proof}


\begin{lemma} \label{HxR}
Every pair $\{x,y\}\in \bigcup_{i\in [n]}H_i\times \bigcup_{r\in \{1,2,3\}}R_r$ is resolved by $S'$.
\end{lemma}
\begin{proof}
First we show that every pair $\{x,y\}\in P(s_i^j,a_r)\times (P(u_{r'}^{i'},a_{r'})\cup P(v_{r'}^{i'},a_{r'})\cup P(u_{r'}^{i'},b_{r'})\cup P(v_{r'}^{i'},b_{r'})\cup P(u_{r'}^{i'},c_{r'})\cup P(v_{r'}^{i'},c_{r'}))$ for $i,i'\in [n], j\in [m]$ and $r,r'\in \{1,2,3\}$ is resolved by $S'$.
We fix arbitrary integers $i,i'\in [n], h\in\{1,2\}, j\in [m]$ and $r,r'\in \{1,2,3\}$.
Suppose that $x\in P(s_i^j,a_r)$, $y_1\in P(u_{r'}^{i'},a_{r'})$, $y_2\in P(v_{r'}^{i'},a_{r'})$,
$z_1\in P(u_{r'}^{i'},b_{r'})$, $z_2\in P(v_{r'}^{i'},b_{r'})$, $w_1\in P(u_{r'}^{i'},c_{r'})$, $w_2\in P(v_{r'}^{i'},c_{r'})$.
For a vertex $x\in P(s_i^j,a_r)$, let $P(x,a_r)$ be the subpath of $P(s_i^j,a_r)$ from $a_r$ to $x$ and let $|P(x,a_r)|=\ell_x$.
For a vertex $y_1\in P(u_{r'}^{i'},a_{r'})$, let $P(y_1,u_{r'}^{i'})$ be the subpath of $P(u_{r'}^{i'},a_{r'})$ from $y_1$ to $u_{r'}^{i'}$ and let $|P(y_1,u_{r'}^{i'})|=\ell_{y_1}$.
For a vertex $y_2\in P(v_{r'}^{i'},a_{r'})$, let $P(y_2,v_{r'}^{i'})$ be the subpath of $P(v_{r'}^{i'},a_{r'})$ from $y_2$ to $v_{r'}^{i'}$ and let $|P(y_2,v_{r'}^{i'})|=\ell_{y_2}$.
For a vertex $z_1\in P(u_{r'}^{i'},b_{r'})$, let $P(z_1,u_{r'}^{i'})$ be the subpath of $P(u_{r'}^{i'},b_{r'})$ from $z_1$ to $u_{r'}^{i'}$ and let $|P(z_1,u_{r'}^{i'})|=\ell_{z_1}$.
For a vertex $z_2\in P(v_{r'}^{i'},b_{r'})$, let $P(z_2,v_{r'}^{i'})$ be the subpath of $P(v_{r'}^{i'},b_{r'})$ from $z_2$ to $v_{r'}^{i'}$ and let $|P(z_2,v_{r'}^{i'})|=\ell_{z_2}$.
For a vertex $w_1\in P(u_{r'}^{i'},c_{r'})$, let $P(w_1,u_{r'}^{i'})$ be the subpath of $P(u_{r'}^{i'},c_{r'})$ from $w_1$ to $u_{r'}^{i'}$ and let $|P(w_1,u_{r'}^{i'})|=\ell_{w_1}$.
For a vertex $w_2\in P(v_{r'}^{i'},c_{r'})$, let $P(w_2,v_{r'}^{i'})$ be the subpath of $P(v_{r'}^{i'},c_{r'})$ from $w_2$ to $v_{r'}^{i'}$ and let $|P(w_2,v_{r'}^{i'})|=\ell_{w_2}$.
Then $\dist(f(s_i^j,a_r),x)=\dist(f(\pi_i^h,a_r),x)-2=\ell_x$ if $x\neq a_r$ and $\dist(f(s_i^j,a_r),a_r)=\dist(f(\pi_i^h,a_r),a_r)=2$.
For a vertex pair $\{x,y_1\}$, there are two cases.
Case 1: $r'=r$.
$\dist(f(s_i^j,a_r),y_1)=\dist(f(\pi_i^h,a_r),y_1)=2+|P(u_{r'}^{i'},a_{r'})|-\ell_{y_1}$.
For a vertex pair $\{x,y_1\}$ that is not resolved by $f(s_i^j,a_r)$, $\dist(f(s_i^j,a_r),x)=\dist(f(s_i^j,a_r),y_1)=\dist(f(\pi_i^h,a_r),y_1)<\dist(f(\pi_i^h,a_r),x)$.
Thus in this case, every pair $\{x,y_1\}$ is resolved by $f(s_i^j,a_r)$ or $f(\pi_i^h,a_r)$.
Case 2: $r'\neq r$.
$\dist(f(s_i^j,a_r),y_1)=\dist(f(\pi_i^h,a_r),y_1)+2=2+|P(\pi_i^h,a_r)|+|P(\pi_i^h,a_{r'})|+|P(u_{r'}^{i'},a_{r'})|-\ell_{y_1}$.
For a vertex pair $\{x,y_1\}$ that is not resolved by $f(s_i^j,a_r)$, $\dist(f(\pi_i^h,a_r),x)>\dist(f(s_i^j,a_r),x)=\dist(f(s_i^j,a_r),y_1)>\dist(f(\pi_i^h,a_r),y_1)$.
Thus in this case, every pair $\{x,y_1\}$ is resolved by $f(s_i^j,a_r)$ or $f(\pi_i^h,a_r)$.
Similarly, every pair $\{x,y_2\}$ is resolved by $f(s_i^j,a_r)$ or $f(\pi_i^h,a_r)$.
For a vertex pair $\{x,z_1\}$, there are two cases.
Case 1: $r'=r$.
$\dist(f(s_i^j,a_r),z_1)=\dist(f(\pi_i^h,a_r),z_1)=\text{min }(2+|P(u_{r'}^{i'},a_{r'})|+\ell_{z_1},2+|P(u_{r'}^{n},a_{r'})|+|P(u_{r'}^{n},b_{r'})|+|P(u_{r'}^{i'},b_{r'})|-\ell_{z_1})$.
For a vertex pair $\{x,z_1\}$ that is not resolved by $f(s_i^j,a_r)$, $\dist(f(s_i^j,a_r),x)=\dist(f(s_i^j,a_r),z_1)=\dist(f(\pi_i^h,a_r),z_1)<\dist(f(\pi_i^h,a_r),x)$.
Thus in this case, every pair $\{x,z_1\}$ is resolved by $f(s_i^j,a_r)$ or $f(\pi_i^h,a_r)$.
Case 2: $r'\neq r$.
$\dist(f(\pi_i^h,a_r),b_{r'})=\text{min}_{\alpha\in [m]}(2+|P(s_i^{\alpha},a_r)|+|P(s_i^{\alpha},b_{r'})|)$.
Then $\dist(f(\pi_i^h,a_r),z_1)=\text{min }(|P(\pi_i^h,a_r)|+|P(\pi_i^h,a_{r'})|+|P(u_{r'}^{i'},a_{r'})|+\ell_{z_1},\dist(f(\pi_i^h,a_r),b_{r'})+|P(u_{r'}^{i'},b_{r'})|-\ell_{z_1})>30(n+1)>\dist(f(\pi_i^h,a_r),x)$.
Thus in this case, every pair $\{x,z_1\}$ is resolved by $f(\pi_i^h,a_r)$.
Similarly, every pair $\{x,z_2\}$ is resolved by $f(s_i^j,a_r)$ or $f(\pi_i^h,a_r)$.
For a vertex pair $\{x,w_1\}$, there are two cases.
Case 1: $r'=r$.
$\dist(f(s_i^j,a_r),w_1)=\text{min }(2+|P(u_r^{i'},a_{r'})|-2+\ell_{w_1},2+|P(\pi_i^h,a_r)|+|P(\pi_i^h,c_r)|+|P(u_r^{i'},c_r)|-\ell_{w_1})$ if $\ell_{w_1}\geq 2$.
$\dist(f(s_i^j,a_r),w_1)=2+|P(u_r^{i'},a_{r'})|+\ell_{w_1}$ if $\ell_{w_1}<2$.
$\dist(f(\pi_i^h,a_r),w_1)=\text{min }(2+|P(u_r^{i'},a_{r'})|-2+\ell_{w_1},|P(\pi_i^h,a_r)|+|P(\pi_i^h,c_r)|+|P(u_r^{i'},c_r)|-\ell_{w_1})$ if $\ell_{w_1}\geq 2$.
$\dist(f(\pi_i^h,a_r),w_1)=2+|P(u_r^{i'},a_{r'})|+\ell_{w_1}$ if $\ell_{w_1}<2$.
It follows that $\dist(f(s_i^j,a_r),w_1)\geq \dist(f(\pi_i^h,a_r),w_1)$.
For a pair $\{x,w_1\}$ that is not resolved by $f(s_i^j,a_r)$, $\dist(f(\pi_i^h,a_r),x)>\dist(f(s_i^j,a_r),x)=\dist(f(s_i^j,a_r),w_1)\geq \dist(f(\pi_i^h,a_r),w_1)$.
Thus in this case, every pair $\{x,w_1\}$ is resolved by $f(s_i^j,a_r)$ or $f(\pi_i^h,a_r)$.
Case 2: $r'\neq r$.
$\dist(f(s_i^j,a_r),w_1)=\text{min }(2+|P(\pi_i^h,a_r)|+|P(\pi_i^h,a_{r'})|+|P(u_{r'}^{i'},a_{r'})|-2+\ell_{w_1},2+|P(\pi_i^h,a_r)|+|P(\pi_i^h,c_{r'})|+|P(u_{r'}^{i'},c_{r'})|-\ell_{w_1})$ if $\ell_{w_1}\geq 2$.
$\dist(f(s_i^j,a_r),w_1)=2+|P(\pi_i^h,a_r)|+|P(\pi_i^h,a_{r'})|+|P(u_{r'}^{i'},a_{r'})|+\ell_{w_1}$ if $\ell_{w_1}<2$.
$\dist(f(\pi_i^h,a_r),w_1)=\text{min }(|P(\pi_i^h,a_r)|+|P(\pi_i^h,a_{r'})|+|P(u_{r'}^{i'},a_{r'})|+\ell_{w_1}-2,|P(\pi_i^h,a_r)|+|P(\pi_i^h,c_{r'})|+|P(u_{r'}^{i'},c_{r'})|-\ell_{w_1})$ if $\ell_{w_1}\geq 2$.
$\dist(f(\pi_i^h,a_r),w_1)=|P(\pi_i^h,a_r)|+|P(\pi_i^h,a_{r'})|+|P(u_{r'}^{i'},a_{r'})|+\ell_{w_1}$ if $\ell_{w_1}<2$.
It follows that $\dist(f(s_i^j,a_r),w_1)>\dist(f(\pi_i^h,a_r),w_1)$.
For a pair $\{x,w_1\}$ that is not resolved by $f(s_i^j,a_r)$, $\dist(f(\pi_i^h,a_r),x)>\dist(f(s_i^j,a_r),x)=\dist(f(s_i^j,a_r),w_1)>\dist(f(\pi_i^h,a_r),w_1)$.
Thus in this case, every pair $\{x,w_1\}$ is resolved by $f(s_i^j,a_r)$ or $f(\pi_i^h,a_r)$.
Similarly, every pair $\{x,w_2\}$ is resolved by $f(s_i^j,a_r)$ or $f(\pi_i^h,a_r)$.

Then we show that every pair $\{x,y\}\in P(s_i^j,c_r)\times (P(u_{r'}^{i'},a_{r'})\cup P(v_{r'}^{i'},a_{r'})\cup P(u_{r'}^{i'},b_{r'})\cup P(v_{r'}^{i'},b_{r'})\cup P(u_{r'}^{i'},c_{r'})\cup P(v_{r'}^{i'},c_{r'}))$ for $i,i'\in [n], j\in [m]$ and $r,r'\in \{1,2,3\}$ is resolved by $S'$.
We fix arbitrary integers $i,i'\in [n], h\in\{1,2\}, j\in [m]$ and $r,r'\in \{1,2,3\}$.
Suppose that $x\in P(s_i^j,c_r)$, $y_1\in P(u_{r'}^{i'},a_{r'})$, $y_2\in P(v_{r'}^{i'},a_{r'})$,
$z_1\in P(u_{r'}^{i'},b_{r'})$, $z_2\in P(v_{r'}^{i'},b_{r'})$, $w_1\in P(u_{r'}^{i'},c_{r'})$, $w_2\in P(v_{r'}^{i'},c_{r'})$.
We define $\ell_x,\ell_{y_1},\ell_{y_2},\ell_{z_1},\ell_{z_2},\ell_{w_1}$ and $\ell_{w_2}$ in a similar way to that of $\ell_x,\ell_{y_1}$ in the first paragraph.
Then $\dist(f(s_i^j,c_r),x)=\dist(f(\pi_i^h,c_r),x)-2=\ell_x$ if $x\neq c_r$ and $\dist(f(s_i^j,c_r),c_r)=\dist(f(\pi_i^h,c_r),c_r)=2$.
For a pair $\{x,y_1\}$, there are two cases.
Case 1: $r'=r$.
$\dist(f(s_i^j,c_r),y_1)=\text{min }(2+|P(u_r^{i'},c_{r'})|-2+\ell_{y_1},2+|P(\pi_i^h,c_r)|+|P(\pi_i^h,a_r)|+|P(u_r^{i'},a_r)|-\ell_{y_1})$ if $\ell_{y_1}\geq 2$.
$\dist(f(s_i^j,c_r),y_1)=2+|P(u_r^{i'},c_{r'})|+\ell_{y_1}$ if $\ell_{y_1}<2$.
$\dist(f(\pi_i^h,c_r),y_1)=\text{min }(2+|P(u_r^{i'},c_{r'})|-2+\ell_{y_1},|P(\pi_i^h,c_r)|+|P(\pi_i^h,a_r)|+|P(u_r^{i'},a_r)|-\ell_{y_1})$ if $\ell_{y_1}\geq 2$.
$\dist(f(\pi_i^h,c_r),y_1)=2+|P(u_r^{i'},c_{r'})|+\ell_{y_1}$ if $\ell_{y_1}<2$.
It follows that $\dist(f(s_i^j,c_r),y_1)\geq \dist(f(\pi_i^h,c_r),y_1)$.
For a pair $\{x,y_1\}$ that is not resolved by $f(s_i^j,c_r)$, $\dist(f(\pi_i^h,c_r),x)>\dist(f(s_i^j,c_r),x)=\dist(f(s_i^j,c_r),y_1)\geq \dist(f(\pi_i^h,c_r),y_1)$.
Thus in this case, every pair $\{x,y_1\}$ is resolved by $f(s_i^j,c_r)$ or $f(\pi_i^h,c_r)$.
Case 2: $r'\neq r$.
$\dist(f(s_i^j,c_r),y_1)=\dist(f(\pi_i^h,c_r),y_1)+2=2+|P(\pi_i^h,c_r)|+|P(\pi_i^h,a_{r'})|+|P(u_{r'}^{i'},a_{r'})|-\ell_{y_1}$.
It follows that $\dist(f(s_i^j,c_r),y_1)>\dist(f(\pi_i^h,c_r),y_1)$.
For a pair $\{x,y_1\}$ that is not resolved by $f(s_i^j,c_r)$, $\dist(f(\pi_i^h,c_r),x)>\dist(f(s_i^j,c_r),x)=\dist(f(s_i^j,c_r),y_1)>\dist(f(\pi_i^h,c_r),y_1)$.
Thus in this case, every pair $\{x,y_1\}$ is resolved by $f(s_i^j,c_r)$ or $f(\pi_i^h,c_r)$.
Similarly, every pair $\{x,y_2\}$ is resolved by $f(s_i^j,c_r)$ or $f(\pi_i^h,c_r)$.
For a pair $\{x,z_1\}$, there are two cases.
Case 1: $r'=r$.
Let $s_i^{j^*}$ be a vertex which resolves the pair $\{u_r^n,v_r^n\}$.
Then $|P(s_i^{j^*},c_r)|+|P(s_i^{j^*},b_r)|=40(n+1)-5n+1$.
$\dist(f(\pi_i^h,c_r),b_r)=2+|P(s_i^{j^*},c_r)|+|P(s_i^{j^*},b_r)|=40(n+1)-5n+3$.
$\dist(f(\pi_i^h,c_r),z_1)=\text{min }(2+|P(u_r^{i'},c_{r'})|+\ell_{z_1},\dist(f(\pi_i^h,c_r),b_r)+|P(u_r^{i'},b_r)|-\ell_{z_1})>20(n+1)>\dist(f(\pi_i^h,c_r),x)$.
Thus in this case, every pair $\{x,z_1\}$ is resolved by $f(\pi_i^h,c_r)$.
Case 2: $r'\neq r$.
$\dist(f(\pi_i^h,c_r),b_{r'})=\text{min}_{\alpha\in [m]}(2+|P(s_i^{\alpha},c_r)|+|P(s_i^{\alpha},b_{r'})|)$.
Then $\dist(f(\pi_i^h,c_r),z_1)=\text{min }(|P(\pi_i^h,c_r)|+|P(\pi_i^h,a_{r'})|+|P(u_{r'}^{i'},a_{r'})|+\ell_{z_1},\dist(f(\pi_i^h,c_r),b_{r'})+|P(u_{r'}^{i'},b_{r'})|-\ell_{z_1})>20(n+1)>\dist(f(\pi_i^h,c_r),x)$.
Thus in this case, every pair $\{x,z_1\}$ is resolved by $f(\pi_i^h,c_r)$.
Similarly, every pair $\{x,z_2\}$ is resolved by $f(\pi_i^h,c_r)$.
For a pair $\{x,w_1\}$, there are two cases.
Case 1: $r'=r$.
$\dist(f(s_i^j,c_r),w_1)=\dist(f(\pi_i^h,c_r),w_1)=2+|P(u_{r'}^{i'},c_{r'})|-\ell_{w_1}$.
For a vertex pair $\{x,w_1\}$ that is not resolved by $f(s_i^j,c_r)$, $\dist(f(s_i^j,c_r),x)=\dist(f(s_i^j,c_r),w_1)=\dist(f(\pi_i^h,c_r),w_1)<\dist(f(\pi_i^h,c_r),x)$.
Thus in this case, every pair $\{x,w_1\}$ is resolved by $f(s_i^j,c_r)$ or $f(\pi_i^h,c_r)$.
Case 2: $r'\neq r$.
$\dist(f(\pi_i^h,c_r),w_1)=\text{min }(|P(\pi_i^h,c_r)|+|P(\pi_i^h,c_{r'})|+|P(u_{r'}^{i'},c_{r'})|-\ell_{w_1},|P(\pi_i^h,c_r)|+|P(\pi_i^h,a_{r'})|+|P(u_{r'}^{i'},a_{r'})|+\ell_{w_1}-2)$ if $\ell_{w_1}\geq 2$.
$\dist(f(\pi_i^h,c_r),w_1)=|P(\pi_i^h,c_r)|+|P(\pi_i^h,a_{r'})|+|P(u_{r'}^{i'},a_{r'})|+\ell_{w_1}$ if $\ell_{w_1}<2$.
Thus $\dist(f(\pi_i^h,c_r),w_1)\geq 20(n+1)>\dist(f(\pi_i^h,c_r),x)$.
Similarly, every pair $\{x,w_2\}$ is resolved by $f(s_i^j,c_r)$ or $f(\pi_i^h,c_r)$.

Finally we show that every pair $\{x,y\}\in P(s_i^j,b_r)\times (P(u_{r'}^{i'},a_{r'})\cup P(v_{r'}^{i'},a_{r'})\cup P(u_{r'}^{i'},b_{r'})\cup P(v_{r'}^{i'},b_{r'})\cup P(u_{r'}^{i'},c_{r'})\cup P(v_{r'}^{i'},c_{r'}))$ for $i,i'\in [n], j\in [m]$ and $r,r'\in \{1,2,3\}$ is resolved by $S'$.
We fix arbitrary integers $i,i'\in [n], h\in\{1,2\},j\in [m]$ and $r,r'\in \{1,2,3\}$.
Suppose that $x\in P(s_i^j,b_r)$, $y_1\in P(u_{r'}^{i'},a_{r'})$, $y_2\in P(v_{r'}^{i'},a_{r'})$,
$z_1\in P(u_{r'}^{i'},b_{r'})$, $z_2\in P(v_{r'}^{i'},b_{r'})$, $w_1\in P(u_{r'}^{i'},c_{r'})$, $w_2\in P(v_{r'}^{i'},c_{r'})$.
We define $\ell_x,\ell_{y_1},\ell_{y_2},\ell_{z_1},\ell_{z_2},\ell_{w_1}$ and $\ell_{w_2}$ in a similar way to that of $\ell_x,\ell_{y_1}$ in the first paragraph.
For a pair $\{x,y_1\}$, there are two cases.
Case 1: $r=r'$.
$\dist(f(\pi_i^h,a_r),x)=\text{min }(2+|P(s_i^j,a_r)|+|P(s_i^j,b_r)|-\ell_x,2+|P(a_r,u_r^n)|+|P(b_r,u_r^n)|+\ell_x)>20(n+1)$.
$\dist(f(\pi_i^h,a_r),y_1)=|P(a_r,u_r^{i'})|-\ell_{y_1}<20(n+1)$.
Thus in this case, every pair $\{x,y_1\}$ is resolved by $f(\pi_i^h,a_r)$.
Similarly, every pair $\{x,y_2\}$ is resolved by $f(\pi_i^h,a_r)$.
Case 2: $r\neq r'$.
$\dist(f^1(u_{r'}^{i'},v_{r'}^{i'}),y_1)=1+\ell_{y_1}$ if $y_1\neq u_{r'}^{i'}$ and $\dist(f^1(u_{r'}^{i'},v_{r'}^{i'}),u_{r'}^{i'})=2$.
$\dist(f^1(u_{r'}^{i'},v_{r'}^{i'}),x)=\text{min }(\dist(f^1(u_{r'}^{i'},v_{r'}^{i'}),s_i^j)+|P(s_i^j,b_r)|-\ell_{x},\dist(f^1(u_{r'}^{i'},v_{r'}^{i'}),b_r)+\ell_{x})>20(n+1)>\dist(f^1(u_{r'}^{i'},v_{r'}^{i'}),y_1)$.
Thus in this case, every pair $\{x,y_1\}$ is resolved by $f^1(u_{r'}^{i'},v_{r'}^{i'})$.
Similarly, every pair $\{x,y_2\}$ is resolved by $f^1(u_{r'}^{i'},v_{r'}^{i'})$.
For a pair $\{x,z_1\}$, there are two cases.
Suppose that $P(s_i^j,b_r)=20(n+1)+5\lambda+1$ for some $\lambda\in [n]$.
Case 1: $r=r'$.
$\dist(f^{mid}(i,j,h),x)=|P^h(i,j,p_i^{3-h})|/2+2+|P(s_i^j,b_r)|-\ell_x=30(n+1)+5\lambda+3-\ell_x$.
$\dist(f^{mid}(i,j,h),z_1)=\text{min }(2+|P^h(i,j,p_i^{3-h})|/2+|P(s_i^j,b_r)|+|P(b_r,u_r^{i'})|-\ell_{z_1},1+|P^h(i,j,p_i^{3-h})|/2+|P(\pi_i^h,a_r)|+|P(a_r,u_r^{i'})|+\ell_{z_1})=\text{min }(50(n+1)+5\lambda+2-5i'-\ell_{z_1},40(n+1)+1-10i'+\ell_{z_1})$.
$\dist(f^{ecc}(i,j,h,r),x)=|P^h(i,j,a_r)|/2+1+|P(s_i^j,b_r)|-\ell_x=\dist(f^{mid}(i,j,h),x)-1$.
$\dist(f^{ecc}(i,j,h,r),z_1)=\text{min }(1+|P^h(i,j,a_r)|/2+|P(s_i^j,b_r)|+|P(b_r,u_r^{i'})|-\ell_{z_1},2+|P^h(i,j,a_r)|/2+|P(\pi_i^h,a_r)|+|P(a_r,u_r^{i'})|+\ell_{z_1})=\text{min }(50(n+1)+5\lambda+1-5i'-\ell_{z_1},40(n+1)+2-10i'+\ell_{z_1})$.
For a vertex pair $\{x,z_1\}$ that is not resolved by $f^{mid}(i,j,h)$,
$\dist(f^{mid}(i,j,h),x)=\dist(f^{mid}(i,j,h),z_1)=1+|P^h(i,j,p_i^{3-h})|/2+|P(\pi_i^h,a_r)|+|P(a_r,u_r^{i'})|+\ell_{z_1}>\dist(f^{ecc}(i,j,h,r),x)$.
If $\dist(f^{ecc}(i,j,h,r),z_1)=1+|P^h(i,j,a_r)|/2+|P(s_i^j,b_r)|+|P(b_r,u_r^{i'})|-\ell_{z_1}$, then obviously $f^{ecc}(i,j,h,r)$ resolves this pair.
Otherwise, $\dist(f^{ecc}(i,j,h,r),z_1)=40(n+1)+2-10i'+\ell_{z_1}>\dist(f^{mid}(i,j,h),z_1)>\dist(f^{ecc}(i,j,h,r),x)$.
It follows that every pair $\{x,z_1\}$ is resolved by $f^{mid}(i,j,h)$ or $f^{ecc}(i,j,h,r)$.
Similarly, every pair $\{x,z_2\}$ is resolved by $f^{mid}(i,j,h)$ or $f^{ecc}(i,j,h,r)$.
Case 2: $r\neq r'$.
$\dist(f^1(u_{r'}^{i'},v_{r'}^{i'}),z_1)=2+\ell_{z_1}<20(n+1)$.
$\dist(f^1(u_{r'}^{i'},v_{r'}^{i'}),x)=\text{min }(\dist(f^1(u_{r'}^{i'},v_{r'}^{i'}),s_i^j)+|P(s_i^j,b_r)|-\ell_{x},\dist(f^1(u_{r'}^{i'},v_{r'}^{i'}),b_r)+\ell_{x})>20(n+1)>\dist(f^1(u_{r'}^{i'},v_{r'}^{i'}),z_1)$.
Thus in this case, every pair $\{x,z_1\}$ is resolved by $f^1(u_{r'}^{i'},v_{r'}^{i'})$.
Similarly, every pair $\{x,z_2\}$ is resolved by $f^1(u_{r'}^{i'},v_{r'}^{i'})$.
For a pair $\{x,w_1\}$, there are two cases.
Case 1: $r=r'$.
$\dist(f(\pi_i^h,c_r),b_r)=\text{min}_{\alpha\in [m]}(2+|P(s_i^{\alpha},c_r)|+|P(s_i^{\alpha},b_r)|)=3+40(n+1)-5n>30(n+1)$.
$\dist(f(\pi_i^h,c_r),x)=\text{min }(2+|P(s_i^j,c_r)|+|P(s_i^j,b_r)|-\ell_x,\dist(f(\pi_i^h,c_r),b_r)+\ell_x)$.
$\dist(f(s_i^j,c_r),x)=\text{min }(|P(s_i^j,c_r)|+|P(s_i^j,b_r)|-\ell_x,\dist(f(\pi_i^h,c_r),b_r)+\ell_x)$.
$\dist(f(\pi_i^h,c_r),w_1)=\dist(f(s_i^j,c_r),w_1)=2+|P(c_r,u_{r'}^{i'})|-\ell_{w_1}=2+20(n+1)+10i'-\ell_{w_1}<30(n+1)$.
For a pair $\{x,w_1\}$ that is not resolved by $f(\pi_i^h,c_r)$, $\dist(f(\pi_i^h,c_r),x)=\dist(f(\pi_i^h,c_r),w_1)=2+|P(s_i^j,c_r)|+|P(s_i^j,b_r)|-\ell_x=\dist(f(s_i^j,c_r),w_1)>\dist(f(s_i^j,c_r),x)=|P(s_i^j,c_r)|+|P(s_i^j,b_r)|-\ell_x$.
Thus in this case, every pair $\{x,w_1\}$ is resolved by $f(\pi_i^h,c_r)$ or $f(s_i^j,c_r)$.
Case 2: $r\neq r'$.
$\dist(f^1(u_{r'}^{i'},v_{r'}^{i'}),w_1)=1+\ell_{w_1}$ if $w_1\neq u_{r'}^{i'}$ and $\dist(f^1(u_{r'}^{i'},v_{r'}^{i'}),u_{r'}^{i'})=2$.
$\dist(f^1(u_{r'}^{i'},v_{r'}^{i'}),x)=\text{min }(\dist(f^1(u_{r'}^{i'},v_{r'}^{i'}),s_i^j)+|P(s_i^j,b_r)|-\ell_{x},\dist(f^1(u_{r'}^{i'},v_{r'}^{i'}),b_r)+\ell_{x})>30(n+1)>\dist(f^1(u_{r'}^{i'},v_{r'}^{i'}),w_1)$.
Thus in this case, every pair $\{x,w_1\}$ is resolved by $f^1(u_{r'}^{i'},v_{r'}^{i'})$.
Similarly, every pair $\{x,w_2\}$ is resolved by $f^1(u_{r'}^{i'},v_{r'}^{i'})$.
This completes the proof for the lemma.
\end{proof}

\begin{lemma}  \label{FxG}
For any vertex $v_f\in {\cal F}$, every vertex pair $\{x,y\}\in \{v_f\}\times V(G')\setminus \{v_f\}$ is resolved by $S'$.
\end{lemma}
\begin{proof}
Without loss of generality, suppose that $v_1,v_2,v_c\in F^1(u_r^i,v_r^i)$ for some $r\in \{1,2,3\},i\in [n]$, where $v_c$ is the connecting vertex of $F^1(u_r^i,v_r^i)$, $v_1,v_2$ are the false twins and $v_1\in S'$. Then obviously every vertex pair of $\{v_1\}\times V(G')\setminus \{v_1\}$ is resolved by $v_1$. Every vertex pair of $\{v_2\}\times V(G')\setminus \{v_2\}$ is resolved $v_1$ except the vertex pair $\{v_2,v_c\}$. Let $w_f$ be an arbitrary vertex of $S'\setminus F^1(u_r^i,v_r^i)$. Then there is a shortest path from $w_f$ to $v_2$ going through $v_c$. Thus $\dist(w_f,v_2)=\dist(w_f,v_c)+1$ and $\{v_2,v_c\}$ is resolved by $w_f$. For any vertex $u\in V(G')\setminus F^1(u_r^i,v_r^i)$, $\dist(v_1,u)>\dist(v_1,v_c)=1$. Then the correctness of the lemma follows.
\end{proof}

With Lemmas~\ref{UxU}-~\ref{FxG}, we show that every pair of distinct vertices of $G'$ is resolved by some vertex of $S'$.
It follows that Lemma~\ref{completeness} is true and this proves the completeness of the reduction.

Finally, with Lemmas~\ref{multiRSnp},~\ref{soundness},~\ref{completeness} and~\ref{pathwidth} in hand, we can prove the correctness of Theorem~\ref{thm:main}.

\bibliography{refs}
\end{document}